\documentclass[sigconf,nonacm,screen]{acmart}

\setcitestyle{nosort}
\AtEndPreamble{%
    \theoremstyle{acmdefinition}
    }

\usepackage{mathrsfs}
\usepackage{subcaption}
\usepackage{tikz}
\usepackage{thm-restate}
\usepackage[shortlabels]{enumitem}

\usepackage{natbib}
\setcitestyle{numbers,square}
\usepackage{hyperref}

\AtEndPreamble{%
  \theoremstyle{acmplain}
  \newtheorem{theorem*}{Theorem}
  \newtheorem{conjecture*}{Conjecture}}

\makeatletter
\@ifpackageloaded{newtxmath}{}{\usepackage{amssymb}}
\makeatother

\numberwithin{equation}{section}

\let\rel\mathbf                 %
\let\clo\mathscr                %
\let\equals\approx              %
\let\fro\leftarrow              %
\let\bind\sharp
\let\tensor\otimes
\let\incl\hookrightarrow
\let\iff\Leftrightarrow
\let\rcirc\circ
\newcommand{\False}{\bot}     %

\DeclareMathOperator{\ar}{ar}

\DeclareMathOperator{\im}{im}

\DeclareMathOperator{\supp}{supp}

\DeclareMathOperator{\CSP}{CSP}
\DeclareMathOperator{\PCSP}{PCSP}
\DeclareMathOperator{\PMC}{PMC}
\DeclareMathOperator{\Pol}{Pol}
\DeclareMathOperator{\SA}{SA}
\DeclareMathOperator{\FPC}{\textsf{FPC}}
\DeclareMathOperator{\FPR}{\textsf{FPR}}
\DeclareMathOperator{\Free}{{\rel F}}
\newcommand{\leqsub}[1]{\mathrel{\leq_{#1}}}
\newcommand{\leqdl}{\leqsub{\textsf{\normalshape DL}}}
\newcommand{\leqcons}[1][k]{\leqsub{#1\text{\normalshape -cons}}}
\newcommand{\leqarc}{\leqsub{\textsf{\normalshape AC}}}
\newcommand{\ac}{\kappa_{\text{arc}}}
\newcommand{\conv}{{\text{conv}}}

\newcommand{\NP}{\textsf{NP}}
\newcommand{\Ptime}{\textsf{P}}

\newcommand{\Datalog}{Datalog\texorpdfstring{$^\cup$}{} }

\newcounter{commentno}

\copyrightyear{2024}
\setcopyright{cc}

\title{Local consistency as a reduction between constraint satisfaction problems}

\author{Víctor Dalmau}
\email{victor.dalmau@upf.edu}
\affiliation{%
  \institution{Universitat Pompeu Fabra}
  \city{Barcelona}
  \country{Spain}
}

\author{Jakub Opršal}
\email{j.oprsal@bham.ac.uk}
\affiliation{%
  \institution{University of Birmingham}
  \city{Birmingham}
  \country{UK}
}

\ccsdesc[500]{Theory of computation~Constraint and logic programming}
\ccsdesc[500]{Theory of computation~Problems, reductions and completeness}

\keywords{constraint satisfaction problem, Datalog, Karp reduction, polymorphism}

\begin{abstract}
  We study the use of local consistency methods as reductions between constraint satisfaction problems (CSPs), and promise version thereof, with the aim to classify these reductions in a similar way as the algebraic approach classifies gadget reductions between CSPs. This research is motivated by the requirement of more expressive reductions in the scope of promise CSPs. While gadget reductions  are enough to provide all necessary hardness in the scope of (finite domain) non-promise CSP, in promise CSPs a wider class of reductions needs to be used.

  We provide a general framework of reductions, which we call \emph{consistency reductions}, that covers most (if not all) reductions recently used for proving NP-hardness of promise CSPs. We prove some basic properties of these reductions, and provide the first steps towards understanding the power of consistency reductions by characterizing a fragment associated to arc-consistency in terms of polymorphisms of the template. In addition to showing hardness, consistency reductions can also be used to provide feasible algorithms by reducing to a fixed tractable (promise) CSP, for example, to solving systems of affine equations. In this direction, among other results, we describe the well-known Sherali--Adams hierarchy for CSP in terms of a consistency reduction to linear programming.
\end{abstract}

\thanks{This is a full version of a paper published at LICS 2024 \cite{DalmauO24}.\par
This project has received funding from the European Research Council (ERC) under the European Union's Horizon 2020 research and innovation programme (grant agreement No 714532).
This project has received funding from the European Union’s Horizon 2020 research and innovation programme under the Marie Skłodowska-Curie Grant Agreement No 101034413.
Jakub Opršal was also supported by the UK EPSRC grant EP/R034516/1.
Victor Dalmau was supported by the MICIN under grants PID2019-109137GB-C22 and PID2022-138506NB-C22, and the Maria de Maeztu program (CEX2021-001195-M)}

\begin{document}

\maketitle

\begin{acks}
  A part of this paper is based on an unpublished note by Marcin Wrochna \cite{Wro22} which contains the characterisation of the arc-con\-sis\-tency reduction and several observations and statements that served as a ground for research presented in this paper.

  The second author would also like to thank Anuj Dawar, Joanna Ochremiak, Adam Ó Conghaile, Antoine Mottet, and Manuel Bodirsky for inspiring discussions.
  We would also like to thank the reviewers of this paper for careful reading and their insightful remarks, and all who read a previous version of this paper and provided useful feedback.
\end{acks}

\section{Introduction}
  \label{sec:introduction}

Is it possible to find a mathematical invariant that characterises when one computational problem reduces to another by a po\-ly\-no\-mial-time \emph{(Karp) reduction}? This question is far beyond our current understanding as it would imply a characterisation of polynomial-time solvable problems, possibly resolving the \Ptime\@ vs.~\NP\@ problem. The present paper is a contribution to the effort to provide a partial answer to this question by characterising certain \emph{well-structured efficient reductions} between structured problems.

A characterisation of a subclass of log-space reductions between \emph{constraint satisfaction problems (CSPs)} is the core of a theory referred to as \emph{the algebraic approach} to the CSP.
The development of this theory started with the work of Jeavons, Cohen, and Gyssens \cite{JCG97}; Bulatov, Jeavons, and Krokhin \cite{BJK05}, and after 20 years of active research provided the celebrated \emph{dichotomy theorem} proven by Bulatov \cite{Bul17} and Zhuk \cite{Zhu20} after being conjectured by Feder and Vardi \cite{FV98}.
The scope of the theory and the dichotomy theorem are finite-template CSPs. Finite-template CSPs can be formulated as deciding whether there is a homomorphism from a given structure to a fixed finite (relational) structure (called \emph{template}), or alternatively as deciding whether the template satisfies a primitive-positive sentence; the atoms occurring in this sentence are called \emph{constraints}.
By varying the template one can encode a large family of computational problems, including SAT (and its variants), graph $k$-colouring, or solving systems of linear equations over a finite field. Many of these problems naturally arise in various settings.
A general theorem \cite{BJK05} characterises when one such problem can be reduced to another using a simple \emph{gadget reduction} by the means of the \emph{polymorphisms} of the template (see also Theorem~\ref{thm:gadget-characterisation} below).
In the context of CSPs and this paper, gadget reductions are those reductions that replace each constraint of the input instance with a `gadget' consisting of several constraints over the output language (see Section~\ref{sec:gadgets}). This is a special case of a more general concept in complexity theory where one only requires that every bit of the output depends on a bounded number of bits of the input.
The dichotomy theorem can be formulated as follows: For each finite template, either its CSP can be proven to be \NP-complete via a gadget reduction from graph $3$-colouring \cite{BJK05}, or there is a polynomial time algorithm solving it \cite{Bul17,Zhu20}.

The motivation of the present paper stems from two failures of gadget reductions.
Firstly, their inability to sufficiently explain the tractability side of the dichotomy; we would like to see that all tractable CSPs (assuming $\Ptime \neq \NP$) reduce to a small class of problems that are tractable for an obvious reason (e.g., solving systems of linear equations over finite fields). This is not the case for gadget reductions, and a lot of effort has been put into understanding the hardest (w.r.t.\ gadget reductions) tractable CSPs (see, e.g., Barto, Brady, Bulatov, Kozik, and Zhuk \cite{BartoBBKZ21}, or Barto, Bodor, Kozik, Mottet, and Pinsker \cite{BartoBKMP23}). In this paper, we suggest an alternative approach that might make this work unnecessary.

Secondly, gadget reductions fail to explain \NP-hardness in a slightly more general setting of \emph{promise constraint satisfaction problems}.
In the promise version of the problem, each instance comes with two versions of each constraint, a \emph{stronger} and a \emph{weaker} one. The goal is then decide between two (disjoint, but not complementary) cases: all stronger constraints can be satisfied, or not even the weaker constraints can be satisfied. A prominent problem described in this scope is \emph{approximate graph colouring} which asks, e.g., to decide between graphs that are 3-colourable and those that are not even 6-colourable.
More formally, a template of a promise CSP consists of a pair of structures, such that the first structure maps homomorphically to the second. The goal is to decide, given as input a third structure, whether the input maps homomorphically to the first structure of the template, or it does not map homomorphically even to the second.
A systematic study of promise CSPs has been started by Austrin, Guruswami, and Håstad \cite{AGH17} where a certain promise version of SAT was considered, and continued in a series of papers including generalisations of the algebraic approach by Brakensiek and Guruswami \cite{BG21} and Barto, Bulín, Krokhin, and Opršal \cite{BBKO21}.
In particular, \cite{BBKO21} characterises gadget reductions in terms of (generalised) \emph{polymorphisms} in a similar way as \cite{BJK05} characterises gadget reductions in the scope of (classic) CSPs (see Theorem~\ref{thm:gadget-characterisation} for a formal statement).
Many of the \NP-hardness results in the promise setting \cite[e.g.,][]{DRS05,Hua13,AGH17,KO19,FKOS19,WZ20,BWZ21,BBB21,NZ22} cannot be adequately explained by gadget reductions. These hardness results usually rely on variations of the PCP theorem of Arora and Safra \cite{AS98} which is better suited for a different, more analytical, approximation version of CSPs. Indeed, new reductions for promise CSPs have been called for by Barto, Bulín, Krokhin, and Opršal \cite{BBKO21}, and by Barto and Kozik \cite{BK22}.
The goal is to find a class of reductions that would explain all of these hardness results, and subsequently extend the hardness to other interesting promise CSPs, e.g., the approximate graph colouring.

\subsection*{Our contributions}

The primary objective of the present paper is to offer a novel perspective on the complexity landscape of CSPs and promise CSPs that is focused on reductions rather than algorithms or algebraic hardness conditions.
We aim to identify a class of efficient reductions that addresses both the tractability side of CSPs and the hardness side of promise CSPs.
To this end, the suggested class of reductions should have the following properties:
\begin{itemize}
  \item \emph{robustness}, informally meaning that the notion of reduction remains invariant under minor variations in its definition.
  \item \emph{expressiveness}, that is, it should significantly extend the capabilities of gadget reductions.
  \item \emph{characterisability}, meaning that it must be possible to characterise, in some terms similar to the algebraic characterisation of gadget reductions, whether a given (promise) CSP reduces to another given (promise) CSP.
\end{itemize}

We introduce a new general class of reductions between (promise) CSPs, referred to as \emph{$k$-consistency reductions} or \emph{\Datalog reductions}, and supplement it with several results that showcase the robustness and expressiveness of this family of reductions. Furthermore, we take initial steps towards its characterisability. A fully fledged theory for $k$-consistency reductions would have several overreaching consequences and, hence, it falls beyond the scope of the present paper, which aims solely to initiate it.

\begin{figure*}[t]
  \begin{subfigure}{.33\textwidth}
    \[
      \includegraphics[height=1.75in]{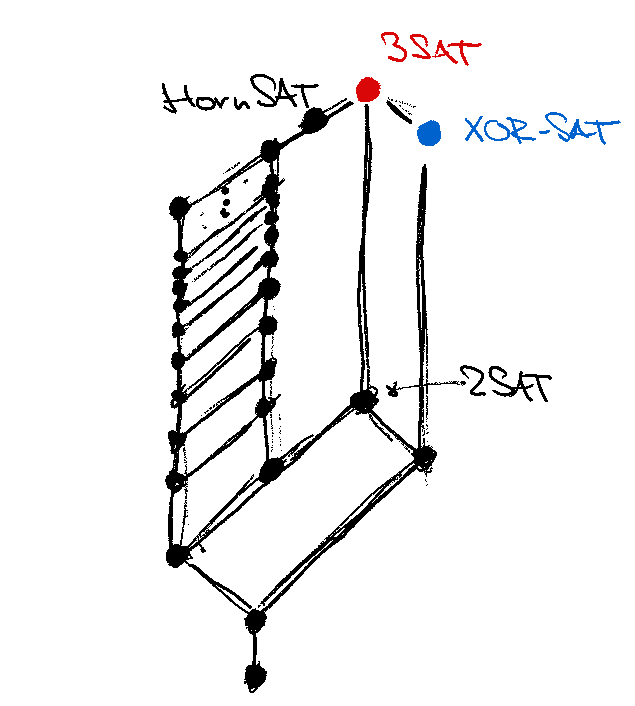}
    \]
    \caption{gadget reductions}
    \label{fig:bv20}
  \end{subfigure}%
  \begin{subfigure}{.33\textwidth}
    \[
      \includegraphics[height=1.75in]{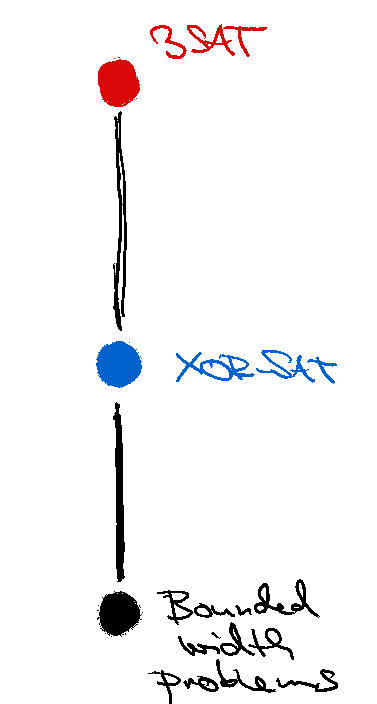}
    \]
    \caption{\Datalog reductions}
    \label{fig:boolean-datalog}
  \end{subfigure}%
  \begin{subfigure}{.33\textwidth}
    \[
      \includegraphics[height=1.4in]{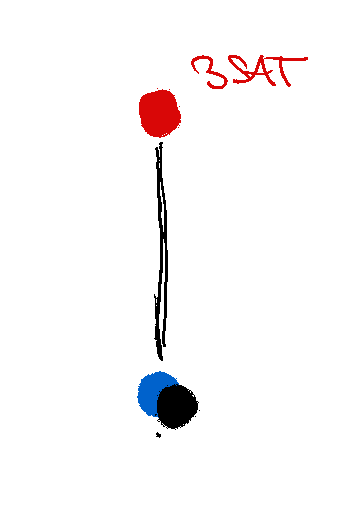}
    \]
    \caption{polynomial-time reductions}
    \label{fig:dichotomy}
  \end{subfigure}
  \caption{Boolean CSPs ordered by three classes of reductions with examples of problems belonging to some classes.}
  \label{fig:boolean}
\end{figure*}

\subsubsection*{Robustness}

There are two key properties of a robust class of reductions in the scope of (promise) CSPs: \emph{monotonicity} with respect to the homomorphism preorder, which we require in order to preserve the structure of the problem, and \emph{closure under composition}, which is a natural requirement that is essential in order to chain reductions together. Most of the reductions used in the scope of promise CSPs are monotone, but not all compose; in particular, the reductions provided by Barto and Kozik \cite{BK22} do not compose (see Example~\ref{ex:bk-do-not-compose}).

First, we describe our suggested class of reductions as logical interpretations in a monotone fragment of the fixed-point logic, namely in the logic of Datalog programs endowed with an operator that allows to take a disjoint union of relations.
We call them \emph{\Datalog reductions}.

It is relatively easy to show that this reduction has the first two properties. %
In particular, we show that \Datalog reductions compose (Theorem~\ref{thm:ddatalog-composes}).
Moreover, we show that \Datalog reductions encapsulate gadget reductions, which is a generalisation of a result of Atserias, Bulatov, and Dawar \cite[Theorem 5]{ABD09} to the promise setting, although we use disjoint unions instead of parameters.

\begin{theorem*}[Theorem~\ref{thm:gadget-is-ddatalog} informally]
  \label{thm:main-2}
  Every gadget reduction is expressible as a \Datalog reduction up to homomorphic equivalence.
\end{theorem*}

Secondly, we provide a combinatorial counterpart of \Datalog reductions. This alternative construction has two ingredients: gadget reductions and the local consistency algorithm for (promise) CSP. Local consistency is ubiquitous in constraint programming (see for example \cite{RBW06}). Intuitively, the consistency algorithm (with a parameter $k$) iteratively adds all newly implied constraints that can be derived by simultaneously considering $k$ variables.  From a theoretical standpoint, it has been more common to consider a decision version (which we call \emph{consistency test}) that merely outputs `no' if the set of derived constraints contains an unsatisfiable constraint and `yes' otherwise. More specifically, several papers have investigated which CSPs are solvable by such algorithms, leading to a complete characterisation in the non-promise case by Barto and Kozik \cite{BK14}. Datalog programs, too, can be used to solve (promise) CSPs by turning them into a decision algorithm choosing a nullary goal predicate. It was soon established by Kolaitis and Vardi \cite{KV00} that the consistency test and Datalog programs have the same power in solving non-promise CSPs. This equivalence extends to the promise setting (set $\mathscr A$ to be the union of positive and negative instances, and $B$ to be the first structure of the template in the statement of \cite[Theorem 4.8]{KV00}).

Getting back to reductions, we refer to the combination of a gadget reduction and the local consistency algorithm with parameter $k$ as the \emph{$k$-consistency reduction}. We show that, indeed, consistency reductions and Datalog reductions have the same power.

\begin{theorem*}[Theorem~\ref{thm:canonical-width} informally] \label{thm:main-3}
  A promise CSP reduces to a second promise CSP via a \Datalog reduction if and only if it reduces via the $k$-consistency reduction.
\end{theorem*}

Alternatively, the preceding theorem can be viewed as an extension of the well-known equivalence between the decision versions of consistency and Datalog to the realm of reductions. The proof, however, faces two primary technical challenges. Firstly, we are dealing here with the full-fledged version (instead of the Boolean version) of consistency and Datalog. And, more importantly, it is additionally necessary to understand how Datalog interpretations interact with gadget reductions.%

\subsubsection*{Expressiveness}

We present several examples illustrating that $k$-consistency reductions effectively tackle, at least to some extent, the aforementioned shortcomings of gadget reductions.

Firstly, it is conceivable and consistent with our current understanding of the complexity of promise CSPs, that a promise CSP is \NP-hard if and only if it allows a $k$-consistency reduction from an \NP-hard CSP, e.g., graph $3$-colouring.
For example, it is relatively easy to check that the reductions of Barto and Kozik \cite{BK22} are covered by $k$-consistency reductions, and consequently many results about \NP-hardness of promise CSPs can be explained by a $k$-consistency reduction including the hardness of $5$ colouring $3$-colourable graphs \cite{BBKO21} and the \NP-hardness in \cite{DRS05,AGH17,KO19,FKOS19,BBB21,NZ22}.
Furthermore, the reduction used by Wrochna and Živný \cite{WZ20} to provide \NP-hardness of colouring $k$-colourable graphs with $\binom k{\lfloor k/2\rfloor} - 1$ colours is a \Datalog reduction from approximate graph colouring which is \NP-hard by a previous result of Huang \cite{Hua13}. The result of Huang is one of the few exceptions; we currently do not know whether it can be also explained by a $k$-consistency reduction from graph $3$-colouring.

Second, several well-known hierarchies of relaxation algorithms for promise CSPs fall within our framework. A comprehensive look at hierarchies is given in Section~\ref{sec:hierarchies} where we describe several hierarchies in terms of $k$-consistency reductions to a fixed problem. This includes the \emph{bounded width hierarchy} (i.e., the hierarchy of promise CSPs expressible by Datalog), and the \emph{Sherali--Adams hierarchy} of linear programming for promise CSPs.
Furthermore, we show that the Sherali--Adams hierarchy of (promise) CSPs coincides with the hierarchy of (promise) CSPs reducible to linear programming by a $k$-consistency reduction (Theorem~\ref{thm:sa}).
A different framework to study hierarchies of relaxations of algorithms  was recently introduced by Ciardo and Živný \cite{CZ23-hierarchies}. The Ciardo-Živný hierarchy is always subsumed by the hierarchy of $k$-consistency reductions to a fixed problem, and in some cases, e.g., the Sherali--Adams hierarchy, both hierarchies align perfectly.

Since consistency reductions contain, as a particular case, gadget reductions, it follows that any class of (promise) CSPs closed under consistency reductions can be, at least theoretically, described in terms of polymorphisms. This is a very desirable feature as polymorphisms lie at the heart of the algebraic approach. As a direct consequence of our findings it follows that both the class of promise CSPs solvable by the consistency test and the class of promise CSPs solvable by some level of the Sherali--Adams relaxation are closed under gadget reductions. While the former result had previously been established \cite[Lemma 7.5]{BBKO21}, the latter result, in its full generality, is new. Prior to our work it had only been known for the particular case of non-promise CSPs where it can be derived by combining the results of Thapper and Živný \cite{TZ17} and Barto and Kozik \cite{BK14}.

Finally, we get to the tractability side of CSPs.
In the Boolean case (i.e., for templates with two elements) it is known that every tractable CSP is reducible via a gadget reduction to one of three fundamental tractable problems (HornSAT, 2SAT, and XOR-SAT) \cite{Schaefer78,BV20}. If we replace gadget reductions by \Datalog reductions then every tractable Boolean CSP can be reduced to a single problem, namely, solving systems of linear equations over $\mathbb Z_2$ (XOR-SAT).
See also Figure~\ref{fig:boolean} comparing gadget, Datalog$^\cup$, and general polynomial-time reductions on Boolean CSPs.\footnote{Fig.~\ref{fig:bv20} is due to Bodirsky and Vucaj \cite{BV20}; and Fig.~\ref{fig:dichotomy} is due to Bulatov \cite{Bul17} and Zhuk \cite{Zhu20} and assumes $\Ptime \neq \NP$.}

However, a finite-template tractable CSP does not necessarily reduce to XOR-SAT via a \Datalog reduction (when template is allowed to have more than two elements). Yet, this seems to be caused by the fact that that the problem has the ability to `express' equations modulo a prime different from $2$. This leads us to believe that a finite-template tractable CSP can perhaps reduce via a \Datalog reduction to solving systems of linear equations over $\mathbb Z_n$ for some integer $n$ (which is divisible by all the primes $p$ such that equations mod $p$ are `expressible' by the problem). Given that, in turn, this problem reduces via a \Datalog reduction to solving systems of linear equations over integers, we put forth the following conjecture, supported by evidence presented in Section~\ref{sec:hierarchies}.

\begin{conjecture*}[Conjecture~\ref{the-conjecture}]
  Every finite template CSP either allows a \Datalog reduction from 3-colouring (and hence it is \NP-hard), or is reducible to solving systems of affine equations over $\mathbb Z$ via a \Datalog reduction (and hence it is in \Ptime).
\end{conjecture*}

The CSPs that are conjectured to be \NP-hard in the above conjecture include all CSPs hard under the Bulatov--Zhuk dichotomy (since gadget reductions are subsumed by \Datalog reductions due to Theorem~\ref{thm:main-2}). Naturally, this implies that the two classes coincide unless $\Ptime = \NP$. A direct proof (without assuming $\Ptime \neq \NP$) that every CSP allowing a \Datalog reduction from 3-colouring also allows a gadget reduction from 3-colouring is not known.
We have decided to express the dichotomy using \Datalog reductions instead of gadget reductions because this formulation generalizes better to promise CSPs:
It is consistent with our knowledge that the hardness condition given in the conjecture covers all \NP-complete finite-template promise CSPs, while there are \NP-hard promise CSPs (e.g., 5-colouring 3-colourable graphs) whose hardness can be shown via a \Datalog reduction but not via a gadget reduction from 3-colouring.
Likewise, the tractability side of our conjecture gives an algorithm which generalises in a straight-forward way to promise CSPs unlike the algorithms of Bulatov and Zhuk.

\subsubsection*{Characterisability}

We describe a characterisation of a special case of $k$-consistency reduction, referred to as \emph{arc-consistency reduction} (which is obtained by replacing the $k$-consistency algorithm in the $k$-consistency reduction by a weaker \emph{arc-consistency}), in terms of a certain transformation $\omega$ (see Definition~\ref{def:omega}) of \emph{polymorphisms} of the templates (see, e.g., \cite[Section 2.2]{BBKO21} for the necessary standard definitions).

\begin{theorem*} [Theorem~\ref{thm:unary} informally]
  A promise CSP with polymorphisms $\clo A$ reduces to a promise CSP with polymorphisms $\clo B$ via the arc-consistency reduction if and only if there is a `weak' minion homomorphism $\omega(\clo B) \to \clo A$.
\end{theorem*}

In comparison, the above-mentioned theorem of Barto, Bulín, Krokhin, and Opršal \cite{BBKO21} claims that there is a gadget reduction between the two problems if and only if there is a minion homomorphism $\clo B \to \clo A$.

We note that a different characterisation of the arc-consistency reduction was recently provided by Mottet \cite[Theorem 12]{Mottet24}.

\subsection*{Organisation of the paper}

We present basic definitions and notation in Section~\ref{sec:preliminaries}. This is followed by the definition of our framework and formulation of the basic theorems of our theory in Section~\ref{sec:local-reductions}. Section~\ref{sec:hierarchies} describes several hierarchies in terms of $k$-consistency reductions, and discusses evidence for our conjecture about tractability of finite-template CSPs. In Section~\ref{sec:arc-consistency}, we describe the characterisation of the arc-consistency reduction. And finally, in Section~\ref{sec:conclusion}, we outline a few directions for future research, and possible applications of the theory developed in the present paper.

The appendices contain detailed proofs of claims made in the main part of the paper, further results, and several examples establishing wider context of consistency reductions. We present appendices in slightly different order from the sections of the main part to account for interdependencies. Appendix~\ref{app:preliminaries} contains some definitions omitted from the main part, and Appendices~\ref{app:proofs-iii}, \ref{app:proofs-v}, and \ref{app:proofs-iv} contain extended materials for Sections~\ref{sec:local-reductions}, \ref{sec:arc-consistency}, and \ref{sec:hierarchies}, respectively.

\section{Preliminaries}
  \label{sec:preliminaries}

Many of the constructions described in this paper can be defined using several different languages, e.g., an instance of a CSP can be viewed as a list of constraints, a logical formula, or a relational structure. In this paper, the structural perspective takes precedence.

Sets are generally denoted by capital letters, and integers by lower-case letters. We often do not explicitly specify that such symbols represent sets, or integers when it is clear from the context, e.g., when the symbol appears in an index.
We denote by $[n]$ the set of the first $n$ positive integers, and by $A^X$ the set of all functions $f\colon X \to A$. We use the notation $\pi^{-1}(y)$ for the set of preimages of $y$ under $\pi$, i.e., the set $\{ x\in X \mid \pi(x) = y\}$.  We will denote entries in a tuple $a$ by lower indices, i.e., $a = (a_1, \dots, a_k)$, and occasionally view a tuple $a\in A^k$ as a function $a\colon [k] \to A$.
The identity map on $X$ is denoted by $1_X$.

\subsection{Constraint satisfaction problems}
  \label{sec:csp-prelim}

To settle on the playing field, we formally define the constraint satisfaction problem (CSP), and its promise variant. We refer to Barto, Krokhin, and Willard \cite{BKW17} for a deeper exposition of the algebraic theory of CSPs, and to Krokhin and Opršal \cite{KO22} for more background and examples of promise CSPs. We use the standard notation and definitions for promise CSPs as in Barto, Bulín, Krokhin, and Opršal \cite{BBKO21}.

The \emph{uniform CSP} is a decision problem whose instance consists of variables $v, w, \dots$ each of which is allowed to attain values from some domain and constraints, each of which involves several variables. The goal is to decide whether there is an assignment of values to variables that simultaneously satisfies all constraints.
In the present paper, we are concerned with fixed-template CSP which restricts the possible constraints in the input. This is formalised as a homomorphism problem between relational structures where the target structure is fixed.
To allow for different domains for distinct variables we allow structures to be typed (multisorted).
The multisorted framework does not increase the complexity of the proofs, merely the complexity of notation. Moreover, as noted in \cite[Section B.4]{BK22}, it provides a natural setting for the theory since it places \emph{(layered) label cover problem}, whose promise version is used as an intermediate step in reductions (both in this paper and in the algebraic characterisation of gadget reductions \cite[Theorem 3.12]{BBKO21}), into the framework of CSPs. In addition, several of the key concepts introduced here (e.g., the extension of Datalog with disjoint unions) are more naturally expressed in the multisorted setting.

\begin{definition} \label{def:structure}
  A \emph{(multisorted) relational signature} consists of a set of types (sorts) $\{t, s, \dots \}$ and a set of relational symbols $\{R, S, \dots\}$, where each symbol has a fixed \emph{arity} which is formally a tuple of types. We denote this tuple by $\ar_R, \ar_S, \dots$, and often say simply that $R$ is of arity $k$ where $k$ is the length of $\ar_R$.

  A \emph{relational structure} $\rel A$ of such a signature is a tuple
  \[
    \rel A = (A_t, A_s, \dots; R^{\rel A}, S^{\rel A}, \dots)
  \]
  where $A_t$, $A_s$, etc.\ are sets and $R^{\rel A}$, $S^{\rel A}$, etc.\ are relations on these sets of the corresponding arity, e.g., $R^{\rel A} \subseteq A_{\ar_R(1)} \times \dots \times A_{\ar_R(k)}$ assuming that $\ar_R$ has length $k$.
\end{definition}

We will usually assume that both the signature and the structures are finite, i.e., the signature contains only finitely many types and symbols, and all of the sets $A_t$, etc.\ are finite.
We call any $a\in A_t$ an \emph{element of $\rel A$ of type $t$} (we use sort and type interchangeably), and we do not consider two elements of $\rel A$ to be equal unless they have the same type. For example, if $\rel A$ has two domains $A_t = A_s = \{0, 1\}$, we will view it as a structure with four elements. We will write $x\in \rel A$ for `$x$ is an element of $\rel A$'.
All structures in this paper are typed relational structures, and signature of such structures is given implicitly.

A \emph{homomorphism} between two structures is a mapping that maps elements of one structure to elements of the second structure that preserves types and all relations. A necessary requirement for a homomorphism to exist is that both structures have the same signature.
Formally, a \emph{homomorphism} from a structure $\rel A$ to a structure $\rel B$ is a collection $f$ of maps $f_t\colon A_t \to B_t$, one for each type $t$, such that for each relational symbol $R$ and $(a_1, \dots, a_k) \in R^{\rel A}$, we have
\[
  (f_{\ar_R(1)}(a_1), \dots, f_{\ar_R(k)}(a_k)) \in R^{\rel B}.
\]
We write $f(a)$ instead of $f_t(a)$ if the type $t$ of $a$ is clear from the context, we write $f\colon \rel A \to \rel B$ if $f$ is a homomorphism, and $\rel A \to \rel B$ if such a homomorphism exists.

The \emph{(typed) constraint satisfaction problem with template $\rel A$} is a decision problem whose goal is: given a structure $\rel X$, decide whether $\rel X \to \rel A$. We denote this problem by $\CSP(\rel A)$.
Before we define promise CSP, recall that a promise problem is a generalisation of a decision problem where the \emph{yes} and \emph{no} instances form disjoint, but not necessarily complementary, sets.

\begin{definition}
  A pair of structures $\rel A$, $\rel A'$ with $\rel A \to \rel A'$ is called a \emph{promise template}. We assume that $\rel A$ in a promise template is finite unless explicitly specified otherwise.
  Given such a template, a \emph{(typed) promise constraint satisfaction problem} is a promise problem whose goal is, on an input $\rel X$ which is a structure (of the same signature), to output \emph{yes} if $\rel X \to \rel A$, and \emph{no} if $\rel X \not\to \rel A'$ (since $\rel A \to \rel A'$, these two cases are disjoint).
  We denote this problem by $\PCSP(\rel A, \rel A')$.
\end{definition}

A \emph{constraint} of an instance $\rel X$ of $\PCSP(\rel A, \rel A')$ is an expression of the form $(v_1, \dots, v_k) \in R^{\rel X}$ which we formally view as a pair $((v_1, \dots, v_k), R)$ where $v_1, \dots, v_k$ are elements of $X$ and $R$ is a relational symbol such that the expression is satisfied.
The tuple $(v_1, \dots, v_k)$ is called the \emph{scope} of the constraint.

Alternatively, one could define a promise CSP as a promise \emph{search} problem where the goal would be, given $\rel X$ that is promised to map to $\rel A$ via a homomorphism, to find a homomorphism $\rel X \to \rel A'$. It is currently not known whether these two versions of promise CSPs are of equivalent complexity. This paper focuses on the decision version of promise CSP, although a few results (mostly with some caveats) apply also to the search version. Finally, note that $\CSP(\rel A)$ is $\PCSP(\rel A, \rel A)$, and therefore every result about promise CSPs is applicable to CSPs.

Finally, we will denote by $\rel B^X$ the \emph{$X$-fold power} of a structure $\rel B$, i.e., a structure whose $t$-th domain is $B_t^X$ for each type $t$, and the relations are defined as follows: $(b_1, \dots, b_k) \in R^{\rel B^X}$ if and only if $(b_1(i), \dots, b_k(i)) \in R^{\rel B}$ for all $i\in X$.

\subsection{Reductions}
  \label{sec:reductions-prelim}

In this paper we deal with the simplest form of a reduction between promise problems. By a \emph{reduction} we mean a (usually efficiently computable) function between the instances of the two problems which maps positive instances to positive instances and negative instances to negative instances. In the case of promise CSPs, a reduction from $\PCSP(\rel A, \rel A')$ to $\PCSP(\rel B, \rel B')$ is a mapping $\psi$ between the corresponding categories of structures such that (1) if $\rel X \to \rel A$ then $\psi(\rel X) \to \rel B$, and (2) if $\rel X \not\to \rel A'$ then $\psi(\rel X) \not\to \rel B'$.
The item (1) is usually referred to as \emph{completeness}, and the item (2) as \emph{soundness}. We will usually use and prove the soundness in its converse form: if $\psi(\rel X) \to \rel B'$ then $\rel X \to \rel A'$.%
\footnote{In the case of a reduction between the search version of (promise) CSPs, it would be necessary to additionally require that this converse is witnessed by an efficiently computable function that, given a homomorphism $\psi(\rel X) \to \rel B'$, outputs a homomorphism $\rel X \to \rel A'$.}

Reductions are more general than algorithms: every decision algorithm can be viewed as a reduction to a problem with two admissible instances \emph{yes} and \emph{no} where \emph{yes} is the positive instance. We express this problem as a `trivial CSP' whose template is a structure in a very degenerate signature $\Omega$ which allows only two structures: $\Omega$ has no types and has a single nullary relational symbol $C$.  There are only two $\Omega$-structures: $\bot$ and $\top$. The structure $\bot$ is defined as the structure with empty nullary relation, and $\top$ is the structure with the non-empty nullary relation.

\begin{definition}
  The \emph{trivial CSP} is the CSP with template $\bot$.
  Note that $\bot$ is a positive instance of this CSP (since $\bot \to \bot$), and $\top$ is a negative instance (since $\top \not\to \bot$).
\end{definition}

We assume that all (promise) CSPs considered in this paper have negative instances --- this can be always ensured by adding a nullary constraint that cannot be satisfied.

\subsection{Gadget reductions}
  \label{sec:gadgets}

Gadget reductions are the more traditional reductions in the realm of finite-template CSPs and promise CSPs. They have been classified by the algebraic approach. We refer to Krokhin, Opršal, Wrochna, and Živný \cite[Section 4.1]{KOWZ23} for a detailed exposition.

A \emph{gadget replacement} is a transformation of a $\Pi$-structure $\rel X$ to a $\Sigma$-structure $\gamma(\rel X)$ that replaces every element $v \in \rel X$ by a fixed $\Pi$-structure that only depends on the type of $v$, and replaces every constraint $(v_1, \dots, v_k) \in R^{\rel X}$ by another $\Sigma$-structure that only depends on the relation $R$ identifying some of the vertices of this gadget with the elements of the gadgets introduced by replacing $v_1, \dots, v_k$ (see also \cite[Definition~2.9]{KO22}).%
\footnote{Gadgets and gadget replacements may be formalised elegantly in a categorical language: We view a signature $\Pi$ as a category whose objects are all types and all symbols of $\Pi$, and contains a morphism $p_i \colon R \to \ar_R(i)$ for each symbol $R$ of arity $k$ and $i \in [k]$. Hence a $\Pi$-structure can be viewed as a functor from $\Pi$ to the category of sets. A \emph{gadget} is then a contravariant functor from $\Pi$ to $\Sigma$-structures. Every such gadget then defines a pair of adjoint functors (see, e.g., Pultr \cite{Pul70}); \emph{gadget replacement} is the left adjoint, and the right adjoint corresponds (up to homomorphic equivalence) to \emph{pp-powers} \cite[see Definition 4.7]{BBKO21}.}
If such a gadget replacement yields a reduction, we call it a \emph{gadget reduction}. Gadget reductions are characterised using \emph{polymorphisms and minion homomorphisms} \cite[Definition 2.15]{BBKO21} (see also Definitions \ref{def:minion} and \ref{def:polymorphism-minion}).

\begin{theorem}[Barto, Bulín, Krokhin, and Opršal {\cite[Theorem 3.1]{BBKO21}}] \label{thm:gadget-characterisation}
  $\PCSP(\rel A, \rel A')$ reduces to $\PCSP(\rel B, \rel B')$ by a gadget reduction if and only if there is a minion homomorphism $\Pol(\rel B, \rel B') \to \Pol(\rel A, \rel A')$.
\end{theorem}

The gadget reduction provided by the above theorem in presence of a minion homomorphism between the polymorphisms of the templates is called a \emph{universal gadget reduction} (it should be noted that the reduction depends on the structures $\rel A$ and $\rel B$).

\section{Local consistency reductions}
  \label{sec:local-reductions}

In this section, we define the class of \Datalog and consistency reductions, we discuss some basic properties of these reductions, and provide formal statements of Theorems~\ref{thm:main-2} and \ref{thm:main-3}. We define the reductions from the logical perspective, followed by the combinatorial version and the equivalence of both definitions.

\subsection{Logical reductions}
  \label{sec:datalog} \label{sec:datalog-reduction}

We first describe the class of reductions we consider in this paper in terms of logic, or more precisely Datalog.
Since our signatures are multisorted this means that all logic formulae need to respect the types: We assume that every variable $x_1, x_2, \dots$ has an associated type and define an atomic $\Pi$-formula (where $\Pi$ is a relational signature) to be (1) $R(x_1,\dots, x_n)$ where $R$ is a $\Pi$-symbol, and the type of $x_i$ is $\ar_R(i)$ for all $i\in [k]$, or (2) an equality $x_1 = x_2$ where $x_1$ and $x_2$ are of the same type.

A \emph{Datalog program} $\phi$ with input signature $\Pi$ is constituted by a relational signature $\Delta$ satisfying $\Pi \subseteq \Delta$ (meaning that it has the same types as $\Pi$ and each symbol of $\Pi$ appears in $\Delta$ with the same arity) along with a finite collection of rules that are traditionally written in the form
\[
  t_0 \fro t_1, \dots, t_r
\]
where $t_0$, \dots, $t_r$ are atomic $\Delta$-formulae. In such a rule $t_0$ is called the \emph{head} and $t_1, \dots, t_r$ the \emph{body} of the rule. Moreover, we require that neither the symbols in $\Pi$ nor the equality appear in the head of any rule.  The symbols in $\Pi$ are called \emph{input symbols}, or \emph{EDBs} (standing for \emph{extensional database predicates}, while all the other symbols would be \emph{IDBs}, \emph{intensional database predicates}). Furthermore, one of the $\Delta$-symbols is designed as \emph{output}. \footnote{In database theory, Datalog programs very often have several output predicates (and define structures instead of relations). Such program can be viewed as a special case of a Datalog interpretation which we define later.} A Datalog program receives as input a $\Pi$-structure $\rel A$ and produces a relation with the same arity of the output predicate using standard fix-point semantics. That is, let $\rel B$ be the $\Delta$-structure computed in the following way.
\begin{enumerate}
  \item Start with setting $R^{\rel B} = R^{\rel A}$ if $R \in \Pi$ and $R^{\rel B} = \emptyset$ otherwise.
  \item If there is a rule  $t_0 \fro t_1, \dots, t_r$ and some assignment $h\colon X\rightarrow A$, where $X$ are the variables occurring in the rule, such that the atomic predicates $t_1, \dots, t_r$ hold in $\rel B$ after the substitution then include $(h(x_1), \dots, h(x_k))$ in $R^{\rel B}$ where $t_0 = R(x_1,\dots,x_k)$. Repeat until we get to a fixed point.
  \item Output the relation $R^{\rel B}\subseteq A_{\ar_R(1)} \times \dots \times A_{\ar_R(k)}$ where $R$ is the output predicate.
\end{enumerate}

We will denote the output of such a Datalog program by $\phi(\rel A)$.  Commonly in the context of CSPs, the output of a Datalog program is a nullary predicate, so that such a program outputs either \emph{true}, or \emph{false} --- we call such Datalog programs \emph{Datalog sentences}. The \emph{arity} of a Datalog program is the arity of the output predicate, and the \emph{width} of a Datalog program is the maximal number of variables in a rule.

We define \emph{Datalog interpretations} as a particular case of logical interpretations. This definition hinges on interpretation of Datalog programs as logical formulae. More precisely, a Datalog program $\phi$ whose output is a relation with arity $k$ is seen as a formula with $k$ free variables, so that $\rel A \models \phi(a_1,\dots, a_k)$ if and only if $(a_1, \dots, a_k) \in \phi(\rel A)$.

\begin{definition} \label{def:datalog-interpretation}
  Fix two relational signatures $\Sigma$ and $\Pi$. A \emph{Datalog interpretation} $\phi$ mapping $\Pi$-structures to $\Sigma$-structures (usually referred to as a Datalog interpretation of $\Sigma$ into $\Pi$) consist of: a Datalog program $\phi_t$ with input signature $\Pi$ for each $\Sigma$-type $t$, and a Datalog program $\phi_R$ also with input signature $\Pi$ for each $\Sigma$-symbol $R$, such that, for each symbol $R$ of arity $k$, the arity of $\phi_R$ is the concatenation of the arities of $\phi_{\ar_R(1)}$, \dots, $\phi_{\ar_R(k)}$.
  Such an interpretation is said to be of \emph{width at most $k$} if all $\phi_t$ and $\phi_R$ are of width at most $k$.

  The application of such a Datalog interpretation $\phi$ to a $\Pi$-struc\-ture $\rel A$ produces the $\Sigma$-struc\-ture
  \[
    \phi(\rel A) = (\phi_t(\rel A), \phi_s(\rel A), \dots;
      \phi_R(\rel A), \phi_S(\rel A), \dots),
  \]
  where, for each $\Sigma$-symbol $R$, $\phi_R(\rel A)$ is interpreted as a relation of the required arity $k$ in the natural way, i.e., it consists of all tuples
  \begin{multline*}
    ((a_{1},\dots, a_{j_1}), \dots, (a_{j_{k-1}+1}, \dots, a_{j_k})) \\
    \in \phi_{\ar_R(1)}(\rel A) \times \dots \times \phi_{\ar_R(k)}(\rel A),
  \end{multline*}
  such that $(a_{1},\dots, a_{j_k}) \in \phi_R(\rel A)$.
\end{definition}

Unless there is a clash of notation (in which case we will specify the output predicate explicitly) we shall be using $D_t$ and $R$ as output predicate of $\phi_t$ and $\phi_R$ respectively.

Since we do not allow parameters, nor inequality in Datalog interpretations, we obtain \emph{monotone} transformations in the following sense.

\begin{lemma} \label{lem:datalog-is-monotone}
  Let $\phi$ be a Datalog interpretation, and $\rel A$ and $\rel B$ structures. If $\rel A \to \rel B$, then $\phi(\rel A) \to \phi(\rel B)$.
\end{lemma}

Datalog interpretations are powerful when used as reductions; they can reduce a substantial class of CSPs, including 2SAT and Horn-3SAT, to the trivial problem. More precisely, it can be observed that $\PCSP(\rel A, \rel A')$ reduces to $\CSP(\bot)$ by a Datalog reduction if and only if it is expressible in Datalog, i.e., whenever there is a Datalog sentence which is \emph{true} in all negative instances, and \emph{false} in all positive instances of $\PCSP(\rel A, \rel A')$. This is exemplified in the following example.

\begin{example}[2-colouring]
  We shall define a Datalog interpretation $\phi$ that provides a reduction $\CSP(\rel K_2) \leq \CSP(\bot)$. To define such interpretation, we only need ot provide a nullary Datalog program $\phi_C$ which is based on a well-known program solving 2-colouring (see, e.g., \cite[p.~14]{KV00}):
  \begin{align*}
    P(x, y) &\fro E(x, y) \\
    P(x, y) &\fro E(y, x) \\
    P(x, y) &\fro P(x, z), P(z, w), P(w, y) \\
    C &\fro P(x, x)
  \end{align*}
  It is straightforward to check that $C$ is derived if and only if the input graph has an odd cycle. Consequently, we get $\phi(\rel G) = \top$ if and only if $\rel G$ is not 2-colourable. Since $\top$ is the negative instance of $\CSP(\bot)$, this concludes that $\phi$ is a reduction between the two problems.
\end{example}

Further examples of Datalog interpretations providing tractability of some (promise) CSPs can be obtained through a framework of Ciardo and Živný \cite{CZ23-hierarchies}; more precisely, it may be observed that the tensor construction introduced in \cite{CZ23-hierarchies} can be also expressed as a Datalog interpretation.
For more details about tensor hierarchies, see Appendix~\ref{app:tensors}.

Let us continue with an example of a reduction used to provide NP-hardness of the approximate graph colouring.

\begin{example}[Line digraph construction] \label{ex:arc-graph}
  The \emph{line digraph construction} $\delta$, described below, was used by Wrochna and Živný \cite{WZ20} to provide hardness of some cases of approximate graph colouring. They obtained these results by iterating (i.e., composing) $\delta$ to reduce from $\PCSP(\rel K_k, \rel K_{2^{\Omega(k^{1/3})}})$ for large enough $k$ (known to be NP-hard by a result of Huang \cite{Hua13}) to $\PCSP(\rel K_k, \rel K_{b(k)})$, where $b(k) = \binom k{\lfloor k/2 \rfloor} - 1$, for all $k\geq 4$. This result improved the state-of-the-art for all $k > 5$ and matched it for $k = 5$.
  The same reduction is also used by Guruswami and Sandeep \cite{GS20} to provide conditional hardness of approximate graph colouring under $d$-to-1 conjecture, and by Ciardo and Živný \cite{CZ22-lp-aip} to provide a negative result about solvability of approximate graph colouring by a combination of Sherali--Adams with affine integer programming.

  The \emph{line digraph construction} is a Datalog interpretation from digraphs to digraphs that changes the domain of an input digraph. It can be defined by a Datalog interpretation $\delta = (\delta_V; \delta_E)$ where $\delta_V$ is the program
  \[
    V'(x, y) \fro E(x, y)
  \]
  with output $V'$, and $\delta_E$ is the program
  \[
    E'(x, y, y, z) \fro E(x, y), E(y, z)
  \]
  with output $E'$.
  The application on a digraph $\rel G$ then yields a digraph $\delta(\rel G)$ whose vertices are edges of $\rel G$, and two such `vertices' are connected by an edge if they are incident, i.e., of the form $(u,v)$ and $(v, w)$.
\end{example}

The monotone version of Datalog interpretations that we defined above cannot fully emulate gadget reductions. In particular, it cannot express taking disjoint unions of domains and relations. Atserias, Bulatov, and Dawar \cite{ABD09} use parameters in Datalog interpretations to express these disjoint unions up to a finite number of exceptions --- we note that allowing parameters (or inequality) would yield reductions that are not monotone. Instead, we extend Datalog interpretations with disjoint unions.
Extending logical interpretations to allow taking disjoint unions of definable relations is not completely unusual (see, e.g., \cite[Definition 3.1]{BojanczykK24}). Here, we take an advantage from the multisorted approach to provide a technical solution that helps us to track that the relations are well-defined.

\begin{definition} \label{def:union-gadget}
  Assume that $\Pi$ and $\Sigma$ are two relational signatures. A \emph{union gadget} $\upsilon$ is defined by a pair of mappings $(d, r)$, where $d$ maps $\Pi$-types to $\Sigma$-types and $r$ maps $\Pi$-symbols to $\Sigma$-symbols, such that for each $\Pi$-symbol $R$, we have $\ar_{r(R)} = d\circ \ar_R$.

  Such a union gadget can be then applied on a structure $\rel A$ with disjoint domains\footnote{We can always enforce disjoint domains by replacing an element $a$ of type $t$ with a pair $(a; t)$.} to produce a $\Sigma$-structure $\upsilon(\rel A)$ whose $t$-th domain is the (disjoint) union of $A_i$'s with $d(i) = t$, and similarly, $S^{\upsilon(\rel A)} = \bigcup_{R\in r^{-1}(S)} R^{\rel A}$ for each $\Sigma$-symbol $S$.
\end{definition}

It is not difficult to observe that disjoint unions are also expressible as gadget replacements. 
 
\begin{definition}
  A \emph{Datalog$^\cup$ reduction} is a composition $\upsilon\circ \phi$ where $\upsilon$ is a union gadget and $\phi$ a Datalog interpretation.
  We write
  \[
    \PCSP(\rel A, \rel A') \leqdl \PCSP(\rel B, \rel B')
  \]
  if there exists a Datalog$^\cup$ reduction between the two problems.
\end{definition}

Note that a \Datalog reduction between two single-sorted CSPs may use a multisorted structure as an intermediate step even though we are reducing between two single-sorted CSPs.

The fact that \Datalog reductions compose follows from observing that Datalog interpretations compose, and that they can be permuted with disjoint unions, which also compose. This makes the relation $\leqdl$ (on promise CSPs) transitive, and hence it defines a preorder.

\begin{theorem} \label{thm:ddatalog-composes}
  Let $\phi$ and $\chi$ be two Datalog$^\cup$ reductions such that the output signature of $\phi$ coincides with the input signature of $\chi$. Then there is a Datalog$^\cup$ reduction $\psi$ such that $\psi(\rel A)$ and $\chi\phi(\rel A)$ are isomorphic for all structures $\rel A$.
\end{theorem}

\begin{example} \label{ex:bk-do-not-compose}
  The \emph{$k$-reductions} \cite[see p.~48]{KO22} used by Barto and Kozik~\cite{BK22} do not compose unlike \Datalog reductions. These weaker reductions can be described briefly as a combination of Datalog interpretations without recursion and gadget reductions.
  We denote by $\rel P_n$ the oriented path with $n$ edges, i.e., the digraph $(\{0, \dots, n\}; \{(i, i+1) \mid i < n\})$.
  It can be observed that: (1) $\CSP(\rel P_2)$ reduces to $\CSP(\rel P_1)$ by a gadget reduction which is a special case of a $k$-reduction, (2) $\CSP(\rel P_1)$ reduces to $\CSP(\bot)$ by a $k$-reduction, and (3) $\CSP(\rel P_2)$ does not reduce to $\CSP(\bot)$ by a $k$-reduction (this is since $\CSP(\rel P_2)$ does not have \emph{finite duality} \cite[Definition 3]{BulatovKL08}).
\end{example}

Clearly, every Datalog interpretation is a \Datalog reduction (compose with the trivial union gadget). We show that every gadget replacement can be also expressed as a Datalog$^\cup$ reduction up to homomorphic equivalence.

\begin{theorem} \label{thm:gadget-is-ddatalog}
  For every gadget $\gamma$ there is a Datalog$^\cup$ reduction $\psi$ such that, for all $\rel A$, $\gamma(\rel A)$ and $\psi(\rel A)$ are homomorphically equivalent.
\end{theorem}

\subsection{Combinatorial reductions}
  \label{sec:local-consistency} \label{sec:consistency-reduction}

The combinatorial counterpart of Datalog$^\cup$ reductions is based on the local consistency algorithm for CSPs and the universal gadget reduction.
Consistency reductions are defined by the two templates and a single parameter $k$. Let us assume that we are trying to reduce $\PCSP(\rel A, {*})$ to $\PCSP(\rel B, {*})$; the second part of the templates are irrelevant for the definition of the reduction, but play an important role for the soundness which we do not characterise here.

The $k$-consistency reduction is a combination of two procedures: \emph{Enforcing $k$-consistency} which is possibly the most intuitive approach to solving CSPs, and it is used in many CSPs algorithms (including, e.g., Zhuk's polynomial-time algorithm \cite{Zhu20}), and \emph{the universal gadget replacement} which is based on a standard way to reduce from label cover to $\CSP(\rel B)$ using a gadget consisting of powers of $\rel B$ (see, e.g., \cite[Section 3.3]{BBKO21}).

Let $\rel X$ and $\rel A$ be structures of the same signature and let $K \subseteq X$. A \emph{partial homomorphism} from $K$ to $\rel A$ is a mapping $f\colon K \to A$ that preserves types and relations, i.e., it satisfies the definition of a homomorphism when all variables are quantified in $K$ instead of $X$. Intuitively, a partial homomorphism is simply a partial solution to the instance defined on the given set $K$.
We denote by $\binom{X}{\leq k}$ the set of all at most $k$-element subsets of~$X$.

\begin{definition}[$k$-consistency reduction] \label{alg:k-consistency-reduction}
  Let $k \geq 1$ and let $\rel A$ be a $\Pi$-structure and $\rel B$ be a $\Sigma$-structure.\footnote{It is useful but not necessary to assume that $k \geq m$, where $m$ is the maximal arity of the relations of $\rel A$, since the procedure below ignores all constraints of arity bigger than $k$.}
  The $k$-consistency reduction maps every $\Pi$-structure $\rel X$ to a $\Sigma$-structure $\rel Y$ constructed by the following procedure:
  \begin{enumerate}
    \item For each $K \in \binom X{\leq k}$, let $\mathcal F_K$ be the set of all partial homomorphisms from $K$ to $\rel A$.
    \item Ensure that for each $L \subset K \in \binom X{\leq k}$, the sets $\mathcal F_K$ and $\mathcal F_L$ are \emph{consistent}, i.e., remove from $\mathcal F_L$ all $h\colon L \to A$ that do not extend to a $g\in \mathcal F_K$, and remove from $\mathcal F_K$ all $g$ whose restriction $g|_L$ is not in $\mathcal F_L$.
    \item Repeat (2), while anything changes.
    \item For each $K \in \binom X{\leq k}$, include in $\rel Y$ a copy of $\rel B^{\mathcal F_K}$ whose elements are denoted by $(K; b)$ where $b\colon \mathcal F_K \to B$.
    \item For each $L \subseteq K \in \binom X{\leq k}$ and $b\colon \mathcal F_L \to B$, identify the elements $(K; b \circ \rho)$ and $(L; b)$ where $\rho \colon \mathcal F_K \to F_L$ is the restriction $\rho(g) = g|_L$.
  \end{enumerate}
  We denote the output $\rel Y$ by $\kappa_k^{\rel A,\rel B}(\rel X)$.
\end{definition}

In a nutshell the $k$-consistency reduction is a concatenation of the $k$-consistency algorithm (steps (1), (2), and (3)) and the universal gadget reduction (steps (4) and (5)).
Here, in steps (1), (2), and (3) we are using the standard formulation of the $k$-consistency algorithm as a winning strategy in the existential $k$-pebble game. Note that these steps can be turned into a decision algorithm for CSPs by outputting \emph{no} if one of the sets $\mathcal F_K$ (and consequently each of them) is empty, and outputting \emph{yes} if all $\mathcal F_K$'s are non-empty; this coincides with \cite[Algorithm 1 on p.5]{BK14}. We shall refer to this decision algorithm as the \emph{$k$-consistency test} to distinguish it from the $k$-consistency algorithm corresponding to the application of steps (1), (2), and (3). In turn, they should not be confused with the $k$-consistency reduction defined above.
Steps (4) and (5) correspond to the universal gadget reduction; these two steps are a standard way of reducing from Label Cover to another CSP used in the algebraic approach \cite[Section 3.3]{BBKO21}.

\begin{example} \label{ex:hypergraphs}
  Barto and Kozik \cite{BK22} provided a sufficient condition for the existence of an efficient reduction between two promise CSPs, which is a weaker version of the $k$-consistency reduction.
  Many known NP-hard promise CSPs can be explained by a reduction from graph $3$-colouring using this sufficient condition.
  For example, the approximate hypergraph colouring (first proved to be NP-hard by Dinur, Regev, and Smyth \cite{DRS05}) allows a $k$-consistency reduction from $\CSP(\rel K_3)$ as shown by Wrochna \cite{Wro22a}. The same approach can also provide the NP-hardness results in \cite{AGH17,KO19,FKOS19,WZ20,BWZ21}.
  We return to this sufficient condition in Appendix~\ref{app:sufficient-condition}.
\end{example}

The main result of this section is that \Datalog reductions and consistency reductions have the same power in the scope of promise CSPs. In fact we prove a refined version of the statement that relates the parameters of the two reductions.

\begin{restatable}{theorem}{canonicalwidth} \label{thm:canonical-width}
  Fix $k \geq 1$ and let $\rel A, \rel A'$ and $\rel B, \rel B'$ be two promise templates. The following are equivalent:
  \begin{enumerate}
    \item There exists a Datalog interpretation $\phi$ of width $k$ and a gadget $\gamma$ such that $\gamma\circ\phi$ is a reduction from $\PCSP(\rel A, \rel A')$ to $\PCSP(\rel B, \rel B')$.
    \item $\PCSP(\rel A, \rel A') \leqcons \PCSP(\rel B, \rel B')$.
  \end{enumerate}
\end{restatable}

We prove the theorem by showing that the $k$-consistency reduction is essentially a composition of a \emph{canonical Datalog interpretation} and a \emph{universal gadget}. This claim immediately gives one of the two implications. The other implication is proved by showing that both the canonical Datalog interpretation and the universal gadget have a certain universal property. Loosely speaking, the canonical Datalog interpretation of width $k$ can be used in place of any other Datalog interpretation of width $k$, and the universal gadget can be used in place of any other gadget.
This general idea has a lot of subtle details that are dealt with in the technical part of the paper, e.g., what if the output signatures of the Datalog interpretation does not agree with the output signature of the canonical Datalog program?

As a direct corollary of the above theorem, we get that $k$-con\-sis\-tency reductions are equivalent to Datalog reductions in the following sense.

\begin{corollary} \label{cor:canonical-form} \label{cor:equivalence}
  $\PCSP(\rel A, \rel A') \leqdl \PCSP(\rel B, \rel B')$ if and only if there exists $k \geq 1$ such that $\PCSP(\rel A, \rel A') \leqcons \PCSP(\rel B, \rel B')$.
\end{corollary}

\begin{proof}
  By definition, every \Datalog reduction is a composition of a Datalog interpretation $\phi$ and a union gadget $\upsilon$, which is equivalently expressed as a gadget reduction. Hence Theorem~\ref{thm:canonical-width} applies for $k$ equal to the width of $\phi$. Consequently, $\PCSP(\rel A, \rel A') \leqcons \PCSP(\rel B, \rel B')$.

  If we assume that $\PCSP(\rel A, \rel A') \leqcons \PCSP(\rel B, \rel B')$ for some $k\geq 1$, then the same theorem provides a Datalog interpretation $\phi$ and a gadget $\gamma$ whose composition gives a reduction. Since both $\phi$ and $\gamma$ are equivalent to Datalog$^\cup$ reductions (the latter by Theorem~\ref{thm:gadget-is-ddatalog}), we get that $\PCSP(\rel A, \rel A') \leqdl \PCSP(\rel B, \rel B')$ by composition (Theorem~\ref{thm:ddatalog-composes}).
\end{proof}

\begin{example}
  Building on Example~\ref{ex:hypergraphs}, note that the \NP-hard\-ness of $\PCSP(\rel K_3, \rel K_5)$ follows by a $k$-consistency reduction from $\CSP(\rel K_3)$ since
  \[
    \CSP(\rel K_3) \leqcons \PCSP(\rel H_2, \rel H_{17480}) \leqdl \PCSP(\rel K_3, \rel K_5)
  \]
  where $\PCSP(\rel H_2, \rel H_{17480})$ is the approximate hypergraph colouring problem mentioned in Example~\ref{ex:hypergraphs}. The first reduction is the one described in Example~\ref{ex:hypergraphs} (it is expressible as a \Datalog reduction by the above corollary), and the second reduction is provided by a combination of \cite[Theorem 6.5]{BBKO21} and Theorem~\ref{thm:gadget-is-ddatalog}. Consequently, since \Datalog reductions compose, we get that $\CSP(\rel K_3) \leqdl \CSP(\rel K_3, \rel K_5)$.

  Naturally, we could ask for a concrete Datalog interpretation (together with the associated union gadget) that provides the above reduction. Such an interpretation is far from obvious:
  First, the \Datalog reduction from hypergraph colouring to graph colouring is an unfolding of the gadget reduction provided by the algebraic approach, and hence quite complex to write down explicitly.
  Second, \cite[Theorem 5.1]{BK22} does not provide an explicit bound on $k$ for the $k$-consistency reduction from graph colouring to hypergraph colouring (although this bound could be reconstructed from the proof).
  We also do not know what is the smallest $k$, s.t., $\CSP(\rel K_3) \leqcons \CSP(\rel K_3, \rel K_5)$.
\end{example}

\section{Hierarchies and CSP algorithms}
  \label{sec:hierarchies}

Several hierarchies of algorithms are studied for the tractability of CSPs and promise CSPs. The most prominent are arguably the local consistency hierarchy, the Sherali--Adams hierarchy \cite{SA90}, and the Lasserre hierarchy of semi-definite programming relaxations. Furthermore, Ciardo and Živný \cite{CZ23-hierarchies} introduced a framework to create more hierarchies through a tensor construction. There is ongoing research aiming to characterise which problems are solvable by some level of such a hierarchy (e.g., Ciardo and Živný \cite{CZ22-lp-aip}).

In this section we provide an alternative look at these hierarchies through the lenses of the $k$-consistency reduction. In particular, we associate a hierarchy to every (promise) CSP, where the $k$-th level of the hierarchy consists of all (promise) CSPs that reduce to it by the $k$-consistency reduction. As we will show below, by choosing adequately the initial (promise) CSP, we recover several of the above hierarchies. In particular, the bounded width hierarchy is the `consistency hierarchy' associated to the trivial problem, and the Sherali--Adams hierarchy is the `consistency hierarchy' associated to linear programming.
One of the benefits of expressing hierarchies in this way is that, by Theorems~\ref{thm:gadget-is-ddatalog} and \ref{thm:ddatalog-composes}, we immediately get that the consistency hierarchy of a fixed problem is closed under gadget reductions, which in particular means that there is a possibility to describe membership in this hierarchy by the means of polymorphisms. We note that this is not always clear for the analogous tensor hierarchies.
Finally, we consider new hierarchies not previously studied, namely, the hierarchy associated to solving systems of linear equations over integers, or over a finite (Abelian) group.

In this section, we work with infinite template CSPs, e.g., we express linear programming as a CSP with domain $\mathbb Q$. 
This creates some difficulties since, if $\mathbf B$ is an infinite structure in Definition~\ref{alg:k-consistency-reduction}, then the output of the $k$-consistency reduction will be an infinite instance as well.
Nevertheless, if $\mathbf B$ is one of the infinite templates considered in this section, it is always possible to produce an equivalent, but finite, instance. We provide proofs of this statement for each template separately in the corresponding subsection.

\subsection{Bounded width}
  \label{sec:bw}

As a warm-up, we start with the hierarchy of promise CSPs with bounded width. This hierarchy corresponds to the smallest (easiest) class, w.r.t., \Datalog reductions, of promise CSPs. In the above sense, it can be described as the $k$-consistency hierarchy associated to the trivial problem $\CSP(\False)$.
Traditionally, the problems belonging to this hierarchy are said to be of \emph{bounded width} --- a (promise) CSP is said to be of \emph{width at most $k$} if the $k$-consistency test solves the problem, or equivalently if it is expressible by a Datalog sentence of width $k$ (as we mentioned before, the proof of the corresponding theorem for standard CSPs by Kolaitis and Vardi \cite{KV00} applies to the promise setting). We show that this property is also equivalent to having a $k$-consistency reduction to the trivial problem.

\begin{theorem} \label{thm:bw}
  $\PCSP(\rel A, \rel A')$ has width $k$ if and only if it reduces to $\CSP(\False)$ by the $k$-consistency reduction.
\end{theorem}

\begin{proof}
  If $\PCSP(\rel A, \rel A')$ is expressible by a Datalog formula of width $k$, there is Datalog interpretation $\phi$ of width $k$ that reduces from $\PCSP(\rel A, \rel A')$ to $\PCSP(\False)$. Consequently, by Theorem~\ref{thm:canonical-width}, we get the required $k$-consistency reduction.

  If $\PCSP(\rel A, \rel A') \leqcons \CSP(\False)$ and $\rel X$ is accepted by the $k$-consistency test, then $\kappa_k^{\rel A,\False}(\rel X) = \False$ since $\False^D = \False$ for any $D \neq\emptyset$, and hence there are no constraints. Consequently, we get that $\rel X \to \rel A'$ by the soundness of the $k$-consistency reduction.
\end{proof}

While the previous result is somewhat expected there is more than meets the eye as the statement implies that bounded width promise CSP problems are closed under gadget reductions. Although this has been previously shown \cite[Lemma 7.5]{BBKO21}, a direct proof requires a bit of a case analysis.

We also note without a proof that this hierarchy corresponds to the $k$-consistency hierarchy associated to Horn-SAT.

\subsection{Sherali--Adams}
  \label{sec:sa}

There are several slightly different ways to define the Sherali--Adams (SA) relaxation for CSPs and promise CSPs, e.g., \cite{TZ17,CZ23-hierarchies,BD21,BB22}.
Although they might differ in minor technical details, they are all based on the scheme introduced by Sherali and Adams \cite{SA90} to generate increasingly tight relaxations of a linear program. In all cases, a CSP instance is turned into an instance of linear programming whose tightness can be adjusted using a parameter $k$.

Linear programming can be viewed as a fixed template CSP though with infinite domain and infinitely many relations: Namely, the domain is $\mathbb Q$,\footnote{Although maybe a more natural choice of the domain would be $\mathbb R$, we prefer $\mathbb Q$ since it is countable, and hence there is less discussion about encoding of coefficients and solutions.} and the relations are all relations defined by affine inequalities, e.g., of the form
\[
  a_1 x_1 + \dots + a_n x_n \leq b
\]
for some $a_1, \dots, a_n, b \in \mathbb Q$.
In an instance of linear programming, the relations are usually given as the tuples of their coefficients. Nevertheless, all of these relations are in fact expressible (more precisely, definable by a \emph{primitive positive formula}) using the following three types of inequalities: $x \leq y$, $x_1 + x_2 = y$, and $y = 1$.
With some care, the length of these primitive positive definitions becomes proportionate to the binary encoding of the coefficients.
This means that up to some simple gadget reductions, we can define the template of linear programming, denoted $\rel Q_\conv$, as the structure with these three relations.

In this section we shall show that the Sherali--Adams hierarchy coincides with the $k$-consistency hierarchy associated to $\rel Q_\conv$. While a weaker version of our results (where the hierarchies are only required to interleave) holds for any natural variant of the SA relaxation, we can additionally prove that both hierarchies align perfectly if we choose it suitably.  Our definition agrees with the SA$(k, k)$ system of Thapper and Živný \cite[Section 3.1]{TZ17} assuming that $k$ is at least the maximal arity of the constraints.

In plain words, the goal of the $k$-th level of Sherali--Adams relaxation is to find a collection of probability distributions on partial solutions on each of the subsets of variables of size at most $k$ that have consistent marginals.

\begin{definition}[$k$-th level of Sherali--Adams relaxation] \label{def:sa}
  Fix $k \geq 1$ and a structure $\rel A$. The $k$-th level of Sherali--Adams relaxation with template $\rel A$ is the mapping that given a structure $\rel X$ of the same signature as $\rel A$ produces a linear program in the following way:
  For each $K \in \binom X{\leq k}$, let $\mathcal F_K$ be the set of all partial homomorphisms from $K$ to $\rel A$.
  The $k$-th level of Sherali--Adams relaxation, denoted by $\SA^k(\rel X)$, is the following linear program with variables $x_{K,f} \in [0,1]$ for each $K\in \binom X{\leq k}$ and $f\in \mathcal F_K$:
  \begin{align*}
    \sum_{f\in \mathcal F_K} x_{K,f} &= 1 & \text{for each } K\in \textstyle\binom X{\leq k}, \\
    \sum_{f\in \mathcal F_K, f|_L = g} x_{K,f} &= x_{L,g}  & \text{for each } L \subset K\in \textstyle\binom X{\leq k}, \text{ and } g\in \mathcal F_L.
  \end{align*}
\end{definition}

Observe that $\SA^k(\rel X)$ has a 0-1 solution if there is a homomorphism $h\colon \rel X \to \rel A$; assign $x_{f, K} = 1$ if $f = h|_K$ and $x_{f, K} = 0$ otherwise. This observation provides completeness of the Sherali--Adams relaxation.
We say that the $k$-th level of the Sherali--Adams relaxation \emph{solves} $\PCSP(\rel A, \rel A')$ if it is also sound in the sense that if $\rel X \not\to \rel A'$, for some instance $\rel X$ of the promise CSP, then the $\SA^k(\rel X)$ does not have a feasible solution.
Having settled on the definition, we formulate the main theorem of this subsection.

\begin{theorem} \label{thm:sa}
  $\PCSP(\rel A, \rel A')$ is solvable by the $k$-th level of the Sherali--Adams relaxation if and only if it reduces to $\CSP(\rel Q_\conv)$ by the $k$-consistency reduction.
\end{theorem}

Let us briefly sketch the proof. The theorem is proven by a step-by-step comparison of the $k$-consistency reduction and the $\SA^k$.
We first note that the system $\mathcal F_K$ constructed in $\SA^k$ is precisely equivalent to step (1) of the $k$-consistency procedure.
We then show that steps (2) and (3) of the $k$-consistency are not needed since any solution to the linear program $\SA^k$ yields a consistent system; indeed, it is not hard to check that if $x_{K,f}$ is a solution to $\SA^k$, then
\[
  \mathcal F'_K = \{ f\in \mathcal F_K \mid x_{K,f} > 0 \}
\]
is a consistent system. This means that replacing $\mathcal F_K$ in the construction of the $\SA^k$ system with the sets obtained by enforcing $k$-consistency, i.e., at the end of step (3) of the procedure, does not change the output of $\SA^k$.
Next, we show that we can replace the equations in $\SA^k$ with the universal gadget for linear programming (as in steps (4) and (5) of the $k$-consistency) without increasing (or decreasing) the power of the relaxation, i.e., we show that $\SA^k(\rel X)$ has a feasible solution if and only if $\kappa_k^{\rel A, \rel Q_\conv}(\rel X) \to \rel Q_\conv$.
For the `only if' direction, we can construct, from a feasible solution of $\SA^k(\rel X)$, a homomorphism $h \colon \kappa_k^{\rel A, \rel Q_\conv}(\rel X) \to \rel Q_\conv$ by letting
\begin{equation} \label{eq:sa-homomorphism}
  h([K; q]) = \sum_{f\in \mathcal F_K} x_{f, K}q(f)
\end{equation}
where $K\in \binom X{\leq k}$ and $q\colon \mathcal F_K \to \mathbb Q$. It is relatively straightforward to check that $h$ is indeed a homomorphism. For the `if' direction, starting with a homomorphism $h\colon \kappa_k^{\rel A, \rel Q_\conv}(\rel X) \to \rel Q_\conv$, we can define a solution to the Sherali--Adams system by
\begin{equation} \label{eq:sa-solution}
  x_{K, f} = h([K; e_f])
\end{equation}
where $e_f\colon \mathcal F_K \to \mathbb Q$ is defined by $e_f(f) = 1$ and $e_f(g) = 0$ if $g\neq f$. Again, checking that the $x_{K, f}$'s, thus defined, is a feasible solution of $\SA^k(\rel X)$ is rather straightforward.

In the detailed proof (presented in Appendix~\ref{app:proofs-iv}), we show a more general statement that uses abstract properties of linear programming and its polymorphisms rather than formalising the above argument directly. This argument can be generalised to provide analogues of Theorem~\ref{thm:sa} for other hierarchies. In particular, it can proven that the hierarchy of a~\emph{conic minion} $\clo M$ \cite[Definition 3.6]{CZ23-hierarchies} coincides with $k$-consistency hierarchy for the corresponding promise version of label cover (i.e., the problem sometimes denoted by $\PMC_{\!\clo M}$ \cite[Definition 3.10]{BBKO21}) assuming $k$ is greater or equal than the arity of any relation of the template.

\begin{example}
  Ciardo and Živný \cite{CZ22-lp-aip} prove that $\PCSP(\rel K_c, \rel K_d)$ cannot be solved by any level of a hierarchy of combined LP and AIP for any $d \geq c \geq 3$ (which we call \emph{LP+AIP hierarchy}).
  The hierarchy was introduced by Brakensiek, Guruswami, Wrochna, and Živný \cite{BGWZ20}, and can be described as a tensor hierarchy of a conic minion, and hence as a $k$-consistency hierarchy of a fixed problem for which an analogue of Theorem~\ref{thm:sa} is valid.

  The proof that approximate graph colouring is not solved by this hierarchy relies on the line digraph reduction from Example~\ref{ex:arc-graph}.
  First, Ciardo and Živný show that $\PCSP(\rel K_c, \rel K_d)$ is not solved by the $k$-th level of the hierarchy for any $c \geq {(k^2 + k)/2}$ (we assume $k \geq 2$). Then they use the line digraph construction to reduce a general case to this case.

  The reduction from a general case can be explained using our framework as follows: Observe that the line digraph construction $\delta$ provides a reduction
  \[
    \PCSP(\rel K_{b(c)}, \rel K_{b(d)}) \leqdl \PCSP(\rel K_c, \rel K_d)
  \]
  where $b(n) = \binom n{\lfloor n/2 \rfloor}$ \cite[Lemma 4.16]{KOWZ23}.
  Furthermore, it is easy to observe that for each Datalog interpretation $\phi$ of width $k$, $\phi\circ \delta$ is expressible as a Datalog interpretation of width $2k$ (by inspection of our proof of Theorem~\ref{thm:ddatalog-composes}).
  Hence, if $\PCSP(\rel K_c, \rel K_d)$ is solvable by the $k$-th level of the hierarchy, then $\PCSP(\rel K_{b(c)}, \rel K_{b(d)})$ is solvable by the $2k$-th level of the hierarchy by Theorem~\ref{thm:canonical-width}.
  Chaining these reductions, we obtain that if $\PCSP(\rel K_c, \rel K_d)$ is solvable by the $k$-th level, then $\PCSP(\rel K_{b^n(c)}, \rel K_{b^n(d)})$ is solvable by the $2^nk$-th level. Since $b$ increases exponentially (more precisely $b(c) \geq 2^c/c$), we eventually get that 
  $b^n(c) \geq ((2^nk)^2 + 2^nk)/2$, and we may apply that $\PCSP(\rel K_{b^n(c)}, \rel K_{b^n(d)})$ is not solved by the $2^nk$-th level to finish the prove.

  Our argument applies to any hierarchy which can be described using $k$-consistency reductions, which includes all tensor hierarchies of conic minions for which an argument is given in the full version of \cite{CZ22-lp-aip}.
  We believe that the argument presented here is simpler than the one presented in \cite{CZ22-lp-aip}, and moreover, it still holds if one replaces $\delta$ with any \Datalog reduction (since the width of the composition of two Datalog interpretations is always bounded by the product of the widths).
\end{example}

Finally, we have the following immediate corollary of Theorems~\ref{thm:sa}, \ref{thm:gadget-is-ddatalog}, and \ref{thm:canonical-width}, which we believe has not been shown before in the scope of promise CSPs. The analogous statement in the non-promise setting follows from the characterisation of Sherali--Adams for CSPs \cite{BK14,TZ17}.
Again, the same can be proven for any hierarchy of a conic minion, and in particular for the Lasserre hierarchy.

\begin{corollary} \label{cor:sa}
  Let $\rel A, \rel A'$ and $\rel B, \rel B'$ be two promise templates such that $\Pol(\rel B, \rel B') \to \Pol(\rel A, \rel A')$. If $\PCSP(\rel B, \rel B')$ is solvable by some level of Sherali--Adams hierarchy, then so is $\PCSP(\rel A, \rel A')$.
\end{corollary}

An important consequence of the above corollary is that a characterisation of the applicability of the Sherali--Adams algorithm in the scope of promise CSPs could be theoretically described by an algebraic condition on polymorphisms.

\subsection{Hierarchies of groups}
  \label{sec:groups}

Solving systems of equations over a (finite) Abelian group is a well-known CSP that cannot be solved by neither the $k$-consistency  nor the Sherali--Adams relaxations.
We start with formally defining group CSPs.%
\footnote{We note that what our definition of a group CSP differs from Berkholz and Grohe \cite{BerkholzG15,BerkholzG17}; for Abelian groups this difference does not matter, but the Berkholz--Grohe CSP of a non-Abelian finite group is polynomial time solvable while ours is \NP-complete (see Goldmann and Russell \cite{GR02}).
}

\begin{definition} \label{def:csp-g}
  Let $\mathbb G$ be an Abelian group. We use $\CSP(\mathbb G)$ to denote the following problem: given a systems of equations of the form
  \begin{equation} \label{eq:group-csp}
    a_1 x_1 + \dots + a_n x_n = b
  \end{equation}
  where $a_1, \dots, a_n \in \mathbb Z$ and $b\in \mathbb G$, decide whether it has a feasible solution.
  This problem can be formulated as the CSP with template $\rel G$ whose domain is $G$ and relations are $x_1 + x_2 = y$ and $y = b$ for each $b \in \mathbb G$. Moreover, it is enough to consider $b$ from a fixed generating set\footnote{A set $A \subseteq G$ is said to \emph{generate} a group $\mathbb G$ if there is no proper subgroup of $\rel G$ that contains $A$.} of $\mathbb G$ in the relations of the second form.

  Similarly, if $\mathbb G$ is a finite non-Abelian group, we let $\CSP(\mathbb G)$ to be the CSP with template $\rel G$ defined in the same way as above (note that the equational definition is more tricky in this case). %
\end{definition}

For simplicity, we restrict our investigations to cyclic groups. For several reasons, little generality is loss due to this restriction. First, solving systems of equations over a non-Abelian finite group $\mathbb G$, and hence $\CSP(\mathbb G)$, is \NP-complete. Further, every finite Abelian group $\mathbb G$ is a product of cyclic groups, and finally, the infinite cyclic group $\mathbb Z$ provides a uniform algorithm, the so-called basic affine integer relaxation, for solving $\CSP(\mathbb G)$ where $\mathbb G$ is an Abelian group.
In particular, for every finite Abelian group $\mathbb G$, $\CSP(\mathbb G)$ is reducible via a gadget reduction to $\CSP(\mathbb Z)$ since $\rel G$ has an \emph{alternating polymorphism} \cite[Section 7.3]{BBKO21}.

In the rest of the section, we assume that $\mathbb G$ is a cyclic group generated by an element denoted by $1$ (i.e., it is either isomorphic to the group $\mathbb Z_n$ of addition modulo $n \in \mathbb Z$, or to $\mathbb Z$ itself). We introduce the following relaxation.

\begin{definition}($\mathbb G$-affine $k$-consistency relaxation)
  \label{al:affine}
  Fix $k \geq 1$ and a structure $\rel A$. The $\mathbb G$-affine $k$-consistency relaxation with template $\rel A$ is the mapping that, given a structure $\rel X$ of the same signature as $\rel A$, produces a system of equations over $\mathbb G$ in the following way:
  Run the $k$-consistency to compute sets $\mathcal F_K$, i.e., execute steps (1), (2), and (3) in Definition~\ref{alg:k-consistency-reduction}.
  Consider the same system of linear equations as in Definition~\ref{def:sa} except that the variables $x_{K, f}$ are now valued over~$\mathbb G$.
\end{definition}

Recall that it can be decided in polynomial time whether the resulting system of equations has a feasible solution. Also, it can be easily verified that it has a feasible solution whenever there is an homomorphism $h\colon \rel X \to \rel A$. In particular, the assignment $x_{K, h|_K} = 1$ extends to a solution of the system by assigning $0$ to all other variables.

The $\mathbb Z$-affine $k$-consistency relaxations (for varying $k$) are similar to relaxations considered by Berkholz and Grohe \cite[see Section 2.2]{BerkholzG17}; the system of Berkholz and Grohe is obtained by skipping consistency enforcement, i.e., steps (2) and (3). Even weaker hierarchy of AIP was consider by Ciardo and Živný \cite{CZ23-hierarchies}.

The following lemma shows that the $\mathbb G$-affine $k$-consistency relaxation is equivalent to the $k$-consistency reduction to $\CSP(\mathbb G)$.

\begin{lemma} \label{lem:V.7}
  Let $\mathbb G$ be a cyclic group. The $\mathbb G$-affine $k$-consistency relaxation solves $\PCSP(\rel A, \rel A')$ if and only if $\PCSP(\rel A, \rel A')$ reduces to $\CSP(\mathbb G)$ by the $k$-consistency reduction.
\end{lemma}

\begin{proof}[Proof sketch]
  We only need to show that replacing the linear system in Definition~\ref{al:affine} by steps (4) and (5) of the $k$-consistency reduction does not yield a more powerful reduction. This is argued in a similar way as in the sketch of the proof of Theorem~\ref{thm:sa}. Since $\mathbb G$ is either $\mathbb Z_n$, or $\mathbb Z$, we can assume a ring structure on $\mathbb G$.
  This allows us to use analogous definitions for $h$ and $x$ as in the case of linear programming above, i.e., as in Equations~\eqref{eq:sa-homomorphism} and \eqref{eq:sa-solution} replacing $\mathbb Q$ by $\mathbb G$ in both definitions.
\end{proof}

\subsubsection*{Obstructions to consistency reductions and a conjecture}

In this subsection we shall explore the role of group CSPs as `obstacles' for $k$-consistency reductions. A seminal example is the fact that $k$-consistency does not solve $\CSP(\mathbb G)$ for any non-trivial finite Abelian group $\mathbb G$ \cite[Theorem 30]{FV98}, or equivalently, $\CSP(\mathbb G) \not\leqdl \CSP(\False)$. It then follows that, e.g., $\CSP(\rel K_3)\not\leqdl \CSP(\False)$ since, say, $\CSP(\mathbb Z_3)$ reduces to $\CSP(\rel K_3)$ by a gadget reduction and \Datalog reductions are transitive by Theorem~\ref{thm:ddatalog-composes}.
The following proposition allows us to push this approach a bit further.

\begin{proposition}
  $\CSP(\mathbb Z_p)\not\leqdl\CSP(\mathbb Z_q)$ for any primes $p \neq q$.
\end{proposition}

\begin{proof}
  We start with a $k$-consistent but unsolvable instance $\rel X$ of $\CSP(\mathbb Z_p)$, i.e., with a system of equations modulo $p$ which is not solvable, but is accepted by the $k$-consistency test. Such a system is known to exist (see, e.g., \cite{FV98, ABD09}). Observe that, after each step in the $k$-consistency procedure, the sets $\mathcal F_K$ are always affine subspaces of $\mathbb Z_p^K$ (to begin with they are defined by linear equations, and in each step, we intersect with a projection of another affine subspace). Since the system is consistent, the subspaces are non-empty.

  Further observe that the maps $\pi\colon \mathcal F_K \to \mathcal F_L$ defined by restriction to $L$ are affine, and hence the size of the preimage of an element $f\in \mathcal F_L$ under $\pi$ does not depend on $f$ --- it only depends on the dimension of the kernel of $\pi$. Hence the uniform probability on $\mathcal F_K$, i.e., $\lambda\colon \mathcal F_K \to \mathbb Q$ defined as $\lambda(f) = 1/p^d$ where $d$ is the dimension of $\mathcal F_K$, solves the corresponding Sherali--Adams linear program. Since $p$ and $q$ are coprime, we can interpret the expression $1/p^d$ as an element of $\mathbb Z_q$. Consequently, this assignment is a feasible solution of the $\mathbb Z_q$-affine $k$-consistency relaxation with input~$\rel X$.
\end{proof}

We note that the above proposition could be also proved by using the fact that $\CSP(\mathbb Z_p)$ is not expressible in fixed-point logic with modulo $q$ rank operators (see, e.g., Holm~\cite{Hol11}). This alternative approach has been communicated to us by Anuj Dawar \cite{Daw22} and precedes the proof given here.

It follows from the previous proposition that $\CSP(\rel K_3) \not\leqdl \CSP(\mathbb Z_2)$. More generally, Ciardo and Živný \cite{CZ22-lp-aip} proved that, for any $c\geq 3$ and $k \geq 2$, $\PCSP(\rel K_3, \rel K_c)$ is not solved by the $k$-th level of LP+AIP hierarchy, which implies that $\CSP(\rel K_3) \not\leqdl \CSP(\mathbb Z)$.
Since $\CSP(\rel K_3)$ allows a gadget reduction from any CSP with a finite template, we can view this lack of a reduction as witnessed by $\CSP(\mathbb G)$ where $\mathbb G$ is not Abelian.

This suggests the following intriguing question: Could it be that group CSPs constitute a \emph{complete} set of obstructions for consistency reductions between finite-template CSPs? More concretely, is it true that for every finite structures $\rel A$ and $\rel B$, $\CSP(\rel A) \leqdl \CSP(\rel B)$ if and only if, for every finite group $\mathbb G$, $\CSP(\mathbb G) \leqdl \CSP(\rel A)$ implies $\CSP(\mathbb G) \leqdl \CSP(\rel B)$ or, alternatively,
does every equivalence class of finite-template CSPs up to consistency reductions contain a group CSP?
The positive verification in the particular case $\rel B = \False$ is essentially the notorious bounded width conjecture of Larose and Zádori \cite{LZ07} confirmed by Barto and Kozik \cite{BK14}.
Although extending this proof to all $\rel B$ is currently out of reach, we believe that the answer will be affirmative. In particular, we conjecture the following.

\begin{conjecture} \label{the-conjecture}
  For every finite structure $\rel A$, either $\CSP(\mathbb G) \leqdl \CSP(\rel A)$ for a non-Abelian finite group $\mathbb G$, or $\CSP(\rel A) \leqdl \CSP(\mathbb Z)$.
\end{conjecture}

The conjecture has an important consequence. If true, it would show that all tractable CSPs are solved by the $\mathbb Z$-affine $k$-consistency relaxation for some $k \geq 1$.
Combining results of Section~\ref{sec:local-reductions} and known results it follows that this algorithm solves correctly all CSPs solvable by the $k$-consistency algorithm as well as systems of linear equations over a finite Abelian group, i.e., it solves the two prime examples of tractable CSPs.

Let us note that another polynomial-time CSP algorithm coming from the algebraic approach, \emph{few subpowers}, does not compare well with consistency reductions. In particular, the class of CSPs having few subpowers (i.e., the class of CSPs that are solved by this algorithm) contains all CSPs of Abelian groups, and hence infinitely many problems that are incomparable by consistency reductions. We believe that confirming (or disproving) the conjecture in the case of CSPs with a Mal'cev polymorphism, which form a small subclass of few subpowers, would likely lead us close to a complete resolution.

We conclude by observing that several relaxation variants previously defined are at least as powerful as the $\mathbb Z$-affine $k$-consistency relaxation introduced here: the \emph{LP+AIP hierarchy} \cite{BGWZ20}, and \emph{cohomological $k$-consistency} introduced by Ó~Conghaile~\cite{O-Con22}.
The power of these relaxations is not yet well understood even in the case of CSPs. For example, the comparison of the power of cohomological $k$-consistency with respect to the hierarchy of LP+AIP is not clear.
Our conjecture implies that all variants solve exactly the same CSPs, namely, all tractable ones (assuming $\Ptime \neq \NP$).

\section{Characterisation of the arc-consistency reduction}
  \label{sec:arc-consistency}

In this section, we characterise the applicability of a special case of a \Datalog reduction. Namely, a reduction that is obtained from the $k$-consistency reduction by replacing enforcing $k$-consistency with enforcing \emph{arc-con\-sis\-ten\-cy} instead.
This characterisation can be used together with a result of Barto and Kozik \cite{BK22} to obtain a new sufficient condition for \Datalog reductions (see Appendix~\ref{app:sufficient-condition}).

We start with describing the reduction which is a combination of enforcing arc-consistency and the universal gadget reduction.

\begin{definition}[the arc-consistency reduction] \label{def:arc-consistency}
  Let $\rel A$ be a $\Pi$-struc\-ture, and $\rel B$ be a $\Sigma$-structure.  The corresponding \emph{arc-consistency reduction} applied to a $\Pi$-structure $\rel X$ outputs a $\Sigma$-structure $\rel Y$ constructed as follows:
  \begin{enumerate}
    \item For each variable $v \in X_t$ of type $t$, set $\mathcal F_v = A_t$.
    \item Ensure that for each constraint $(v_1, \dots, v_k) \in R^{\rel X}$, the sets $\mathcal F_{v_i}$ are consistent, i.e., for each $i$, update $\mathcal F_{v_i}$ to be equal to the $i$-th projection of $R^{\rel A} \cap (\mathcal F_{v_1} \times \dots \times \mathcal F_{v_k})$.
    \item Repeat (2), while anything changes.
    \item For each $v \in X$, include in $\rel Y$ a copy of $\rel B^{\mathcal F_v}$ whose elements are denoted by $(v; b)$ where $b\colon \mathcal F_v \to B$.
    \item For constraint $(v_1, \dots, v_k) \in R^{\rel X}$, include in $\rel Y$ a copy of $\rel B^{\mathcal C_{v, R}}$ where $\mathcal C_{v, R} = R^{\rel A} \cap (\mathcal F_{v_1} \times \dots \times \mathcal F_{v_k})$ whose elements are denoted by $((v, R); b)$ where $b\colon \mathcal C \to B$.
    \item For each constraint as above, each $i\in [k]$, and $b\colon \mathcal F_{v_i} \to B$, identify the elements $((v, R); b \circ \pi_i)$ and $(v_i; b)$ where $\pi_i \colon \mathcal C_{v, R} \to \mathcal F_{v_i}$ is the $i$-th projection.
  \end{enumerate}
  We denote the resulting structure $\ac^{\rel A, \rel B}(\rel X)$.
\end{definition}

In order to formulate our main result in this section we recall the basic definitions of the algebraic theory of promise CSPs \cite[Section 2]{BBKO21}. In this paper we work with abstract minions, as opposed to function minions defined in \cite[Definition 2.20]{BBKO21}.

\begin{definition} \label{def:minion}
  An \emph{(abstract) minion} is a functor $\clo M$ from finite sets to sets which preserves emptyness of sets.
  In detail, $\clo M$ is a mapping that assigns to each finite set $X$ a set $\clo M^{(X)}$, and to each function $\pi \colon X \to Y$ between two finite sets $X$, $Y$, a function $\pi^\clo M\colon \clo M^{(X)} \to \clo M^{(Y)}$, such that $1_X^\clo M = 1_{\clo M^{(X)}}$ (recall that $1_X$ denotes the identity map on $X$), and $\pi^\clo M \circ \sigma^\clo M = (\pi \circ \sigma)^\clo M$ for all $\pi$ and $\sigma$ where the composition makes sense.
  Moreover, we require that $\clo M^{(X)} = \emptyset$ if and only if $X = \emptyset$.
  When the minion is clear from the context, we write $f^\pi$ for $\pi^\clo M(f)$.

  A \emph{minion homomorphism} from $\clo M$ to $\clo N$ is a natural transformation $\xi\colon \clo M \to \clo N$, i.e., a collection of maps $\xi_X\colon \clo M^{(X)} \to \clo N^{(X)}$, one for each finite set $X$, such that $\pi^\clo N \circ \xi_X = \xi_Y \circ \pi^\clo M$ for all $\pi \colon X \to Y$.
\end{definition}

A different view on a minion homomorphism is to see it as a structure homomorphism where each minion is viewed as a typed structure, with an element $f$ of type $X$ for each $f \in \clo M^{(X)}$, and a binary relation $R_\pi = \{(f, f^\pi) \mid f\in \clo M^{(X)}\}$ for each $\pi \colon X \to Y$.

\begin{definition} \label{def:polymorphism-minion}
  The \emph{polymorphism minion} $\Pol(\rel A, \rel B)$ of a promise template $\rel A, \rel B$ is the minion $\clo M$ with $\clo M^{(X)} = \{ f \colon \rel A^X \to \rel B \}$ where $\pi^\clo M$ is defined by $f_t^\pi(x) = f_t(x \circ \pi)$ for all types $t$, $x \in X_t^N$, and $\pi\colon N \to N'$.
  Elements of $\clo M^{(X)}$ are called \emph{polymorphisms of arity~$X$}.
\end{definition}

The characterisation of the arc-consistency reduction uses minion homomorphisms together with the following transformation on minions.

\begin{definition} \label{def:omega}
  Let $\clo M$ be an abstract minion; we define a minion $\omega(\clo M)$ as follows:
  \[
    \omega(\clo M)^{(X)} = \{ (Y, f) \mid Y\subseteq X, f\in \clo M^{(Y)} \}
  \]
  The minor-taking operation for $\pi \colon X \to Y$, is defined by $(Z, f)^\pi = (\pi(Z), f^{\pi|_Z})$ where $\pi(Z) = \{\pi(x) \mid x \in Z\}$.
\end{definition}

Intuitively, we interpret the construction above as `remembering' for each element of $\clo M$ on which coordinates it depends. This influences the definition of the minor by taking into account the following argument: if $g$ only depends on coordinates in $Y$ then $g^\pi$ only depends on coordinates in $\pi(Y)$.
The algebraic intuition behind $\omega(\clo M)$ is that this minion is constructed in such a way that it satisfies precisely \emph{balanced minor conditions} (i.e., sets of identities of the form $f^\pi = g$ where $\pi$ is surjective) satisfied in $\clo M$.

We note that $\omega(\Pol(\False))$ is isomorphic to the polymorphism minion of Horn-3SAT; see~\cite[Example 4.2]{BBKO21}. This is related to the fact that arc-consistency test is characterised via gadget reductions to Horn-3SAT \cite[Theorem 7.4]{BBKO21}.

Finally, we can formulate the characterisation of arc-consistency reduction.
The proof relies on Theorem~\ref{thm:gadget-characterisation} and its proof \cite[Section 3]{BBKO21}; the novel ingredient is a connection between $\omega$ and the arc-consistency enforcement (see Lemma~\ref{lem:arc-adjoint} in the appendix).

\begin{theorem} \label{thm:unary}
  Let $\rel A, \rel A'$ and $\rel B, \rel B'$ be two promise templates. Then $\PCSP(\rel A, \rel A') \leqarc \PCSP(\rel B, \rel B')$ if and only if there is a minion homomorphism $\omega(\Pol(\rel B, \rel B')) \to \Pol(\rel A, \rel A')$.
\end{theorem}

\section{Conclusion}
  \label{sec:conclusion}

In the present paper, we introduced a general framework of well-structured polynomial-time computable reductions between (promise) CSPs. We see this theory as a foundational work for further study of the complexity of CSPs and promise CSPs. Apart from the obvious future goal of characterising \Datalog reductions by the means of some mathematical invariants of the problems considered, we enclose a few concrete interesting research directions for future work.

\subsubsection*{Hierarchies of polynomial time algorithms}

The following two questions, which have been resolved for non-promise CSPs, remain open for promise CSPs: (1) a characterisation of promise CSPs of bounded width, and (2) a characterisation of promise CSPs solvable by some level of the Sherali--Adams hierarchy. Both of these classes of problems can be described as a $k$-consistency hierarchy associated to some problem, and hence it is plausible that our approach might be useful here.

For example, we have observed that $\CSP(\mathbb G)$ is not solved by any level of the Sherali--Adams hierarchy for any non-trivial group $\mathbb G$, and hence no promise CSP that allows a \Datalog reduction from one of these problems can be solved by some level of the Sherali--Adams hierarchy. Is the converse true, i.e., is $\PCSP(\rel A, \rel B)$ solved by some level of the Sherali--Adams hierarchy if and only if $\CSP(\mathbb G) \not\leqdl \PCSP(\rel A, \rel B)$ for all non-trivial groups $\mathbb G$?

\subsubsection*{Fragments of \Datalog reductions}

We have aimed to define a class of reductions that is as expressive as possible while remaining polynomial-time computable. Nevertheless, there are fragments of our reductions that are worth studying for several reasons. Naturally, it is useful if this fragment is expressive enough to be able to emulate gadget reductions, and that it composes similarly to \Datalog reductions.

For example, we could use a log-space computable fragment of Datalog, called \emph{symmetric Datalog}, which has been introduced by Egri, Larose, and Tesson \cite{EgriLT07}, and is able to express reflexive symmetric transitive closure which is the essential part of proving that gadget reductions are expressible as \Datalog reductions. In view of the discussion in Kazda \cite[Section 8]{Kaz18}, these \emph{symmetric \Datalog reductions} could be a good approximation for all log-space reductions between (promise) CSPs.

\subsubsection*{Descriptive complexity of CSPs}

We mentioned some results on the descriptive complexity of CSPs in passing: Atserias, Bulatov, and Dawar \cite{ABD09} proved that a finite-template CSP is expressible in the \emph{fixed-point logic with counting operators} $\FPC$ if and only if it has bounded width, and Holm \cite{Hol11} showed that $\CSP(\mathbb Z_p)$ is not expressible in the \emph{fixed-point logic with rank mod $q$ operator} $\FPR_q$ for any two distinct primes $p$ and $q$. We may ask which finite-template (promise) CSPs are expressible in these logics in general?

It can be verified that if a (promise) CSP is expressible in a logic that extends Datalog, then all the promise CSPs that reduce to it via a \Datalog reduction are as well. For example, all the CSPs solvable by some level of the $\mathbb Z_2$-affine $k$-consistency relaxation are expressible in $\FPR_2$. Is the converse true? Could the converse be true in the scope of promise CSPs? Naturally, we can ask the same questions for any prime in place of $2$.

\appendix
\section{Preliminaries to the appendices}
  \label{app:preliminaries}

We recall some definitions that will be useful in the following appendices.

\subsection{Label cover}
  \label{sec:label-cover}

The label cover problem will play a key role in this paper, since many reductions will be either implicitly, or explicitly reductions between variants of label cover and other CSPs. The \emph{label cover problem} is a binary CSP where each constraint is of the form $\pi(u) = v$ for some mapping $\pi \colon D_u \to D_v$. We can view this problem as a fixed-template CSP using the multi-sorted approach although we have to allow an infinite signature.

\begin{definition}(Label cover problem) \label{def:label-cover}
  The signature of label cover has a type $X$ for each finite set $X$. Further, for each pair of finite sets $X$, $Y$, and every mapping $\pi\colon X \to Y$ it has a binary relation symbol $E_\pi$ with arity $(X, Y)$.

  The \emph{label cover problem} is $\CSP(\rel P)$ where its template $\rel P$ is the structure with $P_X = X$ for each finite set $X$ and $E_\pi^{\rel P} = \{(x, \pi(x)) \mid x\in X\}$ for each $\pi \colon X \to Y$.
\end{definition}

We refine the definition of finite structure to account for infinite signatures, and encoding of a label cover instance in a finite amount of space. We say that a structure $\rel D$ in the above signature is \emph{finite} if $D_X$ is non-empty and finite for finitely many types $X$, and empty for all other types. Note that this condition implies that $E_\pi^{\rel D} \neq \emptyset$ for finitely many relational symbols $E_\pi$.
Such a structure can be then encoded as a finite list of all its elements and a finite list of all its non-empty relations. We will call such structures \emph{label cover instances}.

We say that a signature $\Sigma$ is a \emph{finite reduct} of the label cover signature if the types of $\Sigma$ consist of a finite list of label cover types, i.e., a finite list of sets, and relations of $\Sigma$ consist of finitely many relational symbols of label cover. We say that a $\Sigma$-structure $\rel A$ is a \emph{$\Sigma$-reduct} of $\rel B$ if $A_X = B_X$ for each $\Sigma$-type $X$, and $E_\pi^{\rel A} = E_\pi^{\rel B}$ for each $\Sigma$-relation $E_\pi$.
In all cases our reductions produce a finite reduct of label cover whose signature depends on the original template. We use the full signature for convenience.

Finally, let us mention that a promise version of (a reduct of) the label cover problem was used to provide the generalisation of the algebraic approach to promise CSP, see \cite[Definition 3.10 \& Theorem 3.12]{BBKO21}. We will also use similarly defined promise versions of label cover in the present paper, but postpone the definition until Appendix~\ref{app:ac-proof}.

\subsection{Gadget reductions}
 \label{app:gadget-reductions}

Gadget reductions are the reductions that were classified by the algebraic approach. Note that, customarily, the algebraic approach focuses on the templates. That is, it aims to explore under which conditions on $\rel A$ and $\rel B$ there exists a gadget reduction from $\CSP(\rel B)$ to $\CSP(\rel A)$. This relationship between $\rel A$ and $\rel B$ is explained using \emph{pp-interpretations} (which are a special case of Datalog interpretations). However, in the present paper, our focus shifts to understanding  \emph{what the reduction does to an instance $\rel X$ of $\CSP(\rel B)$?}
We start with giving a formal definition of a gadget. Informally, a gadget reduction replaces each constraint of the input structure with several constraints of the output structure. These new constraints overlap depending on which variables appear in the original constraints and on what position. Formally, we replace each constraint with a structure, and we introduce new constrains also by replacing variables (again encoded in a structure); we encode a recipe for overlapping of the constraints by providing homomorphisms between these structures.

\begin{definition} \label{def:gadget}
  Let $\Pi$ and $\Sigma$ be relational signatures. A \emph{gadget} $\gamma$ mapping $\Pi$-structures to $\Sigma$-structures has three ingredients: a $\Sigma$-structure $\rel D_t^\gamma$ for each $\Pi$-type $t$, a $\Sigma$-structure $\rel R^\gamma$ for each $\Pi$-symbol $R$, and a homomorphism $p_{R, i}^\gamma \colon \rel D_{\ar_R(i)}^\gamma \to \rel R^\gamma$ for each $\Pi$-symbol $R$ of arity $k$ and $i \in [k]$.

  A gadget $\gamma$ induces a function, which we call \emph{gadget replacement} and denote by the same symbol $\gamma$, mapping every $\Pi$-structure $\rel A$ to a $\Sigma$-structure $\gamma(\rel A)$ defined as follows:
  \begin{enumerate}
    \item For each $a\in A_t$ of type $t$, introduce to $\gamma(\rel A)$ a copy of $\rel D_t^\gamma$, whose elements will be denoted as $(a; d)$ where $d\in \rel D_t^\gamma$.
    \item For each $R$ and $a \in R^{\rel A}$, introduce to $\gamma(\rel A)$ a copy of $\rel R^\gamma$, whose elements will be denoted as $(a; e)$ for $e\in \rel R^\gamma$.
    \item For each $R$ of arity $k$, $a\in R^{\rel A}$, $i \in [k]$, and $e\in \rel D_{\ar_R(i)}^\gamma$, add an equality constraint $(a; p_{R,i}(e)) = (a_i; e)$.
    \item Collapse all equality constraints (i.e., identify all pairs of elements involved in one of the constraints introduced in the previous step).
    We denote by $[a; d]$ the class of an element $(a; d)$ after collapsing.
  \end{enumerate}
\end{definition}

It is well-known that such a gadget replacement can be computed in log-space (the fact that collapsing equality constraints is in log-space is due to \citet{Rei08}).

\begin{example} \label{ex:manuel-gadget}
  Gadget reductions were extensively used by Hell and Nešetřil \cite{HN90} to provide a dichotomy for CSPs of graphs. Let us give as an example a gadget reduction from $\CSP(\rel K_5)$ to $\CSP(\rel C_5)$, where $\rel C_5$ denotes the unoriented 5-cycle. In words, this reduction is: Replace each edge on the input with a path of length 3 connecting the two original vertices.

  Formally, both signatures $\Pi$ and $\Sigma$ are signatures with one binary relation. We ignore issues with the orientation as they are irrelevant for the two CSPs. To define a gadget $\gamma$, we need two graphs $\rel D_t^\gamma$ (for the unique type $t$) and $\rel E^\gamma$ and two homomorphisms $p^\gamma_{E, 0}$ and $p^\gamma_{E, 1}$; for simplicity we will omit $\gamma$, $t$, and $E$ from the indices.
  The graph $\rel D$ is defined by what happens to the vertices. Since we keep the original vertices, $\rel D$ is a singleton graph with no edges, i.e., $\rel D = (\{{*}\}; \emptyset)$.
  The graph $\rel E$ is defined by what happens to the edges, hence $\rel E = \rel P_3$ where $\rel P_3 = (\{0, 1, 2, 3\}; \{(0, 1), (1, 2), (2, 3)\})$ is a path of length 3. The homomorphisms $p_1$ and $p_2\colon \rel D \to \rel P_3$ then define how these gadgets are pasted together. We have $p_1({*}) = 0$ and $p_2({*}) = 3$.

  The fact that $\gamma$ is a valid reduction between the two problems is proven in the same way as in \cite[Sec.~A]{HN90}.
\end{example}

\begin{example} \label{ex:strict-gadget}
  A slightly more complicated example is the following gadget reduction from $\CSP(\rel K_2)$ to $\CSP(\rel K_\infty)$ where $\rel K_\infty$ is the countable clique.\footnote{In this and connected examples, we are deviating from the convention that requires that the template is finite. Also, we treat graphs as relational structures with one binary relation, in particular the instances of $\CSP(\rel K_2)$ can be directed graphs with loops, though the orientation of edges does not matter for this example.}
  Since we are defining a graph to graph gadget $\gamma$, we need to provide two graphs $\rel D^\gamma$ and $\rel E^\gamma$ together with two homomorphisms 
  \[
    p_{E, 1}^\gamma, p_{E, 2}^\gamma \colon \rel D^\gamma \to \rel E^\gamma.
  \]
  We let $\rel D^\gamma = \rel E^\gamma = \rel K_2 = (\{0, 1\}; {\neq})$, and define $p_{E, 1}^\gamma, p_{E,2}^\gamma$ to be the two distinct automorphisms of $\rel K_2$, e.g., $p_{E,1}^\gamma(x) = x$ and $p_{E,2}^\gamma(x) = 1 - x$.

  The gadget replacement $\gamma$ then produces from a graph $\rel G$ another graph $\rel H$ in the following way:
  \begin{enumerate}
    \item Replace each vertex $v$ of $\rel G$ with a pair of vertices $(v; 0)$ and $(v; 1)$, which we will simply denote by $v_0$ and $v_1$, connected by an edge, i.e., replace each vertex with a copy $\rel K_2$.
    \item Replace each edge $e = (u, v)$ of $\rel G$ with an edge between two new vertices $e_0$ and $e_1$ (again, we write $e_i$ instead of $(e; i)$).  Introduce equality constraints $e_0 = u_0 = v_1$ and $e_1 = u_1 = v_0$, essentially identify the edges $(e_0, e_1)$, $(u_0, u_1)$ and $(v_1, v_0)$, in particular we identify the two edges introduced by replacing the two vertices \emph{in the reverse orientation}.
    \item Collapse all equality constraints as in Definition~\ref{def:gadget}.
  \end{enumerate}
  This gadget produces for each connected component of the input $\rel G$ either a loop, if the component contains an odd cycle, or an edge, if the component is bipartite. Armed with this observation, it is not hard to check that $\gamma$ is a valid reduction from $\CSP(\rel K_2)$ to $\CSP(\rel K_\infty)$.
\end{example}

An important example for us is the ability of gadget replacements to emulate disjoint unions (see Definition~\ref{def:union-gadget}).

\begin{example}[A union gadget] \label{ex:unions-as-gadget}
  We describe an example of a gadget that expresses disjoint unions. Let us for example consider a signature $\Pi$ consisting of several types $t, s, \dots$ and several binary relational symbols $R, S, \dots$. We construct a gadget $\gamma$ from $\Pi$ into digraphs (i.e., the signature with a single type and single binary relation $E$) whose application can be loosely described as `forgetting the types and names of relations':
  Let $\rel D_t^\gamma$ be the digraph with 1 vertex (element) and no edges for each $\Pi$-type $t$, and let $\rel R^\gamma$ be the digraph consisting of a single oriented edge, i.e., $\rel R^\gamma = (\{0, 1\}; \{(0, 1)\})$, for each $\Pi$-symbol $R$, and homomorphisms $p_{R,1}^\upsilon$ and $p_{R, 2}^\upsilon$ map the unique vertex of $\rel D_{\ar_R(1)}$ and $\rel D_{\ar_R(2)}$, respectively, to the initial and terminal vertex of the edge.

  Applying this gadget on a $\Pi$-structure $\rel X$ then yields a digraph whose domain (respectively edge-set) is the disjoint union of all domains (respectively relations) of $\rel X$. More concretely, the domain of $\gamma(\rel X)$ is $\{ (x; t) \mid x\in X_t \}$ and two vertices $(x; t)$ and $(y; s)$ are related in $\gamma(\rel X)$ by an edge if there is a $\Pi$-symbol $R$ with $ar_R(1) = t$ and $\ar_R(2) = s$ such that $(x, y) \in R^{\rel X}$.

  A similar construction works for any union gadget.
\end{example}

\begin{example} \label{ex:incidence-graph}
  Another interesting example of a gadget replacement is so-called \emph{incidence graph}. The incidence graph of a structure $\rel X$ is defined as a directed graph whose vertices are the disjoint union of all the domains of $\rel X$ together with all the relations of $\rel X$, and there is an edge from $y$ to $x$ if $y\in R^{\rel X}$ for some $R$, and $x = y(i)$ for some $i$. It is defined by the gadget composed of the following structures: $\rel D_i$ is the singleton graph with no edges, i.e., $\rel D_i = (\{v\}; \emptyset)$ for all types $i$. And, for a relational symbol $R$ of arity $k$, $\rel D_R$ is the graph with $k+1$ vertices $v_1, \dots, v_k, r$, and edges $(r, v_1)$, \dots, $(r, v_k)$. Further $p_{R,j}\colon \rel D_{i_j} \to \rel D_R$ is defined by $p_{R,j}(v) = v_j$.
\end{example}

In many of our proofs, it will be useful to decompose a reduction into a composition of several simpler reductions. In particular, we decompose a gadget reduction into two steps: the first step produces a structure with binary relations, or a label cover instance and can be easily described as a Datalog interpretation, the second step then places a restriction on the gadget that makes our subsequent proofs easier.
The first step has two slightly different variants --- the first one, which we call \emph{reification} and denote by $\rho$, does not depend on the template, and the second one, which we denote by $\rho^{\rel A}$, produces a label cover instance which is essentially equivalent to the construction $\Sigma(\rel A, \rel X)$ from \cite[Section 3.2]{BBKO21}.

\begin{definition} \label{def:reification}
  Assume $\rel X = (X_t, X_s, \dots; R^{\rel X}, S^{\rel X}, \dots)$ is a $\Pi$-struc\-tu\-re. The \emph{reification} of $\rel X$ is a structure $\rho(\rel X)$ constructed as follows.
  The signature $\Pi^*$ of $\rho(\rel X)$ consists of all the types of $\Pi$ together with a type $R$ for each $\Pi$-relation $R$, and a binary relation $P_{R,i}$ for each $\Pi$-symbol $R$ of arity $k$ and $i \in [k]$.
  We define $\rho(\rel X)$ as
  \[
    \rho(\rel X) = (X_t, X_s, \dots, R^{\rel X}, S^{\rel X}, \dots; P_{R,i}^{\rho(\rel X)}, \dots)
  \]
  where
  \[
    P_{R,i}^{\rho(\rel X)} =
      \{ ((a_1, \dots, a_k), a_i) \mid (a_1, \dots, a_k)\in R^{\rel X} \}
  \]
  for each $\Pi$-symbol $R$ of arity $k$ and $i\in [k]$.
\end{definition}

Reification is expressed as a Datalog interpretation $\rho$ as follows:
The domains are defined in the obvious way, i.e., the $t$-th domain is defined by the program $\rho_t$ consisting of a single rule $D_t(x) \fro x = x$ where $x$ is of type $t$ for each $\Pi$-type $t$, and the $R$-th domain is defined by the program $\rho_R$ consisting of a single rule $D_R(x_1, \dots, x_k) \fro R(x_1, \dots, x_k)$ for each $\Pi$-symbol $R$ (which is then a $\Pi^*$-type).
A relation $P_{R,i}$ is then defined by the program $\rho_{P_{R,i}}$ with a single rule
\[
  P_{R, i} (x_1, \dots, x_k, x_i) \fro R(x_1, \dots, x_k).
\]

\begin{definition} \label{def:reification-to-lc}
  Let $\Pi$ be a signature, and fix a $\Pi$-structure $\rel A$. We define a transformation $\rho^\rel A$ that maps every $\Pi$-structure $\rel X$ to a label cover instance $\rho^{\rel A}(\rel X)$ in a similar way as in Definition~\ref{def:reification} except that we rename types and symbols of $\Pi^*$ to get a signature, which is a reduct of the label cover signature:
  We use the type $A_t$ instead of $t$ and the type $R^{\rel A}$ instead of $R$, and a relation $E_\pi$, where $\pi\colon R^{\rel A} \to A_{\ar_R(i)}$ is the $i$-th projection, instead of $P_{R,i}$.
\end{definition}

Clearly, this other flavour of reification can be expressed by the Datalog program obtained from the Datalog program for reification by substituting symbols.

The restriction on the gadgets in the second step is given by the following definition.

\begin{definition} \label{def:strict-gadget}
  Let $\Sigma$ and $\Pi$ be relational signatures, and assume that each $\Pi$-symbol is binary.
  A \emph{projective gadget} $\gamma$ (of $\Sigma$ into $\Pi$) has two ingredients: a $\Sigma$-structure $\rel D_t$ for each $\Pi$-type $t$, and a homomorphism $p_R\colon \rel D_s \to \rel D_t$ for each $\Pi$-symbol $R$ with $\ar_R = (t, s)$ --- note that the order of $t$ and $s$ is reversed.
  In the same way as for a general gadget, a projective gadget $\gamma$ induces a function, which we will call \emph{(projective) gadget replacement} and denote by the same symbol $\gamma$, mapping every $\Pi$-structure $\rel A$ to a $\Sigma$-structure $\gamma(\rel A)$ defined as follows:
  \begin{enumerate}
    \item For each $a\in A_t$ introduce to $\gamma(\rel A)$ a copy of $\rel D_t$, whose elements will be denoted as $(a; d)$ for $d\in \rel D_t$.
    \item For each $R$ with $\ar_R = (t, s)$, $(a, b)\in R^{\rel A}$, and every $d\in \rel D_s$, add an equality constraint $(a; p_R(d)) = (b; d)$.
    \item Collapse all equality constraints (i.e., identify all pairs of elements involved in one of the constraints introduced in the previous step).
  \end{enumerate}
\end{definition}

\begin{example} \label{ex:strict-gadget-ii}
  The gadget from Example~\ref{ex:strict-gadget} is also expressible as a projective gadget: This projective gadget is defined by a graph $\rel D_1 = \rel K_2 = (\{0, 1\}; {\neq})$, and a map $p_E\colon \rel K_2 \to \rel K_2$ that switches $0$ and $1$ (the non-trivial automorphism of $\rel K_2$).
\end{example}

As we mentioned before, we can decompose each gadget replacement as a reification followed by a projective gadget.

\begin{lemma} \label{lem:projective-decomposition}
  For each gadget $\gamma$, there is a projective gadget $\gamma'$ such that $\gamma(\rel X)$ and $\gamma'\rho(\rel X)$ are isomorphic for all $\rel X$.
\end{lemma}

\begin{proof}
  Let $\Pi$ be the input signature of $\gamma$. The projective gadget $\gamma'$ is defined by the structures $\rel D_t^{\gamma'} = \rel D_t^\gamma$ for each $\Pi$-type $t$ and $\rel D_R^{\gamma'} = \rel R^\gamma$ for each $\Pi$-symbol $R$, together with maps $p_{P_{R, i}}^{\gamma'} = p_{R, i}^\gamma$ for each symbol $R$ and $i$.
  It is straightforward to check that indeed $\gamma'\rho(\rel X)$ and $\gamma(\rel X)$ are isomorphic for all $\rel X$.
\end{proof}

\begin{example} \label{ex:manuel-projective}
  Following on Example~\ref{ex:manuel-gadget}, we convert the gadget described above into a projective gadget. Reification converts a graph $(V; E)$ into a 2-sorted structure $(V, E; S, T)$ where $S, T \subseteq V \times E$ is defined as $S = \{ (u, (u, v)) \mid (u, v) \in E \}$ and $T = \{ (v, (u, v)) \mid (u, v) \in E \}$. We will denote the types of this signature by $v$ and $e$. The projective gadget $\gamma$ is then defined as $\rel D^\gamma_v = (\{*\}; \emptyset)$, $\rel D^\gamma_e = \rel P_3$ where $\rel P_3$ is the path of length $3$, i.e., $\rel P_3 = (\{0, 1, 2, 3\}; \{(i, i+1) \mid i = 0, 1, 2\})$, and $p^\gamma_S({*}) = 0$, $p^\gamma_T({*}) = 3$.
\end{example}

The general theory of gadget reductions also provides the `best projective gadget' from label cover to any other signature. 
It is used to reduce from label cover to any (promise) CSP \cite[Theorem 3.12]{BBKO21}, and it constructs the so-called indicator structure for a \emph{minor condition} denoted by $\rel I_\Sigma(\rel B)$ in \cite[Section 3.3]{BBKO21}. Here we will apply it on label cover instances instead of minor conditions and denote it by $\pi_{\rel B}$. The name comes from a universal property of this gadget, which we explicitly prove in Lemma~\ref{lem:universality-of-gadget}.

\begin{definition}[Universal gadget] \label{def:universal-gadget}
  Let $\rel B$ be a $\Sigma$-structure. The \emph{universal gadget for $\rel B$} is the projective gadget $\pi_\rel B$ from the label cover signature to $\Sigma$ defined as follows: $\rel D_X^{\pi_{\rel B}} = \rel B^X$ for all types $X$, and $p_{E_\sigma}^{\pi_{\rel B}}(b) = b \circ \sigma$ for each $\sigma\colon X \to Y$.\footnote{This gadget is the contravariant functor $X \mapsto \rel B^X$.}
\end{definition}

\subsection{Elements of consistency reductions}
  \label{sec:appendix-consistency}

As with the gadgets above, it will be useful to decompose consistency reductions into smaller steps. We can describe both $k$-consistency reduction and arc-consistency reduction as a three step process: (1) a reduction from $\PCSP(\rel A, {*})$ to label cover, (2) a transformation of a label cover instance (this step enforces consistency), and (3) encoding the resulting label cover as an instance of $\PCSP(\rel B, {*})$ using the universal gadget. The difference between arc-consistency and $k$-consistency is only in the first step. Let us nevertheless start with step (2) which is essentially just arc-consistency run on label cover instances (see Definition~\ref{def:arc-consistency}).

\begin{definition} \label{def:ac-on-lc}
  Let $\rel S$ be a label cover instance.
  We construct $\ac(\rel S)$ as a fixed-point:
  \begin{enumerate}
    \item For each variable $v\in \rel S$, set $\mathcal F_v = D_v$ where $D_v$ is the type of~$v$.
    \item Ensure that, for each constraint $(v, w) \in E_\pi^{\rel S}$, the sets $\mathcal F_v$ and $\mathcal F_w$ are consistent, i.e., update $\mathcal F_w$ by removing all elements that are not in the image of $\mathcal F_v$ under $\pi$, and update $\mathcal F_v$ by removing all elements that do not map to $\mathcal F_w$ by $\pi$.
    \item Repeat (2) while anything changes.
    \item Output the label cover instance with a variable $v$ of type $\mathcal F_v$ for each $v \in \rel S$, and a constraint $(v, w) \in E_{\pi'}^{\ac(\rel S)}$ for each constraint $(v, w) \in E_\pi^{\rel S}$ where $\pi' = \pi|_{\mathcal F_v}$; note that the image of the restriction is a subset of $\mathcal F_w$ due to (2).
  \end{enumerate}
\end{definition}

Observe that the arc-consistency reduction is simply a composition of $\rho^\rel A$, $\ac$, and $\pi_\rel B$, i.e., that
\begin{equation}
  \ac^{\rel A, \rel B} = \pi_\rel B \circ \ac \circ \rho^{\rel A}.
\end{equation}
Step (1) of Definition~\ref{def:arc-consistency} loosely corresponds to $\rho^\rel A$, Steps (2)--(3) to $\ac$, and steps (4)--(6) to $\pi_\rel B$.
The $k$-consistency is expressed similarly, except we change the first step to the following construction.

\begin{definition}
  Fix a $\Pi$-structure $\rel A$ and $k \geq 1$. For each $\rel X$ construct a label cover instance $\sigma_k^{\rel A}(\rel X)$ as follows:
  \begin{enumerate}
    \item Introduce a variable $v_K$ for each $K \in \binom X{\leq k}$ with type (domain) $\mathcal D_K$ which is the set of all partial homomorphisms from $K$ to $\rel A$.
    \item For each $L \subset K \in \binom X{\leq k}$, add a constraint $(v_K, v_L) \in E_\pi$ where $\pi \colon \mathcal D_K \to \mathcal D_L$ is the restriction map, i.e., $\pi(f) = f|_L$.
  \end{enumerate}
\end{definition}

Again, it is not hard to observe that
\begin{equation}
  \kappa_k^{\rel A, \rel B} = \pi_\rel B \circ \ac \circ \sigma_k^{\rel A}.
\end{equation}
The three constructions in this composition, $\pi_\rel B$, $\ac$, and $\sigma_k^{\rel A}$ correspond loosely to the steps in the $k$-consistency reduction (Definition~\ref{alg:k-consistency-reduction}) as follows: Step (1) corresponds to $\sigma^\rel A_k$, Steps (2) and (3) to $\ac$, and Steps (4) and (5) to $\pi_\rel B$. Occasionally, it will be useful to also talk about the composition $\ac\circ \sigma_k^{\rel A}$ which we denote by $\kappa_k^{\rel A}$.

\section{Proofs of Theorems~\ref{thm:ddatalog-composes}, \ref{thm:gadget-is-ddatalog}, and \ref{thm:canonical-width}}
  \label{app:proofs-iii}

In this section, we prove three core theorems of the general theory of \Datalog reductions presented in Section~\ref{sec:local-reductions}.

\subsection{\Datalog reductions compose}
  \label{sec:datalog-composes}

We prove Theorem~\ref{thm:ddatalog-composes}, which claims that \Datalog reductions compose, by showing that both Datalog interpretations and union gadgets compose, and that a union gadget and a Datalog interpretation can be permuted in the following sense.

\begin{lemma} \label{lem:cases}
  \begin{enumerate}
    \item Let $\upsilon$ and $\upsilon'$ be union gadgets such that the output signature of $\upsilon$ and the input signature of $\upsilon'$ coincide. There is a union gadget $\mu$ such that, for all structures $\rel A$, $\mu(\rel A)$ and $\upsilon'\upsilon(\rel A)$ are isomorphic.
    \item Let $\phi$ and $\phi'$ be Datalog interpretations such that the output signature of $\phi$ and the input signature of $\phi'$ coincide. There is a Datalog interpretation $\psi$ such that, for all structures $\rel A$, $\psi(\rel A)$ and $\phi'\phi(\rel A)$ are isomorphic.
    \item Let $\upsilon$ be a union gadget and $\phi$ be a Datalog interpretation such that the output signature of $\upsilon$ and the input signature of $\phi$ coincide. There exist a union gadget $\upsilon'$ and a Datalog interpretation $\phi'$ such that, for all structures $\rel A$, $\upsilon'\phi'(\rel A)$ and $\phi\upsilon(\rel A)$ are isomorphic.
  \end{enumerate}
\end{lemma}

\begin{proof}
  \begin{enumerate}
    \item
      Let $\upsilon$ be defined by $(d, r)$ and $\upsilon'$ by $(d',r')$. It is immediate that the union gadget defined by $(d'\circ d,r'\circ r)$ gives isomorphic outputs to $\upsilon'\circ\upsilon$.

    \item
      It is well-known that Datalog programs are closed under composition. We sketch how to extend it to Datalog interpretations. We construct a new Datalog interpretation $\psi$ as follows. For every Datalog program $\phi'_j$ in $\phi'$ we include a Datalog program $\psi_j$ in $\psi$ obtained from $\phi'_j$ and $\phi$ as follows: Start with the union of all programs in $\phi$. For each symbol $R$ of arity $i_1\dots i_k$ in $\phi'_j$, include in $\psi_j$ a new symbol $R'$ whose arity is the concatenation of arities of $\phi_{i_1}$, \dots, $\phi_{i_k}$. Then, for each rule $r$ of $\phi'_j$, include in $\psi_j$ a rule $r'$ obtained from $r$ as follows. First, in every atomic formula in $r$ we replace its predicate $R$ by $R'$ and replace every variable $x$ occurring in it by a tuple $(x_{j_1}, \dots, x_{j_k})$ of fresh variables where $i$ is the type of $x$, $j_1\dots j_k = \ar_{S_i}$, and $S_i$ is the output symbol of $\phi_i$. It is immediate to verify that $\psi$ and $\phi'\circ\phi$ give isomorphic outputs.

    \item
      Let $\upsilon = (d, r)$ be a union gadget that maps $\Pi$-structures to $\Delta$-structures, and let $\phi$ be a Datalog interpretation that maps $\Delta$-structures to $\Sigma$-structures. 

      We start with briefly sketching the idea of the proof. Generally, we would want to change the input signature of $\phi$ from $\Delta$ to $\Pi$ by adding into $\phi$ a rule
      \[
        S(x_1, \dots, x_k) \fro R(x_1, \dots, x_k)
      \]
      for each $\Delta$-symbol $S$ and $\Pi$-symbol $R$ such that $S = r(R)$. These rules obviously compute the union of $R$'s with $r(R) = S$.
      Nevertheless, there is a formal problem with the rule above: the variables $x_i$ present in the rule are not properly typed (the type of $x_i$ on the left hand side is $\ar_R(i)$ but for distinct $R$ and $T$ with $r(R) = r(T)$, we could have $\ar_R(i) \neq \ar_T(i)$).
      We resolve this problem by adding a copy $S^t$ of each relational symbol $S$ appearing in $\phi$ for each tuple $t$ of $\Pi$-types such that $d\circ t = \ar_S$ (i.e., $d(t_i) = \ar_S(i)$ for all $i$).
      Naturally, we will need to change the rules of $\phi$ accordingly. Moreover, the output predicate of a component $\phi_*$ of $\phi$ (where $*$ stands either for a $\Sigma$-type $t$, or a $\Sigma$-symbol $R$) is split to several copies as well. Consequently, we change the output signature of $\phi'$ by creating several copies of each type and each relational symbol, one for each new copy $S^t$ of the output predicate $S$ of the corresponding component of $\phi$. The new union gadget $\upsilon'$ then unites all these copies into one.

      More precisely, let us start with defining the output signature of $\phi'$ which we denote by $\Sigma'$. This signature has a type $s^t$ for each $\Sigma$-type $s$, and every tuple $q$ of $\Pi$ types such that the arity tuple of (the output predicate of) $\phi_s$ is equal to $d \circ q$, and similarly, $\Sigma'$ has a symbol $S^q$ for each $\Sigma$-symbol $S$ and a tuple $q$ such that the arity tuple of $\phi_S$ is equal to $d \circ q$.
      The components of the interpretation are then defined as follows: Let $s^q$ be a $\Sigma'$-type, i.e., $s$ is a $\Sigma$-type and $q$ is a tuple of $\Pi$-types satisfying the above condition.
      The program $\phi_{s^q}'$ is defined from $\phi_s$: Replace each symbol $S$ in $\phi_s$ by symbols of the form $S^t$ where $t$ ranges over all the tuples of $\Pi$-types such that $d\circ t = \ar_S$.
      Replace each rule in $\phi_s$ with variables $x_1, \dots, x_n$ of $\Delta$-types $p_1, \dots, p_n$ with a copy of this rule for each $\Pi$-types $t_1, \dots, t_n$ with $d(t_i) = p_i$ for all $i\in [n]$ where each variable $x_i$ is given the $\Pi$-type $t_i$ (instead of $p_i$), and each atomic formula $S(x_{i_1}, \dots, x_{i_m})$ is replaced by $S^{(t_{i_1}, \dots, t_{i_m})}(x_{i_1}, \dots, x_{i_m})$. Finally, add the rules of the form
      \[
        S^{\ar_R}(x_1, \dots, x_k) \fro R(x_1, \dots, x_k)
      \]
      for each $\Delta$-symbol $S$ and each $\Pi$-symbol $R$ such that $S = r(R)$. Note that in this rule, we have $d \circ \ar_R = \ar_S$ since this is required by the definition of a union gadget. Finally, we designate $D_s^q$ to be the output predicate of $\phi_{s^q}'$ where $D_s$ is the output predicate of $\phi_s$.
      The programs $\phi'_{S^q}$ where $S^q$ is a $\Sigma'$-symbol are defined by an analogous transformation of $\phi_S$.

      Finally, the union gadget $\upsilon'$ is defined by $(d', r')$ where $d'(s^t) = s$ and $r'(S^t) = S$. It is immediate to verify that $\upsilon' \circ \phi'$ and $\phi \circ \upsilon$ give isomorphic outputs.
    \qedhere
  \end{enumerate}
\end{proof}

\begin{example} \label{ex:switching}
  To bring a little light to the proof of case (3), let us describe the construction of $\phi'$ and $\upsilon'$ in one particular case which is not fully general.
  We consider a union gadget $\upsilon$ falling into the scheme of Example~\ref{ex:unions-as-gadget}, namely, a gadget whose input signature has two types named simply as $0$ and $1$, and four binary relations $R, S, T, U$ of arities $(0, 0)$, $(0, 1)$, $(1, 0)$, and $(1, 1)$, respectively. The output are digraphs, i.e., a signature with one type $v$ and one binary relation $E$. The mappings $d$ and $r$ then map everything to $v$ and $E$, respectively.
  For the Datalog interpretation $\phi$ let us take a rather trivial program from digraphs to the signature of the trivial CSP (i.e., with no types and a single nullary symbol $C$). The interpretation is given by the program $\phi_C$ with the rule
  \[
    C \fro E(x, x).
  \]

  The Datalog interpretation $\phi'$ has again just one component $\phi'_C$ (since there is only one tuple $t$ of $\Sigma$ types such that $d\circ t$ is the empty tuple). The rules of $\phi'_C$ are
  \begin{align*}
    E^{0, 0}(x, y) &\fro R(x, y) & E^{0, 1}(x, y) &\fro S(x, y) \\
    E^{1, 0}(x, y) &\fro T(x, y) & E^{1, 1}(x, y) &\fro U(x, y) \\
    C &\fro E^{0, 0}(x, x) & C &\fro E^{1, 1}(x, x)
  \end{align*}
  Naturally, this program simplifies to
  \[
    C \fro R(x, x) \qquad C \fro U(x, x).
  \]
  The union gadget $\upsilon'$ is trivial in this case, i.e., both $d$ and $r$ are identity maps, since the output signature of $\phi'$ agrees with the output signature of $\phi$.
\end{example}

Since a \Datalog reduction is defined as a composition of a Datalog interpretation with a union gadget, Theorem~\ref{thm:ddatalog-composes} is a direct corollary of Lemma~\ref{lem:cases}.

\subsection{Gadget reductions are \Datalog reductions}
  \label{sec:gadget-is-datalog}

The goal of this subsection is to prove Theorem~\ref{thm:gadget-is-ddatalog} which claims that every gadget replacement is also expressible as a Datalog interpretation.
Before we delve into the proof, let us present two examples of \Datalog reductions constructed from gadget reductions.

\begin{example} \label{ex:manuel}
  In this example, we express the gadget reduction $\gamma$ from $\CSP(\rel K_5)$ to $\CSP(\rel C_5)$ given in Example~\ref{ex:manuel-gadget} as a \Datalog reduction.

  Recall that the gadget reduction replaces each edge on the input with a path of length 3, introducing two new vertices into the instance. To emulate introducing two new vertices, we use disjoint unions: the new domain is the disjoint union $V \cup E \cup E$. This means that the output signature of our Datalog interpretation will have three types, let us call them $v$, $1$, and $2$, and three binary relations $S, Q, T$ with arities $(v, 1)$, $(1, 2)$, and $(2, v)$, respectively.

  We then define a Datalog interpretation $\phi$ as follows:
  \begin{align*}
    D_v(x) &\fro x = x &
    S(x, x, y)  &\fro E(x, y) \\
    D_1(x, y) &\fro E(x, y) &
    Q(x, y, x, y)  &\fro E(x, y) \\
    D_2(x, y) &\fro E(x, y) &
    T(x, y, y)  &\fro E(x, y)
  \end{align*}
  Given a graph $\rel G = (V^{\rel G}; E^{\rel G})$, the structure $\phi(\rel G)$ is then
  \[
    (V^{\rel G}, E^{\rel G}, E^{\rel G}; S^{\phi(\rel G)}, Q^{\phi(\rel G)}, T^{\phi(\rel G)})
  \]
  where $S \subseteq V \times E$ consists of pairs $(u, e)$ where $u$ is the initial vertex of $e$, and analogously $T \subseteq E \times V$ consists of pairs $(e, v)$ where $v$ is the terminal vertex of $e$. Finally, the relation $Q \subseteq E \times E$ is the equality relation on $E$. Note that $S$ relates the $V^{\rel G}$ with the first copy of $E^{\rel G}$, $Q$ relates both copies of $E^{\rel G}$, and $T$ relates the second copy of $E^{\rel G}$ with $V^{\rel G}$.

  If we were to draw the resulting structure on a piece of paper, we would obtain the same graph as via the gadget reduction described above --- the only difference is that the vertices are labelled by their types: $v, 1, 2$, and edges are labelled by $S, Q, T$. The union gadget allows us to `forget this labelling'.
  Since the signature of graphs is single-sorted and has only one binary relation, there is only one way to define the union gadget $(d, r)$: $d(v) = d(1) = d(2)$ is the type of vertices, and $r(S) = r(Q) = r(T) = E$. (This is a special case of the union gadget from Example~\ref{ex:unions-as-gadget}.)
  Observe that the composition of $\phi$ and this union gadget results in the same graph as the gadget reduction.

  In this case, we managed to construct a Datalog interpretation $\phi$ such that $\phi(\rel G)$ and $\gamma(\rel G)$ are isomorphic for all graphs $\rel G$.
\end{example}

More complicated Datalog reductions are necessary in cases when the mappings $p_{i, R}^\gamma$ are not injective or with overlapping images, e.g., the following example.

\begin{example} \label{ex:k2-to-kinfty-datalog-1}
  Recall the gadget reduction $\gamma$ from $\CSP(\rel K_2)$ to $\CSP(\rel K_\infty)$ from Examples~\ref{ex:strict-gadget} and \ref{ex:strict-gadget-ii}. Here, we express an equivalent \Datalog reduction.

  We start with describing a Datalog interpretation $\phi$.
  The intermediate signature (output signature of $\phi$) has two types, $0$ and $1$, which mirror elements of the gadget $\rel K_2$. We define the domains by programs
  \[
    D_0(x) \fro x = x \qquad D_1(x) \fro x = x.
  \]
  Furthermore, we introduce four relational symbols to the output signature of $\phi$, namely $E_{0,1}$, $E_{1,0}$, $E_{0,0}$, and $E_{1,1}$. The arity of symbol $E_{i,j}$ is $(i,j)$. The inclusion of symbols $E_{0,1}$ and $E_{1,0}$ is intuitive --- they correspond to the two edges of the gadget $\rel K_2$. The other two symbols are needed to emulate the last step of the gadget replacement $\gamma$: instead of glueing vertices $u$ and $v$, we will treat the two copies as identical introducing an edge $(u,w)$ whenever $(v,w)$ is present. Since the types of $u$ and $v$ might not agree, this will lead to creating edges in $E_{0,0}$ and $E_{1,1}$. This idea is implemented by the following Datalog program:
  \begin{align*}
    E_{1-j,i}(y, z) &\fro E(x, y), E_{j,i}(x, z) &
    E_{0,1}(x, x) &\fro x = x \\
    E_{i,1-j}(z, y) &\fro E(x, y), E_{i,j}(z, x) &
    E_{1,0}(x, x) &\fro x = x \\
    E_{1-j,i}(x, z) &\fro E(x, y), E_{j,i}(y, z) \\
    E_{i,1-j}(z, x) &\fro E(x, y), E_{i,j}(z, y)
  \end{align*}
  where $i,j \in \{0, 1\}$. The programs $\phi_{E_{ij}}$ only differ in the choice of the output symbol.

  The union gadget $\upsilon$ is similar as in the previous examples; there is a unique choice of $d$ and $r$ so that the final output signature is that of graphs.

  It can be checked that indeed $\upsilon\phi(\rel G)$ and $\gamma(\rel G)$ are homomorphically equivalent, although they might not be isomorphic. For example, $\upsilon\phi(\rel C_4)$ is isomorphic to the complete bipartite graph $\rel K_{4,4}$ while $\gamma(\rel C_4)$ is a single edge, i.e., it is isomorphic to $\rel K_2$.
\end{example}

We proof Theorem~\ref{thm:gadget-is-ddatalog} for projective gadget reductions. This assumption is made without loss of generality: Every gadget reduction is a composition of the reification and a projective gadget using Lemma~\ref{lem:projective-decomposition}. Since reification is a Datalog interpretation, and \Datalog reductions compose, knowing that the projective gadget is expressible as a \Datalog reduction is enough to derive that the original gadget reduction is expressible as one as well.
For the rest of this subsection, we assume that $\gamma$ is a projective gadget mapping $\Pi$-structures to $\Sigma$-structures.

The first step is to prove that a Datalog program can `emulate' a single relation in the $\gamma$-image in the sense of the following lemma.
To simplify the statement, we assume that the structures $\rel D_t^\gamma$ (where $t$ ranges through $\Pi$-types) from the gadget $\gamma$ are pairwise disjoint, and we denote by $\rel D^\gamma$ the (disjoint) union of $\rel D_t^\gamma$ over all $\Pi$-types $t$.

\begin{lemma} \label{lem:gadget-into-datalog}
  For each $\Sigma$-symbol $R$ of arity $k$, and each $k$-tuple $g\in (\rel D^\gamma)^k$, there is a Datalog program $\phi_{R,g}$ such that:
  For every $\Pi$-structure $\rel A$ and every tuple $a\in A^k$,
  $a\in \phi_{R,g}(\rel A)$ if and only if $([a_1; g_1], \dots, [a_k; g_k]) \in R^{\gamma(\rel A)}$.
\end{lemma}

In the statement above, we implicitly assume that, for each $i\in [k]$, the $\Sigma$-type of $g_i$ is $\ar_R(i)$ and that the $\Pi$-type type $t_i$ of $a_i$ satisfies $g_i \in \rel D_{t_i}^\gamma$. Consequently, the arity of the program $\phi_{R,g}$ is $(t_1, \dots, t_k)$.

\begin{proof}
  We construct the program $\phi_{R,g}$ with the required properties. In addition to all $\Pi$-symbols which are EDBs, $\phi_{R,g}$ has a binary IDB $I_{h_1, h_2}$ for each $h_1, h_2 \in \rel D^\gamma$, and an IDB $R_g$ of arity $k$, and $\ar_R(i) = t_i$ for all $i$, which is designated as output.
  Now, let us describe the rules of $\phi^g$.
  \begin{enumerate}
    \item For every $\Pi$-symbol $S$ and $h \in \rel D_{\ar_2(S)}^\gamma$, include the rule:
      \[
        I_{p_S(h), h}(x, y) \fro S(x, y)
      \]
    \item Include the rules:
      \begin{align*}
        I_{h_1, h_1}(x, x) &\fro x = x\\
        I_{h_1, h_2}(x, y) &\fro I_{h_2, h_1}(y, x) \\
        I_{h_1, h_3}(x, z) &\fro I_{h_1, h_2}(x, y), I_{h_2, h_3}(y, z)
      \end{align*}
      for all $h_1, h_2, h_3 \in \rel D^\gamma$.
    \item Finally, add the rule:
      \[
        R_g(x_1, \dots, x_k) \fro
        I_{g_1, h_1}(x_1, y), \dots, I_{g_k, h_k}(x_k, y)
      \]
      for each $h \in R^{\rel D^\gamma}$.
  \end{enumerate}

  Let us show that $\phi_{R,g}$ has the required property. We start by observing that a relation $I_{h_1, h_2}(x, y)$ is derived in $\rel A$ if and only if $[x; h_1] = [y; h_2]$ --- which follows since rules in item (1) introduce into $I$ the equality constraints of the gadget, and rules in item (2) compute the transitive symmetric reflexive closure of $I$. Observing this, it is straightforward to see that
  \[
    ([a_1; g_1], \dots, [a_k; g_k]) \in R^{\gamma(\rel A)}
  \]
  if and only if there exist a $\Pi$-type $s$, $b \in A_s$, and $h \in R^{\rel G_s}$ such that $[b; h_i] = [a_i; g_i]$ for all $i \in [k]$, which is equivalent to triggering the rule (3).
\end{proof}

The above lemma is the only place in the proof where we use recursion in our Datalog program. In particular, it may be observed that recursion is used only to resolve symmetric transitive closure of the relations $I_{i,j}$. This means that the proof provides an argument for any fragment of Datalog that is able to express that, e.g., the \emph{symmetric Datalog} introduced by Egri, Larose, and Tesson~\cite{EgriLT07}.

We provide an example of more detailed implementations of the above lemma building on Example~\ref{ex:strict-gadget}.

\begin{example} \label{ex:gadget-into-datalog}
  Let us describe programs $\phi_{E,(i,j)}$ (where $i,j \in \{0, 1\}$) satisfying the claim of the above lemma for the relational symbol $E$ and the gadget replacement $\gamma$ described in Example~\ref{ex:strict-gadget}.
  All of the programs $\phi_{E,(i,j)}$ contain the following rules:
  \begin{align*}
    I_{0,1}(x, y) &\fro E(x, y) \\
    I_{1,0}(x, y) &\fro E(x, y) \\
    I_{i,i}(x, x) &\fro x = x \\
    I_{i,j}(x, y) &\fro I_{j,i}(y, x) \\
    I_{i,k}(x, z) &\fro I_{i,j}(x, y), I_{j,k}(y, z)
  \end{align*}
  for all $i$, $j$, $k \in \{0, 1\}$. Further, fixing $i$ and $j$, we add into $\phi_{E,(i,j)}$ the following two rules
  \begin{align*}
    E_{i,j}(x, y) &\fro I_{i,0}(x, z), I_{j,1}(y, z) \\
    E_{i,j}(x, y) &\fro I_{i,1}(x, z), I_{j,0}(y, z)
  \end{align*}
  where $E_{i,j}$ is the output predicate.
  To recall the intuition, $I_{i,j}(x, y)$ is derived if $x_i$ and $y_j$ were identified in the third step of application of $\gamma$.

  Observe that $I_{0,1}(x, y)$ and $I_{1,0}(x, y)$ are derived in $\rel G$ if and only if $x$ and $y$ are connected by a path of odd length in $\rel G$, and $I_{0,0}(x, y)$ and $I_{1,1}(x, y)$ are derived if and only if $x$ and $y$ are connected by a path of even length. This exactly corresponds to when $x_i$ and $y_j$ are identified in the third step of application of $\gamma$. Consequently, we get that, if $i\neq j$, $E_{i,j}(x, y)$ is derived if and only if $x$ and $y$ are connected by a path of even length, and $E_{i,i}$ is derived if and only if $x$ and $y$ are connected by a path of odd length.
\end{example}

We are ready to prove Theorem~\ref{thm:gadget-is-ddatalog}. As mentioned above, it is enough to prove this theorem for projective gadgets, i.e., the following lemma.

\begin{lemma}
  For every projective gadget $\gamma$ of $\Sigma$ into $\Pi$, there is a \Datalog reduction $\psi$ such that, for all $\Pi$-structures $\rel A$, $\gamma(\rel A)$ and $\psi(\rel A)$ are homomorphically equivalent.
\end{lemma}

\begin{proof}
  We construct a Datalog interpretation $\phi$ with output signature $\Sigma'$ (which we define below) and a union gadget $\upsilon$; the required \Datalog reduction $\psi$ is the composition of $\phi$ and $\upsilon$. As above, we assume that the structures $\rel D_t^\gamma$ are disjoint and that $\rel D^\gamma$ is their disjoint union.

  The signature $\Sigma'$ contains a type $g$ for each $g\in \rel D^\gamma$, and a symbol $R_g$ for each $\Sigma$-symbol $R$ of arity $k$ and $(g\in \rel D^\gamma)^k$, where $g_i$ is of type $\ar_R(i)$ for all $i\in [k]$. The arity tuple of $R_g$ is $(g_1, \dots, g_k)$.
  We define the Datalog interpretation $\phi$ using the following components:
  For each $\Pi$-type $t$ and $g \in \rel G_t$, $\phi_g$ is the Datalog program composed of the rule $D_g(x) \fro x = x$, where $x$ is of type $t$.
  The relation $R_g$ is defined by the program $\phi_{R,g}$ as in Lemma~\ref{lem:gadget-into-datalog}, i.e., we let $\phi_{R_g} = \phi_{R,g}$.
  The union gadget $\upsilon$ is defined by maps $(d, r)$ where $d$ maps $g$ to its $\Sigma$-type, and $r(R_g) = R$.

  Let us now prove that $\upsilon\phi(\rel A)$ and $\gamma(\rel A)$ are homomorphically equivalent. Observe that we have the following chain of equivalent statements for all $k$-ary $\Sigma$-symbols $R$, and $k$-tuples $a$ and $g$ (such that $g_i \in \rel D_{s_i}^\gamma$ where $s_i$ is the type of $a$, and the type of $g_i$ is $\ar_R(i)$):
  \begin{multline*}
    ([a_1; g_1], \dots, [a_k; g_k]) \in R^{\gamma(\rel A)} \iff
    a \in \phi_{R,g}^{\rel A}  \iff \\
    R_g^{\phi(\rel A)} \iff
    ((a_1; g_1), \dots, (a_k; g_k)) \in R^{\upsilon\phi(\rel A)}
  \end{multline*}
  where the first equivalence is by Lemma~\ref{lem:gadget-into-datalog} and the last equivalence follows from the definition of $\upsilon$.
  Hence, the mapping $f\colon \upsilon\phi(\rel A)\to \phi(\rel A)$ defined by $f(a; g) = [a; g]$ is a homomorphism, and any mapping $g$ such that $g(b) \in f^{-1}(b)$ is a homomorphism in the opposite direction.
\end{proof}

We finish this subsection with two examples.

\begin{example} \label{ex:gadget-cont}
  Let us describe a \Datalog reduction obtained from the projective gadget replacement $\gamma$ from Example~\ref{ex:strict-gadget-ii}.
  The reduction is composed of a Datalog interpretation $\phi$ and a union gadget $\upsilon$.
  The signature $\Sigma'$ from the above proof consists of two types, $0$ and $1$ (corresponding to the elements of $\rel K_2$), and binary relational symbols $E_{i,j}$ for $i, j \in \{0, 1\}$.
  The interpretation $\phi$ is given by the tuple
  \[
    ( \phi_0, \phi_1;
      \phi_{E, (0, 0)}, \phi_{E, (0, 1)}, \phi_{E, (1, 0)}, \phi_{E, (1, 1)})
  \]
  where $\phi_{E, (i, j)}$ are the programs from Example~\ref{ex:gadget-into-datalog}, and $\phi_i$, for $i \in \{0, 1\}$, are programs with a single rule $D_i(x) \fro x = x$.
  The union gadget $\upsilon$ is one of the constructions described in Example~\ref{ex:unions-as-gadget}. In this particular instance, it is defined by mappings $(d,r)$ where $d$ maps both $0$ and $1$ to the only type of vertices of (di)graphs, and $r$ maps $E_{i,j}$ to $E$.

  We claim that the composition $\upsilon\circ\phi$ produces homomorphically equivalent outputs to $\gamma$. To show that, let us assume for simplicity that the input is a connected unoriented graph $\rel G$. We distinguish two cases:
  \begin{enumerate}
    \item $\rel G$ is bipartite, and $\gamma(\rel G) \simeq \rel K_2$. Since $\rel G$ is connected it splits to two parts $A$ and $B$ uniquely. Observe that $x$ and $y$ are in the same part if and only if $x$ and $y$ are connected by a path of even length, and they are in distinct parts if and only if $x$ and $y$ are connected by a path of odd length. Hence $\upsilon\phi(\rel G)$ is the complete bipartite graph with parts $(A \times \{0\}) \cup (B \times \{1\})$ and $(A\times \{1\}) \cup (B\times \{0\})$. Clearly, it is homomorphically equivalent to $\rel K_2$.

    \item $\rel G$ contains an odd cycle, and $\gamma(\rel G)$ is a loop. In this case there is an odd path from $x$ to $x$ for each $x$. Hence all pairs of elements $x$ and $y$ are connected by both odd and even path. This implies that $\upsilon\phi(\rel G)$ is a clique with all loops, and hence clearly homomorphically equivalent to a loop.
  \end{enumerate}
  Consequently, we have that $\upsilon \circ \phi$ is a reduction from $\CSP(\rel K_2)$ to $\CSP(\rel K_\infty)$.
\end{example}

\begin{example} \label{ex:2-colouring-reduces}
  We can connect the above example with a rather easy Datalog reduction from $\CSP(\rel K_\infty)$ to $\CSP(\False)$, which is defined by the Datalog interpretation $\psi$ where $\psi_C$ is the program with the rule $C \fro E(x, x)$.
  Connecting these two reductions together, we get that
  \[
    \CSP(\rel K_2) \leqdl \CSP(\rel K_\infty) \leqdl \CSP(\False)
  \]
  where the full reduction from left to right is $\psi\circ \upsilon \circ\phi$ where $\upsilon$ and $\phi$ are as in the example above.
  Theorem~\ref{thm:ddatalog-composes} asserts that $\CSP(\rel K_2) \leqdl \CSP(\False)$, i.e., that there is a Datalog interpretation $\phi'$ and a union gadget $\upsilon'$ such that $\upsilon'\circ\phi'$ gives homomorphically equivalent outputs to $\psi\circ \upsilon$. In fact, the union gadget $\upsilon'$ is not needed in this case (due to a property of the structure $\False$). Let us describe how to construct this $\phi'$ from $\phi$, $\upsilon$, and $\psi$.

  First, we switch the order of $\psi$ and $\upsilon$ in the same way as in  Example~\ref{ex:switching} (except that the binary relations have a different name). The Datalog interpretation $\psi'$ has one component $\psi'_C$ with rules
  \[
    C \fro E_{0,0} \qquad C \fro E_{1,1}.
  \]

  The Datalog interpretation $\phi'$ obtained as a composition $\psi'\circ \phi$ is a valid reduction from $\CSP(\rel K_2)$ to $\CSP(\rel K_\infty)$.  The only component of $\phi'$ is the Datalog program $\phi'_C$ which is the union of $\phi_{E, (0,0)}$, $\phi_{E, (1, 1)}$, and the two rules of $\psi'_C$ above, i.e., the following program:
  \begin{align*}
    I_{0,1}(x, y) &\fro E(x, y) &
    I_{i,k}(x, z) &\fro I_{i,j}(x, y), I_{j,k}(y, z) \\
    I_{1,0}(x, y) &\fro E(x, y) &
    E_{i,i}(x, y) &\fro I_{i,0}(x, z), I_{i,1}(y, z) \\
    I_{i,i}(x, x) &\fro x = x &
    E_{i,i}(x, y) &\fro I_{i,1}(x, z), I_{i,0}(y, z) \\
    I_{i,j}(x, y) &\fro I_{j,i}(y, x) &
    C &\fro E_{i,i}
  \end{align*}
  where $i,j,k \in \{0,1\}$.
  Note that this $\psi'_C$ decides whether the input has an odd cycle, or not. This program that we constructed by joining our proofs is a more verbose version of a simpler Datalog program solving $\CSP(\rel K_2)$ described in \cite[Section 4.1]{KV00}.

  Finally, let us remark that although we started with a projective gadget replacement $\gamma$ and a Datalog interpretation $\phi$ of width $1$ that does not use recursion (i.e., the graph whose vertices are relational symbols appearing in a program, and there is an edge from $R$ to $S$ if $S$ appears in the head of a rule which has $R$ in the body, does not have directed cycles), the program $\psi'_C$ has width $3$. It can be shown, using \cite{KV00}, that this cannot be avoided: there is no Datalog program of width $2$ that could be used instead of $\psi'_C$, and no Datalog program that does not use recursion that could be used instead of~$\psi'_C$.
\end{example}

\subsection{The equivalence of \texorpdfstring{\Datalog}{Datalog} and \texorpdfstring{$k$}{k}-consistency reductions}
  \label{sec:datalog-is-k-consistency}

The goal of this section is to prove Theorem~\ref{thm:canonical-width}.
Let us briefly outline the structure of the proof and the intuition behind it. We fix structures $\rel A$, $\rel B$, and a parameter $k \geq 1$. The theorem claims two implications about reductions from $\PCSP(\rel A, \rel A')$ to $\PCSP(\rel B, \rel B')$, namely that:
(1) If $\phi$ is a Datalog interpretation of width $k$ and $\gamma$ a gadget, such that $\gamma\circ\phi$ is a reduction between the two promise CSPs, then so is $\kappa_k^{\rel A, \rel B}$.
(2) If $\kappa_k^{\rel A, \rel B}$ is a reduction between the two promise CSPs then so is $\pi_\rel B\circ \psi^{k, \rel A}$, where $\pi_\rel B$ is the universal gadget, and $\psi^{k, \rel A}$ is a Datalog interpretation of width $k$ which is based on the \emph{canonical Datalog program}, and we define it below.

A feature of the proof is that neither of the constructed reductions depends on the structures $\rel A'$ and $\rel B'$, which play a role in the soundness of reductions. We compare the two reductions without referring to these two structures, and show that in both implications the two involved reductions, let us denote them by $\phi$ and $\psi$, satisfy $\psi(\rel X) \to \phi(\rel X)$ for all instances $\rel X$. This relation can be then used to transfer the soundness of $\psi$ to the soundness of $\phi$.
As usual, completeness is easier to show.

In order to compare reductions in both implications, we decompose $\kappa_k^{\rel A, \rel B}$ into $\kappa_k^{\rel A} = \ac \circ \sigma_k^{\rel A}$, that produces a label cover instance and depends only on $\rel A$, and the universal gadget $\pi_\rel B$.

To compare label cover instances produced as an intermediate step we leverage an algebraic intuition: First, we use a universal property of $\pi_\rel B$ which, loosely speaking, says that $\pi_\rel B$ is the `best gadget' for a reduction from label cover to a $\PCSP(\rel B, {*})$. More precisely, reading the proof of Theorem~\ref{thm:gadget-characterisation} in detail, we may observe that if a gadget $\gamma$ provides a reduction from $\PCSP(\rel P, {*})$ to $\PCSP(\rel B, {*})$, then $\pi_\rel B$ is also such a reduction. We will actually show that this statement is provided by a homomorphism $\gamma(\rel X) \to \pi_\rel B(\rel X)$ which exists whenever $\gamma(\rel P) \to \rel B$, i.e., that the soundness of $\pi_\rel B$ can be derived from the soundness of $\gamma$ by the method we described above.

Second, in both implications, we will need to compare different instances of label cover whose types do not align well, e.g., the elements of $\kappa_k^{\rel A}(\rel X)$ are of types $\mathcal F_K$, which is a subset of partial homomorphism from $K$ to $\rel A$ for some $K\in \binom X{\leq k}$, while the elements of $\psi^{k,\rel A}(\rel X)$ are of types which correspond to relations $R \subseteq A^m$ for $m \leq k$. The set $\mathcal F_K$ and the relation $R$ can be compared using known facts about Datalog and consistency.
Nevertheless, the fact that these sets relate to each other is not enough to provide a homomorphism between the two label cover instances for formal reasons. Instead, we relax the notion of a homomorphism (again guided by an algebraic intuition) to a notion of a `weak homomorphism'.\footnote{If a label cover instance is interpreted as a \emph{minor condition}, these `weak homomorphisms' correspond exactly to `syntactic implications between minor conditions'.} We express these `weak homomorphisms' in the combinatorial language as follows: `$\rel S$ maps weakly to $\rel T$' if $\rel S \to \pi_\rel P(\rel T)$; recall, that $\rel P$ denotes the template of the label cover problem. It turns out that the existence of this homomorphism is equivalent to existence of homomorphisms $\pi_\rel B (\rel S) \to \pi_\rel B (\rel T)$ for all structures $\rel B$. Consequently, we get that a `weak homomorphism' is mapped by a universal gadget to a normal homomorphism which we will use in our proof.
In this instance, we want to show that $\kappa_k^{\rel A}(\rel X) \to \pi_\rel P\psi^{k, \rel A}(\rel X)$ which is a straightforward technical exercise (using the connection between $\mathcal F_K$ and $R$ outlined above). Alternatively, we could again leverage the universality of $\pi_\rel P$ by observing that we can recover $\psi^{k, \rel A}(\rel X)$ from $\kappa_k^{\rel A}(\rel X)$ by application of a simple gadget reduction~$\gamma$.

The soundness of $\kappa_k^{\rel A, \rel B}$ in the other implication is obtained in a similar way.  Let us recall that our goal is to show that $\gamma\phi(\rel X) \to \kappa_k^{\rel A, \rel B}(\rel X)$ where $\phi$ is a Datalog interpretation of width $k$ and $\gamma$ is a gadget.
In order to compare these two reductions, we factor each of them through label cover, and then again use `weak homomorphisms' to compare the intermediate label cover instances. In particular, the reduction $\gamma\circ \phi$ is decomposed as $\gamma' \circ \phi'$ where $\gamma'$ is obtained by decomposing $\gamma$ as $\gamma' \circ \rho^{\phi(\rel A)}$ using Lemma~\ref{lem:projective-decomposition}; $\phi'$ is then $\rho^{\phi(\rel A)} \circ \phi$ which is expressible as a Datalog interpretation of the same width as $\phi$.
We then prove that
\begin{equation} \label{eq:width}
  \phi'(\rel X) \to \pi_\rel P\kappa_k^{\rel A}(\rel X)
\end{equation}
for all structures $\rel X$ with the same signature as $\rel A$. This statement is again proved by leveraging the known connection of Datalog and consistency. The proof of the implication is then concluded by the universality of $\pi_\rel B$.

Incidentally, the proof of \eqref{eq:width} applies without any significant changes to the canonical Datalog program $\psi^{k, \rel A}$ in place of $\phi'$. Consequently, the structures $\pi_\rel B\psi^{k, \rel A}(\rel X)$ and $\kappa_k^{\rel A, \rel B}(\rel X)$ are homomorphically equivalent (in fact, with a bit more effort, it can be shown that they are isomorphic). This observation provides ground for saying that $k$-consistency reduction \emph{is} a Datalog reduction (as opposed to only being equivalent to one). And we also use it to derive the completeness of $\pi_\rel B \circ \psi^{k, \rel A}$ from the completeness of $\kappa_k^{\rel A, \rel B}$.

The rest of this appendix is a formalisation of the ideas described above. The proofs are presented in full technical detail.
We start by introducing some notation, introducing the components of our reductions, and recalling known results.

\subsubsection{Notation}

We will extensively use a functional notation for tuples, i.e., we view $w \in X^n$ as a function $w\colon [n] \to X$. In particular, $\im w$ is the set of all entries of $w$, i.e., $\im w = \{w_1, \dots, w_n\}$, and if $\pi\colon [m] \to [n]$ then $w\circ \pi = (w_{\pi(1)}, \dots, w_{\pi(n)})$.
We write $R\circ \pi$ for the $m$-ary relation $\{ a\circ \pi \mid a\in R \}$.
Finally, we will denote by $X^{\leq k}$ the set of all tuples of elements from $X$ of length at most~$k$.
Occasionally in longer expressions, we will write function composition simply as juxtaposition, i.e., $fg$ instead of $f\circ g$. Note that composition takes precedence, i.e., $fg(x) = f(g(x))$ not $f\circ g(x)$.

\subsubsection{The universality of gadgets}

Let us start with the universality of gadgets $\pi_\rel B$, i.e., the following lemma. Let us reiterate that the intuition is provided by \cite{BBKO21} which implicitly claims that if a gadget reduction is a valid reduction from a reduct of label cover to $\PCSP(\rel B, {*})$ then so is $\pi_\rel B$.

Fix a reduct of label cover signature $\Sigma$, and let $\rel P^\Sigma$ denote the $\Sigma$-reduct of the label cover template $\rel P$, and fix a structure $\rel B$.

\begin{lemma} \label{lem:universality-of-gadget}
  If $\gamma$ is a gadget such that $\gamma(\rel P^\Sigma) \to \rel B$, then  $\gamma(\rel S) \to \pi_{\rel B}(\rel S)$ for every $\Sigma$-structure $\rel S$.
\end{lemma}

\begin{proof}
  Let $b\colon \gamma(\rel P^\Sigma) \to \rel B$ be a homomorphism.
  We define a homomorphism $h\colon \gamma(\rel S) \to \pi_{\rel B}(\rel S)$ by
  \[
    h([x; g]) = [x; d_g]
  \]
  for each $x \in S_X$ and $g\in \rel D_X^\gamma$ where $d_g\colon X\to B$ is defined by $d_g(a) = b([a; g])$; and
  \[
    h([(x,y); g]) = [x; d_g]
  \]
  for each $(x, y) \in E_\sigma^\rel S$, and $g\in \rel E_\sigma^\gamma$, where $d_g\colon X \to B$ is defined by $d_g(a) = b([(a, \sigma(a)); g])$.
  We need to show that $h$ is a well-defined homomorphism which is relatively straightforward.

  Let us start with checking that $h$ preserves the equality constraints introduced to $\gamma(\rel S)$ in Step (3) of the gadget replacement.
  The equality constraints are of two forms,
  \begin{align}
    ((x, y); p_{E_\sigma, 1}^\gamma(g)) &= (x; g) \label{eqc1}\\
    ((x, y); p_{E_\sigma, 2}^\gamma(g)) &= (y; g) \label{eqc2}
  \end{align}
  where $(x, y) \in E_\sigma^{\rel S}$ for some $\sigma\colon X\to Y$, and $g\in \rel D_X$ in the first case, or $g\in \rel D_Y$ in the second case. For simplicity, we will omit $E_\sigma$ from the lower index in $p_i^\gamma$ for $i = 1, 2$.
  Let us start with \eqref{eqc1}: Observe that, for all $a \in X$,
  \[
    d_g(a) = b([a; g]) = b([(a, \sigma(a)); p_1^\gamma(g)]) = d_{p_1^\gamma(g)}(a)
  \]
  where the middle equality holds since $(a, \sigma(a))\in E^\rel P_\sigma$. This equality implies that $h([x; g]) = h([(x, y); p_1^\gamma(g)])$ since the definitions coincide.
  For \eqref{eqc1} we observe similarly that, for all $a \in X$,
  \[
    d_g(\sigma(a)) = b([\sigma(a); g]) = b([(a, \sigma(a)); p_1^\gamma(g)]) = d_{p_2^\gamma(g)}(a).
  \]
  This implies that $d_g \circ \sigma = d_{p_2^\gamma(g)}$, and consequently that
  \[
    h([y; g])
    = [y; d_g]
    = [x; d_g \circ \sigma]
    = [x; d_{p_2^\gamma(g)}]
    = h([(x, y); p_2^\gamma(g)])
  \]
  as we wanted to show.

  To show that $h$ preserves relations, let
  \[
    ([x_1; g_1], \dots, [x_\ell; g_\ell]) \in R^{\gamma(\rel S)}.
  \]
  Since $h$ is well-defined, we may assume without loss of generality that $x_1 = \dots = x_\ell = x$ and that $(g_1, \dots, g_\ell) = g$ is in the $R$-relation of the corresponding component of $\gamma$. Furthermore, let us assume that $x$ is an element of type $X$, and $g\in R^{\rel D_X^\gamma}$ --- the other case, when $x \in E_\sigma^\rel S$ for some $\sigma$ is analogous.
  Since $h([x; g_i]) = [x; d_{g_i}]$, it is enough to prove that $(d_{g_1}, \dots, d_{g_\ell}) \in R^{\rel B^X}$ which is immediate from the definition of the power $\rel B^X$ since $g \in R^{\rel D_X^\gamma}$ and $b$ is a homomorphism.
\end{proof}

The following is a direct corollary of the above.

\begin{lemma} \label{lem:universal-gadget-is-monotone}
  Let $\rel S$ and $\rel T$ be two label cover instances such that $\rel S \to \pi_\rel P(\rel T)$. Then $\pi_{\rel B}(\rel S) \to \pi_{\rel B}(\rel T)$ for all structures $\rel B$.
\end{lemma}

\begin{proof}
  Since $\pi_\rel B$ is monotone, we have that $\pi_\rel B(\rel S) \to \pi_\rel B \pi_\rel P(\rel T)$.
  Observe that $\pi_\rel X(\rel P) \to \rel X$ for all $\rel X$ by mapping $[x, p] \mapsto p(x)$.
  Consequently, $\pi_\rel B \circ \pi_\rel P$ is a gadget such that $\pi_\rel B \pi_\rel P(\rel P) \to \rel B$, and hence the previous lemma applies, and we get $\pi_\rel B\pi_\rel P(\rel T) \to \pi_\rel B(\rel T)$, and the desired by composing the two homomorphisms.
\end{proof}

Let us also note that the converse of the lemma is true, though it is not necessary for our proof: If $\pi_\rel B(\rel S) \to \pi_\rel B(\rel T)$ for all $\rel B$, then in particular $\pi_\rel P(\rel S) \to \pi_\rel P(\rel T)$, and it is not hard to observe that $\rel S \to \pi_\rel P(\rel S)$ (alternatively, we can argue that this homomorphism follows from the universality of $\pi_\rel P$ since the identity mapping is a gadget replacement), and therefore $\rel S \to \pi_\rel P(\rel T)$.

\subsubsection{Canonical Datalog interpretation}

We define the canonical interpretation which is based on the \emph{canonical Datalog program} $\psi$ of width $k$ for a structure $\rel A$. We define this Datalog program as a set of rules without an output predicate --- the Datalog interpretation will use copies of this program with different output predicates as components.

\begin{definition}
  Fix a structure $\rel A$ and $k \geq 1$.
  We denote by $\Delta^{k, \rel A}$ the set of all relations $R \subseteq A^m$, where $m\leq k$, that are definable in $\rel A$ by a Datalog program of width at most $k$, i.e., all relations $R$ of the form $R = \phi(\rel A)$ where $\phi$ is a Datalog program of width at most~$k$.

  The signature of the \emph{canonical program} $\psi$ for $\rel A$ of width $k$ consists (apart from the basic symbols of $\rel A$) of a symbol $R$ for each $R \in \Delta^{k, \rel A}$. We treat each such relation as an abstract symbol of the same arity as the Datalog program $\phi$ that defined it. 
  The rules of the canonical Datalog are all the rules (with at most $k$ variables) that are satisfied in $\rel A$ when each rule is interpreted as an implication $t_0\leftarrow t_1 \wedge \dots \wedge t_r$ where $t_0$ is the head and $t_1, \dots, t_r$ is the body (and the basic symbols $R$ are interpreted as $R^\rel A$).
\end{definition}

With this program at hand, we can construct the canonical Datalog interpretation.

\begin{definition}
  Fix a structure $\rel A$ and $k\geq 1$. The \emph{canonical Datalog interpretation} $\psi^{k, \rel A}$ is defined as follows. The output of $\psi^{k, \rel A}$ is a reduct of a label cover instance. More precisely, the output signature consists of a type $R$ for each $R \in \Delta^{k, \rel A}$ which is viewed as the set $R \subseteq A^m$, and a binary symbol $E_\pi$ for each $\pi \colon R \to S$ of the form $\pi(a) = a\circ p$ for some $p\colon [m_S] \to [m_R]$ where $m_R$ and $m_S$ denote the arity of $R$ and $S$, respectively.

  The domains are defined by the Datalog programs $\psi^{k, \rel A}_R$ obtained from the canonical Datalog program by designing $R$ as the output.
  The relation $E_\pi$ where $\pi\colon R \to S$ is of the form $\pi(a) = a\circ p$ is defined by the program $\psi^{k, \rel A}_{E_\pi}$ obtained from the canonical program by adding the rule
  \[
    E_\pi(x_1, \dots, x_n, x_{p_\pi(1)}, \dots, x_{p_\pi(m)})
      \fro R(x_1, \dots, x_n).
  \]
\end{definition}

Note that $\psi^{k, \rel A}_R(\rel A) = R$ for each $R \in \Delta^{k, \rel A}$, and that $\psi^{k, \rel A}_{E_\pi}(\rel A) = \{ (a, a\circ p) \mid a\in R \}$. Hence $\psi^{k, \rel A}(\rel A)$ is a reduct of $\rel P$ whose domains consist of all relations definable in $\rel A$ by a Datalog program of width~$k$.

Let us recall a few known facts about the connection between the $k$-consistency algorithm and Datalog program of width $k$. This connection is well understood, and we refer to, e.g., \citet{KV95,FV98} for more details. We will use three lemmas which are either implicitly, or explicitly proved in \cite{KV95} (see, e.g., \cite[Theorem 4.8]{KV95}).

\begin{lemma}[\cite{KV95}] \label{lem:fromKV}
  Let $\rel A$ and $\rel X$ are structures with the same signature, $k \geq 1$, and $\mathcal F_K$ denote the sets output by the $k$-consistency enforcement applied on the instance $\rel X$, i.e., $\mathcal F_K$ is the type of the variable $v_K \in \kappa_k^\rel A(\rel X)$.
  Let $w$ be a tuple such that $\im w \in \binom X{\leq k}$, and $R = \mathcal F_{\im w} \circ w$.
  \begin{enumerate}[(1),ref=\ref{lem:fromKV}(\arabic*)]
    \item \label{lem:consistency-and-datalog}
      For each Datalog program $\phi$ of width at most $k$, if $w \in \phi(\rel X)$, then $a\in \phi(\rel A)$ for all $a\in R$.\footnote{This applies to all tuples in $w\in \phi(\rel X)$ since if $\phi$ has width at most $k$, $\im w$ has at most $k$ elements.}
    \item \label{lem:winning-strategy}
      If $w \in X^{\leq k}$, then $R$ is the smallest relation $S \in \Delta^{k, \rel A}$ (w.r.t., inclusion) such that $w \in S$.
    \item \label{lem:canonical-datalog}
      If $w \in X^{\leq k}$, then $w \in \psi^{k, \rel A}_R(\rel X)$.
  \end{enumerate}
\end{lemma}

\subsubsection{Soundness}

We fix structures $\rel A$ and $\rel B$, and a parameter $k \geq 1$. We first prove the implication claiming that if $\kappa_k^{\rel A, \rel B}$ is a reduction then so is $\pi_\rel B \circ \psi^{k, \rel A}$. As outlined above, the first step is to provide the following homomorphism.

\begin{lemma} \label{lem:datalog-and-kappa}
  For each structure $\rel X$ with the same signature as $\rel A$,
  \[
    \kappa_k^\rel A(\rel X) \to \pi_\rel T\psi^{k, \rel A}(\rel X).
  \]
\end{lemma}

Intuitively, the above homomorphism exists since $\kappa_k^\rel A(\rel X)$ is essentially a substructure of $\psi^{k, \rel A}(\rel X)$ if we were allowed to map a variable $v_K$ where $K = \{x_1, \dots, x_\ell\}$ to a tuple $w = (x_1, \dots, x_\ell) \in \psi^{k, \rel A}_R(\rel X)$ where $R = \mathcal F_K \circ w$ (which is Datalog definable by Lemma~\ref{lem:winning-strategy}). The issue is that for a proper homomorphism, we would need that $R = \mathcal F_K$. Instead we provide a mapping $\kappa_k^\rel A(\rel X) \to \pi_\rel P \psi^{k, \rel A}(\rel X)$ by choosing a bijection between $K$ and $[\ell]$ which induces a bijection between $R$ and $\mathcal F_K$.

In the formal proof, we will use the following technical lemma.

\begin{lemma} \label{lem:technical}
  Assume that $\rel S$ is a label cover instance.
  For all $s \in S_{X_s}$, $t\in S_{X_t}$, $f_s\colon X_s \to Y_s$, $f_t\colon X_t \to Y_t$, $\sigma \colon X_s \to X_t$, and $\sigma' \colon Y_s \to Y_t$ such that $\sigma' \circ f_s = f_t \circ \sigma$, and $(s, t) \in E_\sigma^\rel S$, we have
  \[
    ([s, f_s], [t, f_t]) \in E_{\sigma'}^{\pi_\rel P(\rel S)}.
  \]
\end{lemma}

\begin{proof}
  First observe that $[t, f_t] = [s, f_t\circ \sigma]$ and $([s, f_s], [s, \sigma' \circ f_s]) \in E_{\sigma'}^{\pi_\rel P(\rel S)}$ follows directly from the definition of $\pi_\rel P(\rel S)$.
  Using the assumption $\sigma' \circ f_s = f_t \circ \sigma$, the former observation gives that $[t, f_t] = [s, \sigma' \circ f_s]$. This equality together with the latter observation yields the desired.
\end{proof}

\begin{proof}[Proof of Lemma~\ref{lem:datalog-and-kappa}]
  To simplify the notation, let us write $\kappa$ instead of $\kappa_k^\rel A$ and $\psi$ instead of $\psi^{k, \rel A}$.

  To define a homomorphism $h\colon \kappa(\rel X) \to \pi_\rel P\psi(\rel X)$. Fix a bijection $w_K \colon [\lvert K\rvert] \to K$ for each $K \in \binom X{\leq k}$ and let $w_K^{-1}$ denote the inverse of this bijection.
  By Lemma~\ref{lem:canonical-datalog}, $w_K \in \psi_{R_K}(\rel X)$ where $R_K = \mathcal F_K \rcirc w_K$. Hence we may interpret $w_K$ as an element of type $R_K$ in $\psi(\rel X)$. We define $h$ by
  \[
    h(v_K) = [w_K, p_K]
  \]
  where $p_K\colon R_K \to \mathcal F_K$ is defined by $p_K(a) =  a\circ w_K^{-1}$. In the rest of the proof, we check that $h$ is indeed a homomorphism.

  To show that $h$ preserves the relations, we need to show that  if $(v_K, v_L) \in E_\sigma^{\kappa(\rel X)}$, i.e., if $L \subset K$ and $\sigma(f) = f|_L$, then
  \[
    ([w_K; p_K], [w_L; p_L]) \in E_\sigma^{\psi(\rel X)}.
  \]
  By Lemma~\ref{lem:technical}, it is enough to provide $\sigma' \colon R_K \to R_L$ such that $p_L \circ \sigma' = \sigma \circ p_K$, and $(w_K, w_L) \in E_{\sigma'}^{\psi(\rel X)}$. We let $p \colon [\lvert L\rvert] \to [\lvert K\rvert]$ to be the mapping $p = w_K^{-1} \circ w_L$, i.e., the mapping which for each $i$ returns the index of $w_L(i)$ in $w_K$, and define $\sigma'$ by $\sigma'(a) = a \circ p$. Clearly,
  \[
    \sigma'(w_K) = w_K \circ p = w_K w_K^{-1} w_L = w_L,
  \]
  and consequently, $(w_K, w_L) \in E_{\sigma'}^{\psi(\rel X)}$. Next, we check that $\sigma'$ satisfies the equality above. Let $a\in R_K$, then
  \[
    p_L \sigma'(a) = a p w_L^{-1} = a w_K^{-1} w_L w_L^{-1} = (a w_K^{-1})|_L = \sigma p_K(a)
  \]
  as we wanted to show.
\end{proof}

The following corollary, which is enough to provide soundness of $\pi_\rel B \circ \psi^{k, \rel A}$ assuming soundness of $\kappa_k^{\rel A, \rel B}$, follows directly from the above lemma and Lemma~\ref{lem:universal-gadget-is-monotone}.

\begin{corollary} \label{cor:soundness-of-canonical}
  For every structure $\rel X$ with the same signature as $\rel A$,
  \[
    \kappa_k^{\rel A, \rel B} (\rel X) \to \pi_\rel B \psi^{k, \rel A} (\rel X)
  \]
\end{corollary}

We move onto proving soundness of $\kappa_k^{\rel A, \rel B}$ assuming that there is a Datalog interpretation $\phi$ and a gadget $\gamma$ such that $\gamma\circ \phi$ is a valid reduction. The first step is the universality of $\kappa_k$ phrased in the following lemma in which we, again, leverage the connection of Datalog and $k$-consistency algorithm.

\begin{lemma} \label{lem:universality-of-kcons}
  Let $\phi$ be a Datalog interpretation of width at most $k$, then for every structure $\rel X$ with the same signature as $\rel A$,
  \[
    \rho^{\phi(\rel A)}\phi(\rel X) \to \pi_\rel P \kappa_k^{\rel A}(\rel X).
  \]
\end{lemma}

\begin{proof}
  Let us assume that $\Sigma$ is the output signature of $\rho^{\phi(\rel A)}$, which is a reduct of the label cover signature, and that the domains of $\rel A$ are disjoint.
  First, we derive a few properties of the Datalog interpretation $\phi'$ obtained as a composition $\rho^{\phi(\rel A)}\circ \phi$ (see Lemma~\ref{lem:cases}(2)).
  Namely, $\phi'$ is obtained as follows:
  For each type $t$, and each relational symbol $R$ (in the signature of $\phi(\rel A)$), we let $\phi'_{A_t} = \phi_t$ and $\phi'_{R^\rel A} = \phi_R$. These programs define the domains of $\phi'(\rel X)$. The relation corresponding to $E_\sigma$, for a symbol $R$ of arity $\ell$, $i\in [\ell]$, and $\sigma\colon R^\rel A \to A_{\ar(i)}$ is the $i$-th projection, is defined by the program $\phi'_{E_\sigma}$ which is obtained from $\phi_R$ by adding the rule 
  \begin{equation} \label{eq:p-sigma}
    E_\sigma(x_1, \dots, x_m, x_{p_\sigma(1)}, \dots, x_{p_\sigma(n_i)})
      \fro R(x_1, \dots, x_m)
  \end{equation}
  where $n_i$ is the arity of $\phi_{\ar_R(i)}$ and $p_\sigma \colon [n_i] \to [m]$ is the mapping that selects the $n_i$ consecutive coordinates of $R$ that corresponds to the coordinates defining the $i$-th projection of $R$, i.e., $p_\sigma(x) = x + \sum_{j<i} n_j$ (see Definition~\ref{def:datalog-interpretation}), and changing the output predicate accordingly.
  Remark that $\phi'$ is not necessarily of width $k$ --- we will only use the fact that $\phi'_t$ and $\phi'_R$ are of width $k$ and that $\phi'_{E_\sigma}$ is of the form above.

  From now on we will we work with the interpretation $\phi'$ instead of $\phi$, and write $\kappa$ instead of $\kappa_k^\rel A$.
  We define a homomorphism $h \colon \phi'(\rel X) \to \pi_\rel P \kappa(\rel X)$. We have to be a bit careful, since the domains of $\phi'(\rel X)$ are not necessarily disjoint. To avoid confusion, we will define $h$ as a collection of mappings $h_T\colon \phi'_T(\rel X) \to \kappa(\rel X)$ where $T$ is a $\Sigma$-type. Note that each such type corresponds to a domain of $\phi'(\rel A)$, i.e., $T$ is a Datalog definable relation in $\rel A$.
  By Lemma~\ref{lem:consistency-and-datalog}, we have that $\mathcal F_{\im w} \rcirc w \subseteq \phi'_T(\rel A) = T$ for each $w\in \phi'_T(\rel X)$. This observation allows us to define $h_T$ as
  \[
    h_T(w) = [v_{\im w}; p_w]
  \]
  where $p_w\colon \mathcal F_K \to T$ is defined by $p_w(f) = f \circ w$.

  Let us prove that $h$ is indeed a weak homomorphism.
  We use Lemma~\ref{lem:technical} to show that if $(w, w') \in E_\sigma^{\phi'(\rel X)}$, then
  \[
    ([v_{\im w}; p_w], [v_{\im w'}; p_{w'}]) \in E_\sigma^{\pi_\rel P\kappa(\rel X)},
  \]
  i.e., we show that there is $\sigma'$ such that $(v_{\im w}, v_{\im w'}) \in E_{\sigma'}^{\kappa(\rel X)}$ and
  \begin{equation} \label{eq:required}
    p_{w'} \circ \sigma' = \sigma \circ p_w.
  \end{equation}
  Note that $\sigma\colon R \to S$ is of the form $\sigma(a) = a \circ p_\sigma$ and $w' = w \circ p_\sigma$ since $E_\sigma^{\phi'(\rel X)}$ is defined by the rule \eqref{eq:p-sigma}.
  Consequently $\im w' \subseteq \im w$, hence $(\im w, \im w') \in E_{\sigma'}$ where $\sigma'\colon \mathcal F_{\im w} \to \mathcal F_{\im w'}$ is the restriction $\sigma'(f) = f|_{\im w'}$. Let us check the identity \eqref{eq:required} by applying both sides to $f\in \mathcal F_{\im w}$:
  \[
    p_{w'}\sigma' (f)
    = p_{w'}(f|_{\im w'})
    = f w'
    = f w p_\sigma
    = \sigma (f w)
    = \sigma p_w (f).
  \]
  This concludes the proof.
\end{proof}

We now conclude the proof of the soundness of $\kappa_k^{\rel A, \rel B}$ assuming soundness of $\gamma\circ \phi$ with the following corollary.

\begin{corollary} \label{cor:soundness-of-consistency}
  Let $\phi$ be a Datalog interpretation of width at most $k$ and $\gamma$ be a gadget, such that $\gamma\phi(\rel A) \to \rel B$.
  For every structure $\rel X$ with the same signature as $\rel A$,
  \[
    \gamma\phi(\rel X) \to \kappa_k^{\rel A, \rel B}(\rel X)
  \]
\end{corollary}

\begin{proof}
  Decompose $\gamma$ as $\gamma_{\phi(\rel A)} \circ \rho^{\phi(\rel A)}$ using Lemma~\ref{lem:projective-decomposition}. Consequently, we have
  \begin{equation} \label{eq:A.4}
    \rho^{\phi(\rel A)}\phi(\rel X) \to \pi_\rel P \kappa_k^{\rel A}(\rel X)
  \end{equation}
  by Lemma~\ref{lem:universality-of-kcons}. 
  Observe that $\gamma_{\phi(\rel A)}$ maps the corresponding reduct of the label cover template $\rho^{\phi(\rel A)}\phi(\rel A)$ to $\rel B$, which allows us to apply Lemma~\ref{lem:universality-of-gadget} to get a homomorphism
  \[
    \gamma \phi(\rel X)
      = \gamma_{\phi(\rel A)} \rho^{\phi(\rel A)} \phi(\rel X)
      \to \pi_\rel B \rho^{\phi(\rel A)} \phi(\rel X).
  \]
  Further, by Lemma~\ref{lem:universal-gadget-is-monotone} and \eqref{eq:A.4} we have that
  \[
      \pi_\rel B \rho^{\phi(\rel A)} \phi(\rel X)
      \to \pi_\rel B \kappa_k^{\rel A}(\rel X)
      = \kappa_k^{\rel A, \rel B}(\rel X)
  \]
  which together with the above homomorphism gives the desired.
\end{proof}

\subsubsection{Completeness}

Let us start with the completeness of the $k$-consistency reduction which can be proved by the following straightforward argument.

\begin{lemma} \label{lem:completeness-of-consistency}
  If $\rel X \to \rel A$, then $\kappa_k^{\rel A, \rel B}(\rel X) \to \rel B$.
\end{lemma}

\begin{proof}
  Let $g\colon \rel X\to \rel A$ be a homomorphism.  First, we observe that, for each $K \in \binom X{\leq k}$, the restriction $g|_K$, satisfies $g|_K \in \mathcal F_K$: this is since the restrictions are locally consistent partial homomorphisms and will never be removed from $\mathcal F_K$.
  We use this observation to define a homomorphism $h\colon \kappa^{\rel A, \rel B}_k(\rel A) \to \rel B$. We define $h$ by
  \[
    h([v_K; b]) = b(g|_K)
  \]
  where $K \in \binom X{\leq k}$ and $b\in B^{\mathcal F_K}$.
  We need to show that it is well-defined, i.e., that the value of $h$ does not depend on the choice of representative. It is enough to check that for pairs $(v_K; b)$ and $(v_L; b')$ where $L \subset K$ and $b(f) = b'(f|_L)$.  In that case, we have
  \[
    h([v_K; b]) = b(g|_K) = b'((g|_K)|_L) = b'(g|_L) = h([v_L; b']).
  \]
  The fact that $h$ preserves relations is immediate.
\end{proof}

The completeness of $\pi_{\rel B}\circ \psi^{k, \rel A}$ can be proved in several ways. First, observing that both reductions are monotone, it is enough to check that $\pi_{\rel B}\psi^{k, \rel A}(\rel A) \to \rel B$, which can be done in two steps by checking that $\psi^{k, \rel A}(\rel A) \to \rel P$ and that $\pi_{\rel B}(\rel P) \to \rel B$. Alternatively, we may observe the following fact which follows from our proofs of soundness, and allows us to infer the completeness of $\pi_{\rel B}\circ \psi^{k, \rel A}$ from the completeness of $\kappa_k^{\rel A, \rel B}$.

\begin{lemma} \label{lem:consistency-is-datalog+}
  For all structures $\rel X$ of the same signature as $\rel A$, $\pi_{\rel B}\psi^{k, \rel A}(\rel X)$ and $\kappa_k^{\rel A, \rel B}(\rel X)$ are homomorphically equivalent.
\end{lemma}

\begin{proof}
  A homomorphism $\kappa_k^{\rel A, \rel B}(\rel X) \to \pi_\rel B\psi^{k, \rel A}(\rel X)$ is proved in Corollary~\ref{cor:soundness-of-canonical}. A homomorphism $\psi^{k, \rel A}(\rel X) \to \pi_\rel P\kappa_k^{\rel A}(\rel X)$ is constructed during the proof of Lemma~\ref{lem:universality-of-kcons}: observe that the Datalog interpretation $\psi^{k, \rel A}$ has all the properties of $\phi'$ that have been used in the proof of the lemma, and hence the same proof works here as well. Consequently, we get a homomorphism $\pi_\rel B \psi^{k, \rel A}(\rel X) \to \kappa_k^{\rel A, \rel B}(\rel X)$ by Lemma~\ref{lem:universal-gadget-is-monotone}.
\end{proof}

\subsubsection{Putting everything together}

We conclude this appendix with a formal proof of the theorem which we recall for convenience.

\canonicalwidth*

\begin{proof}
  We start with the implication (1) $\to$ (2). We need to prove that $\kappa_k^{\rel A, \rel B}$ is sound and complete reduction between the two promise CSPs. Completeness is proved in Lemma~\ref{lem:completeness-of-consistency}, and the soundness follows from Corollary~\ref{cor:soundness-of-consistency} in the following way: Assume that $\kappa_k^{\rel A, \rel B}(\rel X) \to \rel B'$ then by the corollary, we have that $\gamma\phi(\rel X) \to \kappa_k^{\rel A, \rel B}(\rel X) \to \rel B'$, and hence $\rel X \to \rel A'$ since $\gamma\circ \phi$ is sound.

  The implication (2) $\to$ (1) is proved in a similar way. We let $\phi = \psi^{k, \rel A}$ and $\gamma = \pi_\rel B$. Completeness follows from Lemmas \ref{lem:consistency-is-datalog+} and \ref{lem:completeness-of-consistency}, and the soundness follows from Corollary~\ref{cor:soundness-of-canonical} by the same argument as above.
\end{proof}

Let us note that the proof does not assume that the templates involved are composed of finite structures; nevertheless the finiteness of $\rel A$ and $\rel B$ is useful to ensure that the reductions involved are computable in polynomial-time. No assumptions on finiteness of $\rel A'$ and $\rel B'$ are needed either way.

\section{More on arc-consistency reductions}
  \label{app:arc-consistency} \label{app:proofs-v}

In this appendix, we prove the characterisation of the arc-con\-sis\-ten\-cy reduction provided in Section~\ref{sec:arc-consistency}, and as a consequence provide a new sufficient condition for existence of a \Datalog reduction between two promise CSPs. Finally, we include an argument showing that the arc-consistency reductions compose.

\subsection{The characterisation of arc-consistency reductions}
  \label{app:ac-proof}

We prove Theorem~\ref{thm:unary}. The key step in the proof is that the construction $\omega$ is dual to the reduction $\ac$ (recall Definition~\ref{def:ac-on-lc}) in the sense of the following lemma.\footnote{The two constructions induce a pair of \emph{adjoint functors} between the posetal categories of relational structures of the corresponding signatures.}
To formally state the lemma, we treat minions as infinite structures in the signature of the label cover in a natural way, i.e., $\clo M$ is understood as the structure $\rel M$ with $M_X = \clo M^{(X)}$ for each finite set $X$, and $P_\pi^\rel M = \{ (f, f^\pi) \mid f\in \clo M^{(X)} \}$ for $\pi\colon X \to Y$. Note that minion homomorphisms coincide with homomorphisms of these structures.

\begin{lemma} \label{lem:arc-adjoint}
  Let $\rel X$ be a label cover instance and $\clo M$ a minion. Then $\ac(\rel X) \to \clo M$ if and only if $\rel X \to \omega(\clo M)$.
\end{lemma}

\begin{proof}
  We denote the type of $f \in \rel X$ by $N_f$ and the type of the same symbol $f \in \ac(\rel X)$ by $\mathcal F_f$.
  Let $h\colon \ac(\rel X) \to \clo M$ be a homomorphism. We define a homomorphism $g\colon \rel X \to \omega(\clo M)$ by
  \[
    g(f) = (\mathcal F_f, h(f)).
  \]
  It is straightforward to check that $g$ is indeed a homomorphism.

  For the other implication, assume that $g\colon \rel X \to \omega(\clo M)$ is a homomorphism. Let $(G_f, g_f) = g(f)$.
  Observe that, for each constraint $(f, f') \in E_\pi^{\rel X}$, we have $\pi(G_f) = G_{f'}$, and hence none of the elements of $G_f$'s are removed by enforcing arc-consistency. Consequently, $G_f \subseteq \mathcal F_f$, which allows us to define $h\colon \ac(\rel X) \to \clo M$ by
  \[
    h(f) = g_f^\iota
  \]
  where $\iota\colon G_f \incl \mathcal F_f$ is the inclusion mapping.
  Again, it is straightforward to check that $h$ is indeed a homomorphism.
\end{proof}

The rest of the proof relies on, and closely follows the proof of \cite[Theorem 3.12]{BBKO21}. For that reason, let us reformulate the problem of \emph{promise satisfaction of minor conditions} $\PMC_{\clo M}$ \cite[Definition 3.10]{BBKO21} in the language of promise label cover.\footnote{We note for the reader who is familiar with the notion of a minor condition, that there is a direct translation between label cover instances and minor conditions using the following dictionary: \emph{function symbol} $f$ of arity $D$ --- an \emph{element} $f$ of type $D$; an \emph{identity} $f \equals g^\pi$ --- a \emph{constraint} $(f, g) \in E_\pi$.}
This problem $\PMC_\clo M$ can be expressed as a promise CSP whose template consists of the minion of projections $\clo P$ and the minion $\clo M$ where both are treated as structures.
The minion of projections $\clo P$ coincides with the template of label cover $\rel P$ (see Definition~\ref{def:label-cover}). As an abstract minion is it defined as the identity functor, i.e., $\clo P^{(X)} = X$, and $\pi^\clo P(x) = \pi(x)$ for each $\pi\colon X\to Y$ and $x\in X$.

\begin{definition}
  The \emph{promise label cover problem} assigned to a minion $\clo M$, denoted simply as $\PCSP(\clo P, \clo M)$, is the problem whose input is a label cover instance $\rel S$, and the goal is to decide between the cases $\rel S \to \clo P$ and $\rel S \not\to \clo M$.
\end{definition}

Though the templates in the problem $\PCSP(\clo P, \clo M)$ are infinite, the instances are finite. Often it is enough to work with a finite reduct of the signature of label cover, and hence the structures of the template will become finite assuming that the minion $\clo M$ is \emph{locally finite}, i.e., that $\clo M^{(X)}$ is finite for all finite sets $X$ (which is true if $\clo M$ is a polymorphism minion of a finite promise template).

In this language, a core theorem of the algebraic approach \cite[Theorem 3.12]{BBKO21} states that $\PCSP(\rel A, \rel B)$ and $\PCSP(\clo P, \clo M)$ are interreducible by gadget reductions for any promise template $\rel A, \rel B$ whose polymorphism minion is $\clo M$. More precisely, the theorem states that if $\clo M = \Pol(\rel A, \rel B)$, then:
\begin{enumerate}
  \item $\PCSP(\rel A, \rel B)$ reduces to $\PCSP(\clo P, \clo M)$ by a gadget reduction;
  \item for every finite reduct $\Sigma$ of the label cover signature, there exists a gadget reduction from $\PCSP(\clo P, \clo M)$, with the inputs restricted to $\Sigma$-structures, to $\PCSP(\rel A, \rel B)$.
\end{enumerate}

We note that this statement could be proven as a corollary of Theorem~\ref{thm:gadget-characterisation} since every minion $\clo M$ is isomorphic to $\Pol(\clo P, \clo M)$.\footnote{This isomorphism is known in category theory as the Yoneda Lemma.} Let us nevertheless warn the reader that this `alternative' proof would be cyclic since Barto, Bulín, Krokhin, and Opršal \cite{BBKO21} prove Theorem~\ref{thm:gadget-characterisation} \cite[Theorem 3.1]{BBKO21} using \cite[Theorem 3.12]{BBKO21}.

Let us recall the two reductions used to prove the statement above. The reduction in item (1) is the universal gadget $\pi_\rel A$, the correctness of this reduction is proven using the following lemma.

\begin{lemma}[{\cite[Lemma 3.16 \& Remark 3.17]{BBKO21}}] \label{lem:3.16}
  Let $\rel A, \rel B$ be a promise template. For all label cover instances $\rel X$, $\pi_{\rel A}(\rel X) \to \rel B$ if and only if $\rel X \to \Pol(\rel A, \rel B)$.
\end{lemma}

Even though the above lemma is phrased in \cite{BBKO21} with the assumption that the structures $\rel A$, $\rel B$, and $\rel X$ are finite, it holds for infinite structures as well (although the reduction $\pi_{\rel A}$ ceases to be polynomial-time computable). In particular, we have that for each minion $\clo M$, $\pi_{\rel A}(\clo M) \to \rel B$ if and only if $\clo M \to \Pol(\rel A, \rel B)$.
We note without a proof that the structure $\pi_{\rel A}(\clo M)$ is in fact isomorphic to \emph{free structure} of $\clo M$ generated by $\rel A$, usually denoted by $\Free_\clo M(\rel A)$ \cite[Definition 4.1]{BBKO21} (see also \cite[Lemma 4.4]{BBKO21}).

The other reduction, from $\PCSP(\rel A, \rel B)$ to $\PCSP(\clo P, \clo M)$, is the reduction $\rho^{\rel A}$ which we introduced in Definition~\ref{def:reification-to-lc}. The following lemma follows from \cite[Lemmas 3.14 \& 4.3, and Theorem 4.12]{BBKO21}.

\begin{lemma}[\cite{BBKO21}] \label{lem:4.3}
  Let $\clo M$ be a minion, and $\rel A$, $\rel X$ be two structures of the same signature. Then $\rho^{\rel A}(\rel X) \to \clo M$ if and only if $\rel X \to \pi_{\rel A}(\clo M)$.
\end{lemma}

We note that Theorem~\ref{thm:gadget-characterisation} can be proven directly using just Lemmas~\ref{lem:3.16} and \ref{lem:4.3} and  Krokhin, Opršal, Wrochna, and Živný \cite[Theorem~4.11]{KOWZ23}. We exploit this proof, and adapt it for characterisation of the arc-consistency reduction.

We may now proceed to prove the theorem.

\begin{proof}[Proof of Theorem~\ref{thm:unary}]
  Let $\clo A = \Pol(\rel A, \rel A')$, $\clo B = \Pol(\rel B, \rel B')$.
  First, we show that if there is a minion homomorphism $\xi\colon \omega(\clo B) \to \clo A$ then $\PCSP(\rel A, \rel A')$ reduces to $\PCSP(\rel B, \rel B')$ via an arc-consistency reduction.
  The completeness of the arc-consistency reduction is straightforward and analogous to the proof of the completeness of the $k$-consistency reduction (see Lemma~\ref{lem:completeness-of-consistency}).
  We focus on the soundness. Assume that $\pi_{\rel B}\ac\sigma^{\rel A}(\rel X) \to \rel B'$ with the aim to show that $\rel X \to \rel A'$.
  First, we have that $\ac\sigma^{\rel A}(\rel X) \to \clo B$ by Lemma~\ref{lem:3.16}. Further, $\sigma^{\rel A}(\rel X) \to \omega(\clo B)$ by Lemma~\ref{lem:arc-adjoint}, and consequently $\sigma^{\rel A}(\rel X) \to \clo A$. Finally, the latter implies that $\rel X \to \rel A'$ \cite[Lemma~3.14]{BBKO21} (or Lemmas~\ref{lem:3.16} and \ref{lem:4.3}).

  To show the other implication, assume that $\PCSP(\rel A, \rel A') \leqarc \PCSP(\rel B, \rel B')$, and consider the structure $\pi_{\rel A}\omega(\clo B)$.\footnote{This structure is the \emph{free structure} of $\omega(\clo B)$ generated by $\rel A$.}
  To obtain a minion homomorphism from $\omega(\clo B)$ to $\clo A$, it is enough to show that $\pi_\rel A\omega(\clo B) \to \rel A'$ (Lemma~\ref{lem:3.16}). We prove this by using the soundness of arc-consistency.
  Since $\pi_\rel A\omega(\clo B)$ may be infinite, we may not invoke soundness directly. This can be dealt with by the standard compactness argument (see, e.g., \cite[Remark 7.13]{BBKO21}) --- to show that a $\rel F$ maps homomorphically to a finite structure $\rel A'$, it is enough to show that every finite substructure of $\pi_\rel A\omega(\clo B)$ does. For the rest of the proof, we assume that $\rel F$ is an arbitrary finite substructure of $\pi_\rel A\omega(\clo B)$.
  We have that: $\sigma_\rel A(\rel F) \to \omega(\clo B)$ (Lemma~\ref{lem:4.3}), $\ac\sigma_\rel A(\rel F) \to \clo B$ (Lemma~\ref{lem:arc-adjoint}), and $\pi_\rel B\ac\sigma_\rel A(\rel F) \to \rel B'$ (Lemma~\ref{lem:3.16}). Since the last homomorphism witnesses that $\ac^{\rel A, \rel B}(\rel F)$ is not a negative instance of $\PCSP(\rel B, \rel B')$, we get that $\rel F \to \rel A'$ by the soundness of the arc-consistency reduction as we wanted to show.
\end{proof}

We briefly note that the above proof can be adapted for the cases when $\rel A$, $\rel A'$, $\rel B$, and $\rel B'$ are not necessarily finite: The reverse implication (showing that arc-consistency is a valid reduction) holds unconditionally, though it might not always give an efficient reduction. The other implication (showing necessity of a minion homomorphism) uses only finiteness of $\rel A'$. Alternatively, it can be proven under the assumption that $\pi_\rel A\omega(\clo B)$ is finite, e.g., if $\rel A$ is finite and $\clo B$ is locally finite.

\subsection{A sufficient condition for \Datalog reductions}
  \label{app:sufficient-condition}

We combine Theorem~\ref{thm:unary} with results of \citet{BK22} to provide a new sufficient condition for the existence of a \Datalog reduction. Namely, we can connect the construction $\omega$ with the notion of \emph{minion $(d, r)$-homomorphisms} \cite[Definition 5.1]{BK22}.

\begin{definition}
  \label{def:dr-homomorphism}
  A \emph{chain of minors} in a minion $\clo M$ of length $r$ is a sequence $t_0, \dots, t_r \in \clo M$ together with functions $\pi_{i, j}$ for $0 \le i < j \le r$ such that $t_j = t_i^{\pi_{i, j}}$ and $\pi_{i, k} = \pi_{j, k} \circ \pi_{i, j}$ for all $i < j < k$.

  Let $\clo M$ and $\clo N$ be two minions. A \emph{$(d, r)$-homomorphism} from $\clo M$ to $\clo N$ is a mapping $\xi$ that assigns to each element of $\clo M$ a non-empty subset of $\clo N$ of size at most $d$ such that
  \begin{enumerate}
    \item for all $f\in \clo M$, every $g\in \xi(f)$ has the same arity as $f$.
    \item for each chain of minors in $\clo M$ as above, there exist $i < j$ and $g_i \in \xi(t_i)$, $g_j \in \xi(t_j)$ such that $g_j = g_i^{\pi_{i, j}}$.
  \end{enumerate}
\end{definition}

\begin{corollary} \label{cor:sufficient-condition}
  Assume $\rel A, \rel A'$ and $\rel B, \rel B'$ are two promise templates. If there exist $d, r$ and a minion $(d, r)$-homomorphism from $\omega(\Pol(\rel B, \rel B'))$ to $\Pol(\rel A, \rel A')$, then $\PCSP(\rel A, \rel A') \leqdl \PCSP(\rel B, \rel B')$.
\end{corollary}

The proof is a combination of the proof of \cite[Theorem 5.1]{BK22} and the proof of Theorem~\ref{thm:unary}. Loosely speaking, we have the following chain of reductions:
\[
  \PCSP(\rel A, \rel A')
  \leqdl \PCSP(\clo P, \omega(\clo B))
  \leqarc \PCSP(\rel B, \rel B')
\]
where the \Datalog reduction is given by \cite[Theorem 5.1]{BK22}. To make this proof formal, we need to phrase \cite[Theorem 5.1]{BK22} in the language of Datalog reductions. This could be done by a careful inspection of the proof.
Nevertheless, we present a complete argument building on the proof of the said theorem assuming only the Combinatorial Gap Theorem \cite[Theorem~2.1]{BK22}. The core of the argument remains the same, though the argument is slightly simpler than in \cite{BK22} due to the fact that the typed setting is more natural here, as noted in \cite[Appendix B.4]{BK22}.

To formulate the Combinatorial Gap Theorem we need the following technical definition.

\begin{definition}
  Let $\rel X$ be an instance of $\CSP(\rel A)$. We say that a sequence $(\mathcal I_0, \dots, \mathcal I_r)$ is a~\emph{consistent sequence of partial solution schemes} of arities $k_0 \geq \dots \geq k_r$ and a value $d > 0$, if
  \begin{itemize}
    \item $\mathcal I_i$ assigns to each $U \in \binom X{\leq k_i}$ a set of at most $d$ partial solutions on $U$, and
    \item for every $U_0 \supseteq \dots \supseteq U_K$ with $U_i \in \binom X{\leq k_i}$, there exists $i < j$ such that
    \[
      \mathcal I_j(U_j) \cap (\mathcal I_i(U_i) \circ \iota_{i, j}) \neq \emptyset,
    \]
    where $\iota_{i, j} \colon U_j \incl U_i$ is the inclusion mapping.
  \end{itemize}
\end{definition}

\begin{theorem}[Combinatorial Gap Theorem \cite{BK22}]
  For all positive integers $m$, any (possibly infinite) structure $\rel D$ whose relations are of arity not more than $m$, and $r, d \geq 1$, there exists a sequence $k_0\geq \dots \geq k_r$ such that every instance that allows a consistent sequence of partial solution schemes of arities $k_0, \dots, k_r$ and value $d$ is solvable.
\end{theorem}

Using the above theorem, we can prove the corollary.

\begin{proof}[Proof of Corollary~\ref{cor:sufficient-condition}]
  Let $k_0, \dots, k_r$ be as in the Combinatorial Gap Theorem for parameters $d$ and $r$ and template $\rel A'$. We claim that $\PCSP(\rel A, \rel A') \leqcons \PCSP(\rel B, \rel B')$ for $k = k_0$. As usual, it is enough to prove soundness of the $k$-consistency reduction. Assume $\kappa_k^{\rel A, \rel B}(\rel X) \to \rel B'$. Hence $\pi_{\rel B}\ac\sigma_k^\rel A(\rel X) \to \rel B'$ which implies that $\ac\sigma_k^\rel A(\rel X) \to \Pol(\rel B, \rel B')$ (Lemma~\ref{lem:3.16}). Consequently, $\sigma_k^\rel A(\rel X) \to \omega(\Pol(\rel B, \rel B'))$ by Lemma~\ref{lem:arc-adjoint}.
  The rest of the proof is a rephrasing of the proof of \citet[Theorem 5.1]{BK22}.

  We claim that the instance $\rel X$ of $\CSP(\rel A')$ allows a consistence sequence of partial assignment schemes of arities $k_0 \geq \dots \geq k_r$.
  We let
  \[
    \mathcal I_i(K) =
      \{ g(\mathcal F_K) \mid g\in \xi(f_K) \}
  \]
  where $g(\mathcal F_K)\colon K \to \rel A'$ denotes the component-wise application of $g$ to all the elements of $\mathcal F_K$. More precisely, we let $g(\mathcal F_K)\colon x \mapsto g(e_x)$ for each $x\in K$ where $e_x\colon \mathcal F_K \to A'$ is the evaluation at $x$, i.e., $e_x(f) = f(x)$.
  Since each $f \in \mathcal F_K$ is a partial homomorphism from $K$ to $\rel A$ and $g$ is a polymorphism, $g(\mathcal F_K)$ is a partial homomorphism from $K$ to $\rel A'$.
  We get the consistency relation from the fact that $\xi$ is a $(d,r)$-homomorphism: if $K_0 \supseteq \dots \supseteq K_r$, and $g_i \in \xi(f_{K_i})$, then for some $i < j$, $g_i^{\pi_{i,j}} = g_j$ where $\pi_{i,j} \colon \mathcal F_{K_i} \to \mathcal F_{K_j}$ is defined by $\pi_{i,j}(f) = f|_{K_j}$, and hence
  \begin{multline*}
    g_j(\mathcal F_{K_j})(x) = g_j(e_x) = g_i^{\pi_{i,j}}(e_x) \\
      = g_i (e_x \circ \pi_{i,j}) = g_i(e_x) = g(\mathcal F_{K_j})(x)
  \end{multline*}
  for all $x\in K_j$. Above, we used $e_x$ to denote both a function from $\mathcal F_{K_i}$ as well as from $\mathcal F_{K_j}$; we have $e_x(f) = e_x(f|_{K_j})$ since $x\in K_j \subseteq K_i$. This concludes checking that the assumptions of the Combinatorial Gap Theorem are satisfied, and hence $\rel X \to \rel A'$, as we wanted.
\end{proof}

We note that the condition given in Corollary~\ref{cor:sufficient-condition} is sufficient but not necessary as is shown in the following example.

\begin{example}
  We claim that there is no $(d, r)$-homomorphism from $\clo H$ to $\Pol(\rel K_2)$ for any $d, r \geq 1$ while $\CSP(\rel K_2)$ reduces to the trivial problem by the $3$-consistency reduction (as shown by Example~\ref{ex:2-colouring-reduces}).

  The key observation is that $\Pol(\rel K_2)$ contains no $f$ that satisfies $f^\pi = f$ for a permutation $\pi$ of order two without fixed points (i.e., with only cycles of length 2), while $\clo H$ contains an $f$ that satisfies $f^\pi = f$ for all such permutations on any fixed set. We use this $f\in \clo H$ to build a long chain of minors, and consequently force that some of the $d$ elements in the image of some $(d, r)$-homomorphisms satisfy $c_1^\pi = c_2$. If we manage to force enough of these relations, we will obtain a $c \in \chi(f)$ with $c^\pi = c$ which will yield a contradiction.

  Fix $d$ and $r$, and let $n \geq (r+1)d$. We start with constructing a permutation group where all elements are of order two and have no fixed points. For that, we take the Cayley representation of $\mathbb Z_2^n$ ---  i.e., let $\mathbb Z_2^n$ act on its base set $N = \{0, 1\}^n$ by left multiplication.
  We denote this permutation group by $\mathbb G$.

  Now let $f\in \clo H^{(N)}$ be an element satisfying $f^\pi = f$ for all $\pi\in \mathbb G$, i.e., $f = N$. To construct a chain of minors we need to fix $\pi_{i,j} \in \mathbb G$ such that $\pi_{j,k}\circ \pi_{i,j} = \pi_{i,k}$. We pick $\epsilon_0, \dots, \epsilon_{k-1}$ to be elements in $\mathbb G$ corresponding to a vector-space basis of $\mathbb Z_2^k$, and note that the product of $\prod_{i\in I} \epsilon_i$ is non-trivial for all $I \neq \emptyset$. Now, we can let
  \[
    \pi_{i,j} = \prod_{\ell = i+1}^j \epsilon_\ell.
  \]
  Indeed, we have that $f^{\pi_{i,j}} = f$, for all $i < j$, hence $f$ itself creates a chain of minors of length $(r+1)d$.

  Assume that $\chi\colon \clo H \to \Pol(\rel K_2)$ is a $(d,r)$-homomorphism. We consider a labelled graph $E$ with vertices $\chi(f)$ defined as follows: two vertices $u, v$ are connected by an edge labelled by a non-trivial element $\pi \in \mathbb G$ if $u^\pi = v$ (note that since $\pi$ is of order two, this condition is equivalent to $u = v^\pi$, hence the edges are unoriented).
  We use the long chain of minors to force edges into this graph to obtain a cycle: For all $k < d$, there exist $i, j$ with $(r+1)k \leq i < j < (r+1)(k + 1)$ such that
  \[
    \chi(f)^{\pi_{i,j}} \cap \chi(f^{\pi_{i,j}}) \neq \emptyset.
  \]
  Since $f^{\pi_{i,j}} = f$, this forces an edge into $E$ labelled by $\pi_{i,j}$. Hence the graph has at least $d$ edges. Moreover observe that any product of a subset of the labels of these edges is non-trivial. Since we have at least $d$ edges and only $d$ vertices, there is a cycle in $E$. Let $c_0, \dots, c_k$ be such a cycle with labels $\lambda_0$, $\lambda_1$, etc. We claim that $c_0^\pi = c_0$ where $\pi$ is the product of all the labels of edges of the cycle. By induction, we get
  \[
    c_{i+1}^{\prod_{j \leq i} \lambda_j}
      = (c_{i+1}^{\lambda_i})^{\prod_{j < i} \lambda_j}
      = c_i^{\prod_{j < i} \lambda_j}
      = c_0
  \]
  where $c_{k+1} = c_0$.
  Since $\pi = \prod_{i \leq k} \lambda_i$ is non-trivial, we get that $c_0^\pi = c_0$ is in contradiction with $c_0$ being a polymorphism of $\rel K_2$.

  Finally, let us note to the reader who is familiar with the reduction used in \cite{BK22}, described as \emph{$k$-reduction} in \cite[Section 4.2]{KO22}, that $\CSP(\rel K_2)$ reduces to Horn-3SAT by 3-reduction, hence $(d,r)$-homomorphisms are not even necessary for this $k$-reduction.
\end{example}

\subsection{Arc-consistency reductions are transitive}
  \label{sec:kleisli}

We argue that arc-consistency reductions compose by using the language of category theory that is incredibly elegant for this purpose. Essentially, this claim follows from an observation that $\omega$ is a \emph{comonad}, and Theorem~\ref{thm:unary} which characterises the arc-consistency reduction in terms of \emph{co-Kleisli arrows} of this comonad. Let us briefly outline the definitions. Although it is not the traditional way, we define these notions together to highlight their connection.

\begin{definition}
  Assume $\eta$ is an endofunctor of some category, and $A, B$ are two objects. A \emph{co-Kleisli arrow} is a morphism\footnote{Usually, a morphism $\eta(A) \to B$ is called a co-Kleisli arrow only if $\eta$ is a comonad, and in that case the `co-' prefix is often dropped since it is implied.}
  \[
    f\colon \eta(A) \to B.
  \]
  A functor $\eta$ is a \emph{comonad} if its co-Kleisli arrows form a category with the same objects as the underlying category, i.e., if there is an associative binary operator $\circ_\eta$ that assigns to every pair of co-Kleisli arrows $f\colon \eta(A)\to B$ and $g\colon \eta(B) \to C$ a co-Kleisli arrow $g \mathbin{\circ_\eta} f \colon \eta(A) \to C$, and for each object $A$, there is a co-Kleisli arrow $\nu^A\colon \eta(A) \to A$ that acts as the identity w.r.t.~$\circ_\eta$.
\end{definition}

\begin{lemma}
  $\omega$ is a comonad.
\end{lemma}

\begin{proof}
  The unit $\nu^{\clo M}$ is the map $\nu\colon \omega(\clo M) \to \clo M$ we mentioned above, i.e., $\nu^\clo M_X(Y, f) = f^\iota$ where $\iota\colon Y \incl X$ is the inclusion mapping. To define the composition, it is easier to show how to get from a co-Kleisli arrow $\xi\colon \omega(\clo M) \to \clo N$ a minion homomorphism $\xi^\bind\colon \omega(\clo M) \to \omega(\clo N)$; this $\xi^\bind$ is defined by $\xi_Y^\bind(X, f) = (X, \xi_X(X, f))$ where $Y \supseteq X$ and $f\in \clo M^{(X)}$. The composition is then defined as
  \[
    \xi \circ_\omega \zeta = \xi \circ \zeta^\bind.
  \]
  Checking that $\circ_\omega$ is associative and that $\nu^\clo M$ is the identity element of $\circ_\omega$ is straightforward.
\end{proof}

Remark that, in the above proof, we showed that $\omega(\clo M) \to \omega(\clo N)$ if and only if $\omega(\clo M) \to \clo N$.
Finally, as an easy corollary of the above lemma and the fact that co-Kleisli arrows of a comonad compose, we get that arc-consistency reductions are transitive.

\begin{corollary}
  The relation $\leqarc$ is transitive on promise CSPs.
\end{corollary}

\section{Sherali-Adams is a \texorpdfstring{$k$}{k}-consistency reduction to linear programming}
  \label{app:proofs-iv}

In this appendix, we provide a detailed proof of the fact that Sherali-Adams hierarchy coincides with the $k$-consistency hierarchy of linear programming. The proof is a comparison of Sherali-Adams understood as a reduction to linear programming and the $k$-con\-sis\-ten\-cy reduction. It requires two key properties of linear programming, and its polymorphisms: firstly, that the linear system produced by the Sherali-Adams is as good as the universal gadget, and secondly, that the polymorphisms of linear programming allow a homomorphism $\Pol(\rel Q_\conv) \to \omega(\Pol(\rel Q_\conv))$.

Let us start with recalling a description of the polymorphism minion of linear programming, see also \cite[Section~7.2]{BBKO21}. The minion $\Pol(\rel Q_\conv)$ is isomorphic to the minion $\clo Q_\conv$ described as follows:
$\clo Q_\conv^{(X)}$ consists of all rational probability distributions $\lambda$ on $X$, i.e., $\lambda\colon X\to \mathbb Q$ with $\sum_{x\in X} \lambda(x) = 1$ and $\lambda(x) \geq 0$ for all $x\in X$, and for $\pi\colon X \to Y$, we let $\lambda^\pi$ be the probability distribution of $\pi(x)$ when $x$ is sampled according to $\lambda$, i.e., $\lambda^\pi(x) = \sum_{y\in \pi^{-1}(x)} \lambda(y)$.

\begin{lemma} \label{lem:pol-qconv}
  The minions $\Pol(\rel Q_\conv)$ and $\clo Q_\conv$ are isomorphic.
\end{lemma}

\begin{proof}
  The key observation is that every polymorphism of $\rel Q_\conv$ is of the form
  \[
    f(x_1, \dots, x_n) = \sum_{i=1}^n \lambda_i x_i
  \]
  for some $\lambda_1, \dots, \lambda_n \geq 0$ such that $\sum_{i=1}^n \lambda_i = 1$. Identifying such polymorphism with the probability distribution $\lambda\colon i \mapsto \lambda_i$ then yields the required isomorphism. It is not hard to check that minors are indeed taken accordingly.
\end{proof}

The above isomorphism is implicitly used in characterisation of (promise) CSPs solved by the \emph{basic linear programming relaxation} of CSPs \cite[Theorem 7.9]{BBKO21}. This basic linear programming relaxation of $\CSP(\rel A)$ can be described as a composition $\lambda_\conv \circ \rho^\rel A$, where $\lambda_\conv$ is a gadget that produces from label cover instances to linear programs (i.e., structures similar to $\rel Q_\conv$). It is the same construction as in step (S2) of the Sherali-Adams relaxation. For simplicity, we treat the output of $\lambda_\conv$ as a linear program rather than a structure with the same signature as $\rel Q_\conv$.

\begin{definition}
  Let $\rel S$ be a label cover instance, $\lambda_\conv(\rel S)$ is the following linear program with variables $x_{s, i}$ for each type $X$, $s\in S_X$, and $i\in X$:
  \begin{align}
    \sum_{i \in X} x_{s, i} &= 0
      &\text{$\forall s\in S_X$} \\
    \sum_{i \in \pi^{-1}(j)} x_{s, i} &= x_{t, j}
      &\text{$\forall \pi\colon X \to Y$, $(s, t) \in E_\pi^\rel S$, $j \in Y$.}
  \end{align}
\end{definition}

It is not hard to check that $\lambda_\conv$ can be described as a gadget replacement, and that the $k$-th level of the Sherali-Adams relaxation of an instance $\rel X$ of $\CSP(\rel A)$ can be described as $\lambda_\conv\sigma_k^{\rel A}(\rel X)$. Recall that $\sigma_k^{\rel A}(\rel X)$ denotes the label cover instance obtained from an instance $\rel X$ of $\CSP(\rel A)$ by performing only steps (C1) and (C4) in the $k$-consistency enforcement, i.e., the label cover instance constructed in step (S1) of the Sherali-Adams relaxation.

The important property of $\lambda_\conv$ is that it is `as good as the universal gadget', i.e., that we can freely exchange $\lambda_\conv$ and the universal gadget $\pi_{\rel Q_\conv}$. We informally call this fact the \emph{short code property} of linear programming.\footnote{The gadget $\lambda_\conv$ is a shorter alternative to the universal gadget which is the natural \emph{long code test} for linear programming, hence the name.}

\begin{lemma} \label{lem:lp-short-code}
  For every label cover instance $\rel X$, $\lambda_\conv(\rel X) \to \rel Q_\conv$ if and only if $\pi_{\rel Q_\conv}(\rel X) \to \rel Q_\conv$.
\end{lemma}

\begin{proof}
  Observe that $\lambda_\conv(\rel X) \to \rel Q_\conv$ if and only if $\rel X$ is satisfied in $\clo Q_\conv$: the assignment $\lambda_s(i) = x_{s, i}$ gives the required probability distributions given that $x_{s, i}$ is a solution to $\lambda_\conv(\rel X)$, and vice-versa.
  We also have that $\pi_{\rel Q_\conv}(\rel X) \to \rel Q_\conv$ if and only if $\rel X$ is satisfied in $\Pol(\rel Q_\conv)$ (Lemma~\ref{lem:3.16}). Hence the statement follows from the isomorphism of the two minions (Lemma~\ref{lem:pol-qconv}).
\end{proof}

We turn our attention to the minion $\clo Q_\conv$. The following lemma implies that a (promise) CSP is reducible to linear programming by the arc-consistency reduction if and only if it is solved by the basic linear programming relaxation, i.e., the step enforcing consistency in the reduction can be skipped without any harm.

\begin{lemma} \label{lem:qconv-is-coalg}
  There is a minion homomorphism
  \[
    \xi\colon \clo Q_\conv \to \omega(\clo Q_\conv).
  \]
\end{lemma}
\begin{proof}
  We define $\xi$ as
  \[
    \xi(\lambda) = (\supp(\lambda), \lambda|_{\supp(\lambda)})
  \]
  where $\supp(\lambda)$ is the set of all values $x\in X$ that have a non-zero probability in $\lambda$.

  Checking that $\xi$ is a minion homomorphism is straightforward:
  First, observe that $\supp(\lambda^\pi) = \pi(\supp(\lambda))$ since $y$ has a non-zero probability according to $\lambda^\pi$ if only if there is $x$ with a non-zero probability in $\lambda$ such that $\pi(x) = y$.
  Further we have that
  \[
    \lambda^\pi|_{\supp(\lambda^\pi)} = (\lambda|_{\supp(\lambda)})^{\pi|_{\supp(\lambda)}}
  \]
  since both are probability distributions of $\pi(x)$ when $x$ is sampled according to~$\lambda$.
\end{proof}

Let us now prove Theorem~\ref{thm:sa}.

\begin{proof}[Proof of Theorem~\ref{thm:sa}]
  We prove the theorem by the comparison of the two reductions.
  Sherali-Adams accepts if and only if
  \begin{equation}
    \lambda_\conv\sigma_k^\rel A(\rel X) \to \rel Q_\conv.
    \label{eq:sa}
  \end{equation}
  Compare this to when the $k$-consistency reduction to $\CSP(\rel Q_\conv)$ accepts, which is if and only if
  \begin{equation}
    \pi_{\rel Q_\conv}\ac\sigma_k^\rel A(\rel X) \to \rel Q_\conv.
    \label{eq:cons-lp}
  \end{equation}
  By Lemma~\ref{lem:lp-short-code}, the condition \eqref{eq:sa} is equivalent to $\sigma_k^\rel A(\rel X) \to \clo Q_\conv$ which is equivalent to $\ac\sigma_k^\rel A(\rel X) \to \clo Q_\conv$ by Lemmas~\ref{lem:qconv-is-coalg} and \ref{lem:arc-adjoint}. The latter is equivalent to $\lambda_\conv\ac\sigma_k^\rel A(\rel X) \to \rel Q_\conv$ which is equivalent to the condition \eqref{eq:cons-lp} by Lemma~\ref{lem:lp-short-code}.
\end{proof}

Intuitively, the proof claims that the two gadgets are equivalent, and that the step enforcing arc-consistency can be omitted since linear programming can emulate arc-consistency. The latter is a general property of promise CSPs whose polymorphism minion $\clo M$ admits a homomorphism $\clo M \to \omega(\clo M)$, and hence the theorem and its proof can be generalised for several other hierarchies including a hierarchy that interweaves with the Lasserre hierarchy.

\subsection{On tensor hierarchies}
  \label{app:tensors}

Finally, let us briefly comment on a comparison of our hierarchies to those of \citet{CZ23-hierarchies}. These hierarchies are parametrised by a minion rather than a promise CSP, and moreover the paper deals only with single-sorted structures.
Let us therefore for this section fix a template $\rel A, \rel A'$ whose signature contains only one type, and a minion $\clo M$. We restate the original definition of \citet{CZ23-hierarchies} first as stated (without giving all the definition for which we refer to the cited paper), and then using the language of this paper.

\begin{definition}[Ciardo-Živný $k$-th tensor test]
  Let $\clo M$ be a minion, $k \geq 1$. Assuming that $\rel A$ and $\rel X$ are $k$-enhanced structures in the same signature, the $k$-th tensor test of $\clo M$ accepts if $\rel X^{\tensor k} \to \rel F_\clo M(\rel A^{\tensor k})$, and rejects otherwise.
\end{definition}

Instead of going through the definitions of the terms used above, let us paraphrase the test in the language of this paper. First, we will use the structure $\pi_\rel B(\clo M)$ instead of the free structure $\rel F_\clo M(\rel B)$. As we noted before, these structures are isomorphic. Their homomorphic equivalence, which is enough for our use, follows directly from Lemma~\ref{lem:4.3} and \cite[Lemma 4.11]{BBKO21}.
Second, we describe the $k$-th tensor power together with the operation of $k$-enhancing a structure as a Datalog interpretation in the following definition.

\begin{definition}
  Fix $k \geq 1$, and assume that the signature of $\rel A$ is $\Sigma$. We define a new signature $\Sigma_k$ as follows: it has a single type, a relational symbol $R_k$ of arity $m^k$ for each $\Sigma$-symbol $R$ of arity $m$, and a new relational symbol $T$ of arity $k^k$. The domain is defined by the Datalog program $\tau^k_D$ with the rule
  \[
    D(x_1, \dots, x_k) \fro x_1 = x_1, \dots, x_k = x_k.
  \]
  The relations $R_k$ are defined by the Datalog program $\tau^k_{R_k}$ with the rule
  \[
    R_k(x_{e_m(1)}, \dots, x_{e_m(km^k)}) \fro R(x_1, \dots, x_m)
  \]
  where $e_m$ is a function $e_m \colon [k] \times [m]^k \to [m]$ defined by $e(i, f) = f(i)$ and we identify the set $[km^k]$ with $[k] \times [m]^k$ by ordering the latter lexicographically. And finally, $T$ is defined by a program $\tau^k_T$ with the rule
  \[
    T(x_{e_k(1)}, \dots, x_{e_k(k^{k+1})}) \fro x_1 = x_1, \dots, x_k = x_k
  \]
  where $e_k\colon [k^{k+1}] \to [k]$ is as above.
\end{definition}

Note that the width of $\tau^k$ is the maximum of $k$ and $m$ where $m$ ranges through arities of $\Sigma$-symbols. The $k$-th \emph{tensor test} is then defined as follows.

\begin{definition}[$k$-th tensor test]
  The $k$-th tensor test of a minion $\clo M$ accepts an instance $\rel X$ of $\CSP(\rel A)$ if $\tau^k(\rel X) \to \pi_{\tau^k(\rel A)}(\clo M)$, and rejects otherwise.
\end{definition}

Using Lemma~\ref{lem:4.3}, it is easy to observe that the test correctly solves $\PCSP(\rel A, \rel A')$ if and only if $\rho^{\tau^k(\rel A)}\tau^k$ is a valid reduction from $\PCSP(\rel A, \rel A')$ to $\PCSP(\clo P, \clo M)$.
As a simple corollary of Theorem~\ref{thm:canonical-width}, we get that this reduction is covered by the $k$-consistency reduction, assuming that $k \ge m$ where $m$ is the maximal arity of relations of $\rel A$.

\begin{corollary}
  Let $\rel A, \rel A'$ be a promise template, $k \geq m$ where $m$ is the maximal arity of relations in $\rel A$, and $\clo M$ be a minion.
  If $\PCSP(\rel A, \rel A')$ is correctly solved by the $k$-th tensor test of $\clo M$, then it reduces to $\PCSP(\clo P, \clo M)$ by the $k$-consistency reduction.
\end{corollary}

\begin{proof}
  The tensor test accepts if and only if
  $
    \rho^{\tau^k(\rel A)}\tau^k(\rel X) \to \clo M.
  $
  We also have that if $\rel X \to \rel A$, then $\tau^k(\rel X) \to \tau^k(\rel A)$ and
  $
    \rho^{\tau^k(\rel A)}\tau^k(\rel X) \to \clo P.
  $
  This means that if the test correctly solves the promise CSP, then $\rho^{\tau^k(\rel A)}\circ \tau^k$ is a valid reduction from $\PCSP(\rel A, \rel A')$ to $\PCSP(\clo P, \clo M)$.
  Since $\rho^{\tau^k(\rel A)}$ is expressible as a gadget replacement and $\tau^k$ is a Datalog interpretation of width $k$, we get the required by Theorem~\ref{thm:canonical-width}.
\end{proof}

The converse can be proved under the assumption that $\clo M \to \omega(\clo M)$ using the method of the previous subsection combined with a direct generalisation of an argument in \cite{CZ23-hierarchies} which shows that if a minion is \emph{conic} then the $k$-th tensor test is able to find local homomorphisms on subsets of sizes at most $k$.

\bibliographystyle{ACM-Reference-Format}

\end{document}